\documentclass[aps,onecolumn,notitlepage,superscriptaddress,showpacs,nofootinbib]{revtex4-2}

\usepackage{enumerate,appendix}
\usepackage{amsmath, amsthm, amssymb,commath}
\usepackage{color,calc,graphicx}
\usepackage[dvipsnames,svgnames,table,cmyk]{xcolor}
\usepackage[colorlinks]{hyperref}
\usepackage{optidef}
\hypersetup{
	colorlinks = true,
	urlcolor = {blue},
	citecolor = {blue},
	linkcolor= {blue}
}

\usepackage{latexsym}
\usepackage{bbm}

\usepackage[charter,cal=cmcal,sfscaled=false]{mathdesign}
\usepackage{booktabs}
\usepackage{multirow}

\usepackage{dcolumn}
\usepackage{mathrsfs}

\usepackage{tikz}
\usetikzlibrary{arrows.meta,positioning,calc,fit}
\usetikzlibrary{shapes.geometric}
\tikzset{
  >=Latex,
  party/.style={
    draw, rounded corners, thick,
    minimum width=18mm, minimum height=8mm,
    align=center, fill=gray!8
  },
  charlie/.style={
    draw, rounded corners, thick,
    minimum width=24mm, minimum height=8mm,
    align=center, fill=gray!15
  },
  opbox/.style={
    draw, rounded corners, thick,
    minimum width=18mm, minimum height=8mm,
    align=center, fill=blue!6
  },
  state/.style={
    draw, circle, thick,
    minimum size=7mm,
    inner sep=1pt,
    fill=green!8
  },
  mem/.style={
    draw, ellipse, thick,
    minimum width=16mm, minimum height=8mm,
    align=center, fill=orange!10
  },
  flow/.style={->, thick},
  lab/.style={font=\small},
  smalllab/.style={font=\scriptsize}
}

\def \be {\begin{equation}}
\def \ee {\end{equation}}

\newcommand{\Tr}{\mathrm{Tr}}

\newcommand{\ket}[1]{|#1\rangle}
\newcommand{\bra}[1]{\langle#1|}

\def \cX{{\cal X}}

\def \sofc2{{\cal S}({\mathbb C}^2)}

\def\>{\rangle}
\def\<{\langle}

\newtheorem{definition}{Definition}
\newtheorem{theorem}{Theorem}
\newtheorem{lemma}[theorem]{Lemma}
\newtheorem{proposition}[theorem]{Proposition}
\newtheorem{corollary}[theorem]{Corollary}
\newtheorem{remark}[theorem]{Remark}
\newtheorem{example}{Example}

\def\Label#1{\label{#1}\ [\ \text{#1}\ ]\ }
\def\Label{\label}

\begin{document}

\title{Causal-Order Identification of Memoryless Sequential Quantum Processes from Restricted Projective Data}

\author{Masahito Hayashi}\email{hmasahito@cuhk.edu.cn}
\affiliation{School of Data Science, The Chinese University of Hong Kong, Shenzhen, Longgang District, Shenzhen, 518172, China}
\affiliation{International Quantum Academy, Futian District, Shenzhen 518048, China}
\affiliation{Graduate School of Mathematics, Nagoya University, Nagoya, 464-8602, Japan}

\begin{abstract}
Identifying causal order from restricted projective data is generally
nontrivial.
When two quantum players interact only through an unobserved environment, the
available local measurement statistics are typically not tomographically
complete, so the underlying process cannot in general be reconstructed exactly
from the observed distribution.
As a result, causal direction can be statistically identifiable in some cases
but fundamentally indistinguishable in others.

In this work, we determine necessary and sufficient conditions for deciding
when an observed distribution is compatible with a memoryless sequential
quantum process in a fixed direction.
We show that directional conditional-independence structure and the positivity
criterion based on the pseudo-density matrix, as developed in recent work by
We show that directional conditional-independence structure and PDM-based positivity criteria are not sufficient by themselves.
We then specialize to the two-qubit Pauli setting, where the problem remains
non-tomographic but becomes explicitly tractable.
In this regime, we characterize when the two sequential directions are
statistically indistinguishable, and we show by example that positivity alone
does not exclude more general memoryful strategies, whereas the additional
algebraic consistency requirement does.
\end{abstract}

\maketitle

\begin{center}
\end{center}

\section{Introduction}
Causality has become a central topic across a wide range of disciplines, from
statistics and machine learning to quantum information science
\cite{Pearl2009,Peters2017,Chiribella2009,OCB,Araujo2015}.
In the quantum setting, recent work has shown that causal structure is not only
a conceptual issue but also an operational resource, with applications to
channel discrimination, communication, metrology, and thermodynamic tasks
\cite{BMQ,ebler2018enhanced,bavaresco2022unitary,zhao2020quantum,felce2020quantum}.
Among the basic questions in this area, one of the most fundamental is the
problem of determining temporal order: from observed data alone, can one decide
which subsystem influences which, and when is such a decision impossible?

This problem is relatively transparent when one can prepare a sufficiently rich
set of observations to reconstruct the underlying process matrix
\cite{OCB,Chiribella2009,Araujo2015,CharFreeIdentification}.
In that case, the task reduces to the characterization of process matrices
compatible with a given causal order.
In the earlier work \cite{CharFreeIdentification}, we treated in detail the
complementary regime in which the available observations are tomographically
complete and hence identify the process matrix itself.
In that full-tomography regime, the order-identification problem becomes
essentially a process-level one, and for memoryless sequential structures it is
shown that the problem is effectively reduced to the corresponding Markovian
property.
By contrast, when the available observations are not tomographically complete,
the process matrix cannot be reconstructed exactly, and the temporal-order
problem becomes substantially more subtle.
Instead of process-level reconstruction, one must ask a distribution-level
question: given only the observed statistics, can the underlying process belong
to a given causal class, and if not, why not?
Recent work by Liu, Qiu, Dahlsten, and Vedral used the PDM formalism to
extract causal information from restricted measurement data \cite{Liu3,Liu2,PhD}.
More recently, this measurement-based approach was extended to two systems at two times, more general channels, and coarse-grained projective measurements, and was demonstrated experimentally on a nuclear-magnetic-resonance platform \cite{LiuExperimental2026}. These results demonstrate that restricted projective measurements can provide experimentally accessible information about causal structure.
The present work addresses a complementary mathematical problem. Rather than testing whether the observed data are compatible with a specified causal structure, we derive necessary-and-sufficient distribution-level conditions for memoryless sequential realizability under restricted projective observations. In particular, we identify an algebraic consistency condition that is independent of reconstructed-operator positivity and ensures that a single underlying channel reproduces the complete observed distribution.

Related works have developed state-over-time, quantum Bayesian, and
operator-representation frameworks for spatiotemporal quantum correlations
\cite{FullwoodParzygnat2022StatesOverTime,
ParzygnatFullwood2023QuantumBayes,
FullwoodParzygnat2024OperatorRepresentation}.
These works provide useful formal languages for representing temporal
correlations and relating forward and reverse descriptions, but they do not by
themselves give the memoryless-sequential membership criterion.
Further compatibility-based approaches have studied whether observed
measurement statistics or bipartite states admit spatial, temporal, or causal
explanations
\cite{SongNarasimhacharRegulaElliottGu2024CausalClassification,
SongParzygnat2025CausalExplanation}.
These developments strongly suggest that restricted projective observations may
still carry substantial information about causal order.
At the same time, they also indicate a basic limitation: positivity of
reconstructed objects alone does not provide a complete
characterization of memoryless sequential dynamics.

The purpose of the present paper is to characterize memoryless sequential quantum processes under restricted projective observations. Since these observations do not determine the process matrix, our analysis is carried out at the level of the observed distribution. We first identify the directional Markovian constraints and encode the accessible temporal correlations by compatible pseudo-density matrices. 
Under projective state reduction, the observed distribution determines the relevant marginal state and a candidate Choi matrix. 
However, constructing these two objects from the observed distribution does not by itself establish memoryless sequentiality. 
Even when the candidate Choi matrix is positive, it remains necessary to verify that the channel represented by this matrix reproduces the observed conditional probabilities for every measurement setting.
An additional algebraic consistency requirement on the observed distribution is therefore needed to guarantee such a common-channel representation.

The contributions of the present paper are twofold.  
First, we formulate this algebraic consistency requirement as condition~(C1). We prove that, within the class of distributions satisfying directional Markovianity, condition~(C1) is equivalent to exact reproduction of the full observed distribution by the reconstructed state and candidate Choi matrix. Consequently, any other condition sufficient for such reproduction must imply condition~(C1), so that condition~(C1) is minimal with respect to logical implication. Adding positivity of the reconstructed Choi matrix then yields a necessary-and-sufficient distribution-level criterion for compatibility with a memoryless sequential quantum process in a prescribed direction.
Second, for two-qubit Pauli measurements, we characterize when a
distribution generated in one memoryless sequential direction also admits the
opposite direction, with explicit conditions for representative channel
families.  

The remainder of the paper is organized as follows.
Section~\ref{sec:problem-setting} introduces the process-matrix framework and
the strategy classes.  Section~\ref{S5} formulates the restricted-observation
problem in terms of distributions, directional Markovian constraints, and
compatible pseudo-density matrices.  Section~\ref{S5-2} specializes this
framework to projective state reduction and develops the reconstruction tools
used in the main result.  Section~\ref{SSAL} then proves the fixed-direction
membership criterion, Theorem~\ref{TH10}.

The next two sections apply this criterion to different versions of the
reverse-direction problem in the two-qubit Pauli setting.
Section~\ref{S7} assumes that a forward memoryless representation is already
known and characterizes order indistinguishability.  This assumption makes the
reverse conditions explicitly tractable for general qubit channels and for
representative channel families.  Section~\ref{S7-M} removes that assumption
and gives the reverse-direction membership criterion for an arbitrary observed
distribution.  Finally, Section~\ref{S7-C} compares the memoryless class with
broader memoryful classes and shows by separate examples that reconstructed
Choi positivity and condition~(C1) detect distinct obstructions.

\section{Problem setting and notation}\label{sec:problem-setting}
\subsection{General process-matrix framework and statistical task}
We consider two isolated quantum players, Alice (player~1) and Bob (player~2), who interact only through an external environment controlled by a third party, Charlie. 
The aim of this subsection is to formalize Charlie's strategy by introducing a process matrix $W_S$. 
In each experimental run, Alice receives a quantum input system $\mathcal H_{I,1}$ from Charlie and returns a quantum output system $\mathcal H_{O,1}$ to Charlie, while Bob analogously receives $\mathcal H_{I,2}$ and returns $\mathcal H_{O,2}$. 
We denote the corresponding dimensions by $d_{I,1}, d_{O,1}, d_{I,2}, d_{O,2}$. 
During the experiment Alice and Bob are not allowed to communicate; after all runs are completed, they may share their classical records to form empirical statistics. 
Our goal is to infer structural information about Charlie's ``strategy'' purely from the joint input--output statistics observed by Alice and Bob.

Alice's most general admissible operation in one run is a quantum instrument
$\{\Gamma_{1,z_1}\}_{z_1\in\mathcal Z_1}$,
where each $\Gamma_{1,z_1}$ is completely positive (CP) and trace-nonincreasing, and $\sum_{z_1}\Gamma_{1,z_1}$ is trace preserving (TP). 
The classical outcome $z_1$ is recorded by Alice. 
Bob's operation is similarly described by an instrument
$\{\Gamma_{2,z_2}\}_{z_2\in\mathcal Z_2}$ with classical outcome $z_2$. 
We write $C[\Gamma]$ for the unnormalized Choi matrix of a CP map $\Gamma$,
using Choi's convention \cite{Choi1975}:
\begin{equation}
\operatorname{Tr}_{I}\bigl[(\rho^{\mathsf T}\otimes I_{O})\,C[\Gamma]\bigr]
=
\Gamma(\rho),
\label{eq:choi-convention}
\end{equation}
where ${\mathsf T}$ denotes transpose in a fixed input basis, and the partial
trace is over the input space.  Equivalently,
\[
C[\Gamma]=\sum_{a,b}|a\rangle\langle b|\otimes
\Gamma(|a\rangle\langle b|).
\]
This is the convention used throughout the paper.  It differs from
Jamiołkowski's original convention \cite{Jamiolkowski1972} by a partial transpose on the input system;
accordingly, whenever a partial transpose is needed below, it is written
explicitly as $T_1$ or $T_2$.
For a linear map $\Gamma:{\cal B}({\cal H}_{I,i})\to{\cal B}({\cal H}_{O,i})$,
we denote by $\Gamma^\dagger:{\cal B}({\cal H}_{O,i})\to{\cal B}({\cal H}_{I,i})$
its adjoint supermap with respect to the Hilbert--Schmidt inner product, i.e.,
\begin{equation}
\Tr\!\bigl[A\,\Gamma(B)\bigr]
=
\Tr\!\bigl[\Gamma^\dagger(A)\,B\bigr]
\label{eq:adjoint-supermap}
\end{equation}
for all Hermitian operators $A\in{\cal B}({\cal H}_{O,i})$ and
$B\in{\cal B}({\cal H}_{I,i})$. 
In particular, $\Gamma^\dagger(I)$ represents the corresponding effect operator on the input space.

We now formalize Charlie's strategy by introducing a process matrix $W_S$. 
Charlie does not access Alice's and Bob's laboratories directly; instead, he mediates their interaction by preparing inputs and processing outputs across the global interface
$\mathcal H_{I,1}\otimes\mathcal H_{O,1}\otimes\mathcal H_{I,2}\otimes\mathcal H_{O,2}$~\cite{OCB,Araujo2015,Chiribella2009}. 
We model Charlie's strategy $S$ by a bipartite process matrix $W_S\ge 0$ acting on this interface. 
When Alice and Bob perform their local instruments
$\{\Gamma_{1,z_1}\}_{z_1\in\mathcal Z_1}$
and $\{\Gamma_{2,z_2}\}_{z_2\in\mathcal Z_2}$, respectively, the joint distribution of their outcomes $z_1$ and $z_2$ is given by the generalized Born rule
\begin{equation}
P_{Z_1,Z_2}(z_1,z_2)
=
\mathrm{Tr}\Bigl[
W_S\bigl(C[\Gamma_{1,z_1}]\otimes C[\Gamma_{2,z_2}]\bigr)
\Bigr].
\label{eq:GBR}
\end{equation}
Thus the observed joint distribution is determined by Charlie's process matrix $W_S$. 
Valid process matrices satisfy the standard linear ``no-signalling-in-time'' constraints ensuring that \eqref{eq:GBR} defines a normalized probability distribution for all TP choices $\sum_{z_i}\Gamma_{i,z_i}$~\cite{OCB,Araujo2015}. 
In what follows, however, we only use \eqref{eq:GBR} together with basic positivity and normalization properties, and we do not need the full structural characterization of admissible process matrices.

Fixing the local operational capabilities, namely the families of instruments accessible to Alice and Bob, Charlie induces a statistical model
\begin{equation}
\mathcal M:\quad
S\longmapsto
\mathcal M[S]
:=
\bigl\{
P_{Z_1,Z_2}(z_1,z_2)
\bigr\}_{(z_1,z_2)\in\mathcal Z_1\times\mathcal Z_2}.
\label{eq:model-map}
\end{equation}

Our primary question is:\smallskip
\begin{quote}
\textbf{(Q)}\;\;\emph{Given only the observed distribution $\mathcal M[S]$, what aspects of Charlie's strategy---in particular, its causal direction and the presence/absence of memory---can Alice and Bob determine, and when is the causal direction fundamentally \emph{indistinguishable}?}
\end{quote}

When Alice and Bob have sufficiently rich local operations, for example tomographically complete sets of instruments as in \cite{CharFreeIdentification}, the map $\mathcal M$ can become injective on broad strategy classes, so that reconstruction-based classification is possible. 
The present paper addresses the complementary regime in which the available operations are genuinely restricted. 
We focus on observation schemes based on state reduction, with particular emphasis on projective state reduction. 
In such settings, exact process reconstruction is generally impossible, and the problem becomes one of nontrivial set-membership at the distribution level. 
Our main results therefore provide complete characterizations for memoryless sequential strategies and for the order-indistinguishable region in the minimal nontrivial two-qubit setting.

\subsection{Definite-order strategy classes (process-matrix definitions)}\label{sec:strategy-classes}
In addition to the general process-matrix description in Section~\ref{sec:problem-setting},
we will use several \emph{definite-order} strategy classes that capture parallel vs.~sequential wiring
and the presence/absence of memory.
We adopt the definitions in \cite{CharFreeIdentification} and record them here for completeness.
Throughout this subsection, $W_S$ denotes a valid bipartite process matrix on
${\cal H}_{I,1}\otimes{\cal H}_{O,1}\otimes{\cal H}_{I,2}\otimes{\cal H}_{O,2}$.

\medskip\noindent\textbf{(i) General class.}
${\cal S}_G$ denotes the set of all physically valid bipartite process matrices satisfying the standard no-signalling-in-time constraints
as in \cite{CharFreeIdentification}.
In general, this class may also contain processes that are not of definite order.

\medskip\noindent\textbf{(ii) Individual (product) strategy.}
$S\in{\cal S}_I$ if there exist states $\rho_1$ on ${\cal H}_{I,1}$ and $\rho_2$ on ${\cal H}_{I,2}$ such that
\begin{equation}
W_S=\rho_1\otimes\rho_2\otimes I_{O,1}\otimes I_{O,2}.
\end{equation}
Operationally, Charlie's inputs are independent and do not depend on either output (See Fig. \ref{fig:individual-strategy}).

\medskip\noindent\textbf{(iii) Parallel strategy (quantum memory allowed).}
$S\in{\cal S}_P$ if there exists a (possibly entangled) state $\rho_{12}$ on ${\cal H}_{I,1}\otimes{\cal H}_{I,2}$ such that
\begin{equation}
W_S=\rho_{12}\otimes I_{O,1}\otimes I_{O,2}.
\end{equation}
This class is called \emph{quantum parallel} in \cite{CharFreeIdentification}. We keep the legacy notation ${\cal S}_P$ used in the remainder of this manuscript (See Fig. \ref{fig:parallel-strategy}).

\medskip\noindent\textbf{(iv) Sequential (memoryless) strategy.}
$S\in{\cal S}_{N,1\to 2}$ if there exist an input state $\rho_1$ on ${\cal H}_{I,1}$ and a CPTP map
$\Lambda_{1\to2}:{\cal B}({\cal H}_{O,1})\to{\cal B}({\cal H}_{I,2})$ such that
\begin{equation}
W_S=\rho_1\otimes C[\Lambda_{1\to2}]\otimes I_{O,2},
\end{equation}
where $C[\Lambda]$ is the (unnormalized) Choi matrix in the convention of \eqref{eq:choi-convention}(See Fig. \ref{fig:sequential-memoryless}).
The opposite direction class ${\cal S}_{N,2\to 1}$ is defined analogously. 

\medskip\noindent\textbf{(v) Sequential strategy with quantum memory.}
$S\in{\cal S}_{Q,1\to 2}$ if Charlie may keep a quantum memory that mediates the map from Alice's output to Bob's input;
as presented in \cite[Eq.~(9)]{CharFreeIdentification},
an equivalent process-matrix characterization satisfies 
\begin{align}
(\Tr_{(O,2)}W_{S})\otimes \rho_{mix,(O,2)}
&= W_{S}\Label{BB1},\\
(\Tr_{(O,1),(I,2),(O,2)}W_{S})\otimes I_{(O,1)}
&= \Tr_{(I,2),(O,2)}W_{S}.  \Label{BB2}
\end{align}
See Fig. \ref{fig:sequential-quantum-memory}.
The opposite direction class ${\cal S}_{Q,2\to 1}$ is defined analogously. 
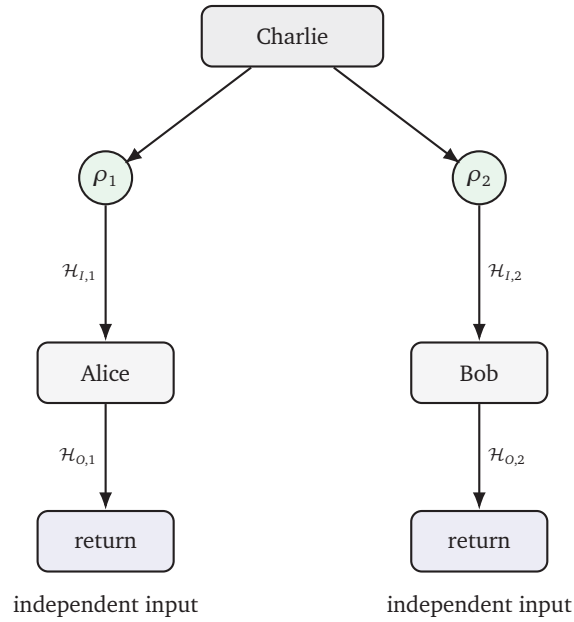
\begin{figure}[t]
\centering
\begin{tikzpicture}[node distance=12mm and 18mm]
\node[charlie] (C) {Charlie};
\node[state, below left=12mm and 10mm of C] (r1) {$\rho_1$};
\node[state, below right=12mm and 10mm of C] (r2) {$\rho_2$};
\node[party, below=18mm of r1] (A) {Alice};
\node[party, below=18mm of r2] (B) {Bob};
\node[opbox, below=14mm of A] (outA) {return};
\node[opbox, below=14mm of B] (outB) {return};
\draw[flow] (C) -- ($(C)!0.9!(r1)$);
\draw[flow] (C) -- ($(C)!0.9!(r2)$);
\draw[flow] (r1) -- node[left,smalllab] {$\mathcal H_{I,1}$} (A);
\draw[flow] (r2) -- node[right,smalllab] {$\mathcal H_{I,2}$} (B);
\draw[flow] (A) -- node[left,smalllab] {$\mathcal H_{O,1}$} (outA);
\draw[flow] (B) -- node[right,smalllab] {$\mathcal H_{O,2}$} (outB);
\node[lab, below=2mm of outA] {independent input};
\node[lab, below=2mm of outB] {independent input};

\end{tikzpicture}
\caption{
Individual (product) strategy:
Charlie prepares two independent input states $\rho_1$ and $\rho_2$ and sends them separately to Alice and Bob.
No correlation or sequential mediation is introduced between the two players.
}
\label{fig:individual-strategy}
\end{figure}

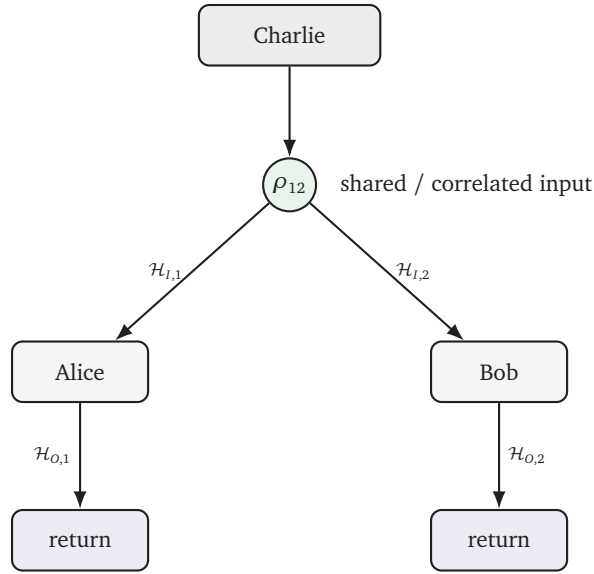
\begin{figure}[t]
\centering
\begin{tikzpicture}[node distance=12mm and 18mm]

\node[charlie] (C) {Charlie};

\node[state, below=12mm of C] (r12) {$\rho_{12}$};

\node[party, below left=18mm and 16mm of r12] (A) {Alice};
\node[party, below right=18mm and 16mm of r12] (B) {Bob};

\node[opbox, below=14mm of A] (outA) {return};
\node[opbox, below=14mm of B] (outB) {return};

\draw[flow] (C) -- (r12);
\draw[flow] (r12) -- node[left,smalllab] {$\mathcal H_{I,1}$} (A);
\draw[flow] (r12) -- node[right,smalllab] {$\mathcal H_{I,2}$} (B);
\draw[flow] (A) -- node[left,smalllab] {$\mathcal H_{O,1}$} (outA);
\draw[flow] (B) -- node[right,smalllab] {$\mathcal H_{O,2}$} (outB);

\node[lab, right=2mm of r12] {shared / correlated input};

\end{tikzpicture}
\caption{
Parallel strategy:
Charlie prepares a bipartite input state $\rho_{12}$ (possibly entangled) and distributes its two subsystems to Alice and Bob in parallel.
The correlation comes only from the shared input state, not from a sequential influence between the parties.
}
\label{fig:parallel-strategy}
\end{figure}

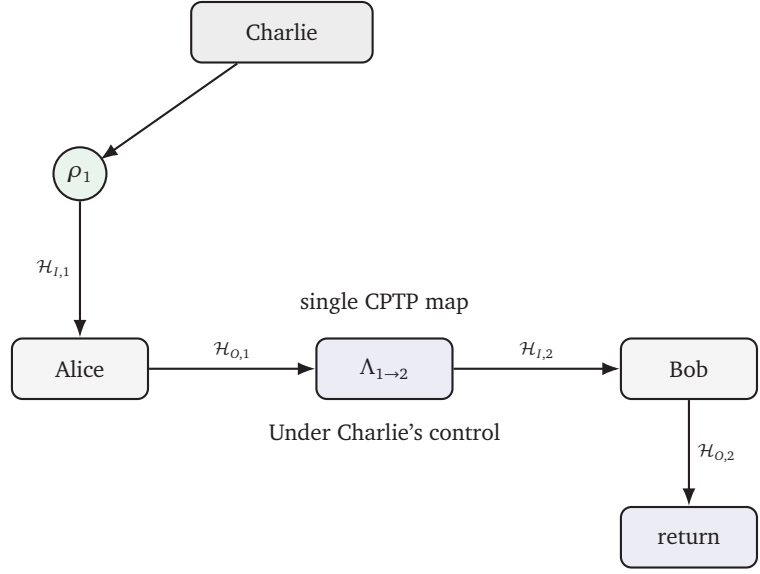
\begin{figure}[t]
\centering
\begin{tikzpicture}[node distance=12mm and 16mm]

\node[charlie] (C) {Charlie};

\node[state, below left=12mm and 12mm of C] (r1) {$\rho_1$};

\node[party, below=18mm of r1] (A) {Alice};

\node[opbox, right=22mm of A] (Lam) {$\Lambda_{1\to2}$};

\node[party, right=22mm of Lam] (B) {Bob};

\node[opbox, below=14mm of B] (outB) {return};

\draw[flow] (C) -- ($(C)!0.9!(r1)$);
\draw[flow] (r1) -- node[left,smalllab] {$\mathcal H_{I,1}$} (A);
\draw[flow] (A.east) -- node[above,smalllab] {$\mathcal H_{O,1}$} (Lam.west);
\draw[flow] (Lam.east) -- node[above,smalllab] {$\mathcal H_{I,2}$} (B.west);
\draw[flow] (B) -- node[right,smalllab] {$\mathcal H_{O,2}$} (outB);

\node[lab, above=2mm of Lam] {single CPTP map};
\node[lab, below=2mm of Lam] {Under Charlie's control};

\end{tikzpicture}
\caption{
Sequential memoryless strategy ($1\to2$):
Charlie first sends an input state $\rho_1$ to Alice, then transforms Alice's output into Bob's input through a single CPTP map $\Lambda_{1\to2}$.
The defining feature is that the mediation from Alice to Bob is described by one fixed channel.
}
\label{fig:sequential-memoryless}
\end{figure}

\begin{figure}[t]
\centering
\begin{tikzpicture}[node distance=18mm and 35mm]

\coordinate (Mid) at (0,0);

\node[state] (r12) at ([yshift=30mm]Mid) {$\rho_{12}$};

\node[charlie, above=15mm of r12] (C) {Charlie};

\node[party] (A) at ([xshift=-25mm]Mid) {Alice};
\node[opbox] (Proc) at ([xshift=25mm]Mid) {joint\\processor};

\node[party, right=35mm of Proc] (B) {Bob};

\node[opbox, below=20mm of B] (outB) {return};

\draw[flow] (C) -- (r12);

\draw[flow] (r12) -- node[left, smalllab, xshift=-4mm, yshift=2mm] {$\mathcal H_{I,1}$} (A.north);
\draw[flow] (r12) -- node[right, smalllab, xshift=4mm, yshift=2mm] {$\mathcal H_{M,1}$} (Proc.north);

\draw[flow] (A.east) -- node[above, smalllab] {$\mathcal H_{O,1}$} (Proc.west);

\draw[flow] (Proc.east) -- node[above, smalllab] {$\mathcal H_{I,2}$} (B.west);

\draw[flow] (B) -- node[right, smalllab] {$\mathcal H_{O,2}$} (outB);

\node[lab, above=3mm of Proc] {uses quantum memory};
\node[lab, below=2mm of Proc] {Under Charlie's control};

\end{tikzpicture}
\caption{
Sequential strategy with quantum memory in the direction $1\to2$.
Charlie prepares a joint state of Alice's input and an internal memory system,
then processes Alice's output together with the retained memory to generate
Bob's input.}
\label{fig:sequential-quantum-memory}
\end{figure}

The hierarchy of the above classes with 
the definite-order is summarized as
\begin{align}
\begin{array}{ccccc}
{\cal S}_{N,1\to 2} & \supset & {\cal S}_{I} & \subset & {\cal S}_{N,2\to 1} \\
\cap &  & \cap &  & \cap \\
{\cal S}_{Q,1\to 2} & \supset & {\cal S}_{P} & \subset & {\cal S}_{Q,2\to 1}
\end{array}
\Label{HI1}
\end{align}
Here, ${\cal S}_{I}$ is the smallest class in this hierarchy, ${\cal S}_{P}$ is the parallel class with arbitrary bipartite input state, and ${\cal S}_{N,1\to 2}$ and ${\cal S}_{Q,1\to 2}$ (and their reversed counterparts) describe sequential strategies without and with quantum memory, respectively. See \cite{CharFreeIdentification} for the corresponding process-matrix characterizations.

In addition to these classes, one may also consider subclasses of
${\cal S}_{Q,1\to 2}$, ${\cal S}_{Q,2\to 1}$, or ${\cal S}_{Q}$
in which Charlie's memory is restricted to be classical rather than quantum.
While a full distribution-level characterization of all such classes is
outside the scope of this paper, we later use the forward classical-memory
subclass to compare it with the memoryless and quantum-memory classes.

For example, in the class ${\cal S}_{Q}$,
the subclass in which Charlie's memory is classical corresponds simply to
separable states on ${\cal H}_{I,1} \otimes {\cal H}_{I,2}$.
The analysis of such classical-memory structures requires techniques
that are conceptually different from those employed in this work.

\section{General structure with state reduction}\Label{S5}
The final goal of this paper is a sharp characterization of memoryless sequential dynamics
under restricted \emph{projective} observations.
Before imposing that projective-state-reduction assumption, however,
it is useful to formulate the problem at the more general level of state reduction associated with randomized measurements.
This section develops that general framework.
Its role is to identify what can already be said at the distribution level
without yet restricting the post-measurement dynamics to the projective case,
and to prepare the structures that will later be specialized to projective state reduction.
In this way, the subsequent projective analysis can be understood as a refinement of the general state-reduction framework introduced here.

\subsection{Formulation with state reduction}
We consider the case when
the operations by Alice and Bob
are given as a randomized choice of
measurements.
In this case,
their information can be described by two random variables,
the random variable $Y\in {\cal Y}$ describing the choice of measurement,
and the random variable $X\in {\cal X}$ describing
the measurement outcome.
Hence,
two players, Alice (1st player) and Bob (2nd player),
have their respective random variables
$X_1,Y_1$ and $X_2,Y_2$, respectively.

In this scenario,
Alice and Bob independently decide their choice $Y_1$ and $Y_2$
according to the uniform distribution independently, i.e.,
$Y_1\perp Y_2$.
We denote the state reduction
with the measurement choice $Y_1=y_1\in \overline{\cal Y}_1$ in Alice's side
by $\{\Gamma_{x_1|y_1}^1\}_{x_1 \in {\cal X}_{y_1}}$, where
the state $\rho$ is changed to
$\Gamma_{x_1|y_1}^1(\rho)/\Tr \Gamma_{x_1|y_1}(\rho)$
with the measurement outcome $x_1$.
Here, $\Gamma_{x_1|y_1}^1$ is a CP map from
the input system ${\cal H}_{(I,1)}$ to the output system ${\cal H}_{(O,1)}$,
and $\sum_{x_1 \in {\cal X}_{y_1}}\Gamma_{x_1|y_1}^1$ is a TP-CP map.
Similarly, we define
$\{\Gamma_{x_2|y_2}^2\}_{x_2 \in {\cal X}_{y_2}}$ for
$y_2\in \overline{\cal Y}_2$.

Hence, Alice's and Bob's operations are described as
$\{\Gamma_{1,(x_1,y_1)}\}$
and $\{\Gamma_{2,(x_2,y_2)}\}$
under the formalism given in Section \ref{sec:problem-setting},
respectively, where
$\Gamma_{i,(x_i,y_i)}:= \frac{1}{|\overline{\cal Y}_i|}\Gamma_{x_i|y_i}^i$ for $i=1,2$.
Using \eqref{eq:GBR}, we define the joint distribution
$P_{X_1,Y_1,X_2,Y_2}$, which will be obtained by repeating several experiments.
This joint distribution gives the concrete form of the map ${\cal M}$ defined in \eqref{eq:model-map}.
Due to the condition
$\Gamma_{i,(x_i,y_i)}= \frac{1}{|\overline{\cal Y}_i|}\Gamma_{x_i|y_i}^i$,
the obtained distribution belongs to the following set
\begin{align}
{\cal M}[{\cal S}_G] \subset
{\cal P}_{{\cal U}}:=
\big\{P_{X_1,Y_1,X_2,Y_2}\big|
P_{Y_1,Y_2}(y_1,y_2)=\frac{1}{|\overline{\cal Y}_1||\overline{\cal Y}_2|}
,\quad
\forall y_1\in \overline{\cal Y}_1,y_2 \in \overline{\cal Y}_2\big\}.
\end{align}

Generally,
the set $\{
C[\Gamma_{1,(x_1,y_1)}]\otimes
C[\Gamma_{2,(x_2,y_2)}]
\}$ is not tomographically complete on
${\cal H}_{(I,1)} \otimes
{\cal H}_{(O,1)} \otimes{\cal H}_{(I,2)}\otimes{\cal H}_{(O,2)}$.
In this case,
we cannot recover the process matrix
from the obtained distribution $P_{X_1,Y_1,X_2,Y_2}$.
Therefore, there is a process matrix $W_S$ that
cannot be classified in the sense of \eqref{HI1}
from the obtained distribution $P_{X_1,Y_1,X_2,Y_2}$, i.e.,
the map ${\cal M}$ defined in \eqref{eq:model-map}
is not one-to-one.

Checking whether
$S$ belongs to ${\cal S}_{C|1\to 2}$ or ${\cal S}_{C}$
is difficult
even when the process matrix can be reconstructed.
Hence, we focus on the hierarchies given in \eqref{HI1},
i.e., we do not discuss the classes
${\cal S}_{C|1\to 2}$,
${\cal S}_{C|2\to 1}$, or ${\cal S}_{C}$.
Although the hierarchies given in \eqref{HI1}
yield the hierarchies:
\begin{align}
\begin{array}{ccccc}
{\cal M}[{\cal S}_{N,1\to 2}] & \supset &
{\cal M}[{\cal S}_{I}] &\subset & {\cal M}[{\cal S}_{N,2\to 1}] \\
\cap &  &\cap & & \cap  \\
{\cal M}[{\cal S}_{Q,1\to 2}] & \supset &
{\cal M}[{\cal S}_{P}] &\subset & {\cal M}[{\cal S}_{Q,2\to 1}] ,
\end{array}\Label{HI8}
\end{align}
there is a possibility that
the intersection
${\cal M}[{\cal S}_{N,1\to 2}] \cap {\cal M}[{\cal S}_{N,2\to 1}]
$ is strictly larger than ${\cal M}[{\cal S}_{I}]$.
In the following, we discuss how to characterize the sets
in the hierarchies given in \eqref{HI8}.

For the characterization, we have the following theorem.

\begin{theorem}\Label{TH4I}
When Charlie's strategy $S$ belongs to ${\cal S}_{Q,1\to2}$ or ${\cal S}_{N,1\to2}$,
the distribution ${\cal M}[S]$ satisfies the Markovian chain $X_1-Y_1-Y_2$.
When $S \in {\cal S}_{P}$, it satisfies both $X_1-Y_1-Y_2$ and $Y_1-Y_2-X_2$.
When $S \in {\cal S}_{I}$, it satisfies the Markovian chain $X_1-Y_1-Y_2-X_2$.
\end{theorem}

Except for ${\cal S}_{I}$, the above conditions come from the non-signaling condition.
Since the classes ${\cal S}_{N|1\to 2}$ and ${\cal S}_{Q|1\to 2}$
are characterized by the same Markovian chains $ X_1-Y_1-Y_2$,
the hierarchy ${\cal S}_{N|1\to 2}\subset {\cal S}_{Q|1\to 2}$
cannot be distinguished by using the Markov chain.
In the next section, we study
whether there exist other conditions to
detect the above hierarchies
under a special example.

\subsection{Pseudo-density matrix}
We now introduce the formalism of pseudo-density matrix (PDM)\cite{FJV,Liu3,Liu}
as a state-like representation of the observed distribution
$P_{X_1,Y_1,X_2,Y_2}$.
Its role in the present paper is to describe the compatible bipartite and marginal states
associated with the observed data.
To formulate this representation, we first fix a family of auxiliary Hermitian operators
built from the state-reduction maps introduced in the previous section.
These operators will serve to extract the relevant moments of
$P_{X_1,Y_1,X_2,Y_2}$ and hence to define compatibility with a PDM.

To encode the observed statistics in a convenient linear form, we introduce,
for each measurement setting, a family of independent centered observables.
These observables will later be used to define the moments appearing in the
pseudo-density-matrix description.

We assume that ${\cal H}_{(I,i)}={\cal H}_{(O,i)}$ and denote their common
dimension by $d_i$. We also enlarge the set of measurement choices to
$\overline{\cal Y}_i={\cal Y}_i\cup\{0\}$ for $i=1,2$, where the additional
label $0$ represents the trivial choice of doing nothing. For this choice,
we define $\Gamma_{0|0}^i$ to be the identity map on the density matrices on
${\cal H}_{(I,i)}$.

Fix $y_i\in{\cal Y}_i$. Since the outcome probabilities for a fixed
measurement $y_i$ satisfy one normalization constraint, only
$|{\cal X}_{y_i}|-1$ independent linear combinations are needed. We therefore
choose an auxiliary index set ${\cal Z}_{y_i}$ with
$|{\cal Z}_{y_i}|=|{\cal X}_{y_i}|-1$.
For each $z_i\in{\cal Z}_{y_i}$, we choose real coefficients
$\{g_{x_i|y_i,z_i}\}_{x_i\in{\cal X}_{y_i}}$ and define a Hermitian operator
$G_{y_i,z_i}$ by
\[
G_{y_i,z_i}
=
\sum_{x_i\in{\cal X}_{y_i}}
g_{x_i|y_i,z_i}\,\Gamma_{x_i|y_i}^{i,\dagger}(I).
\]
We require that the family
$\{G_{y_i,z_i}\}_{z_i\in{\cal Z}_{y_i}}\cup\{I\}$
be linearly independent and that each $G_{y_i,z_i}$ be traceless:
$\Tr G_{y_i,z_i}=0$.
Thus, the identity operator $I$ accounts for normalization, while the
operators $G_{y_i,z_i}$ describe the nontrivial independent components of
the measurement statistics.

For later convenience, we collect these operators into a single indexed family
\[
\{G_{t_i}\}_{t_i\in{\cal T}_i}
:=
\{I\}\cup \{G_{y_i,z_i}\}_{y_i\in{\cal Y}_i,\ z_i\in{\cal Z}_{y_i}},
\]
where the index $0\in{\cal T}_i$ is reserved for the identity operator, so that $G_0=I$.
For $t_i\neq0$, we identify $t_i$ with the pair
$(y(t_i),z(t_i))$, and we abbreviate
$g_{x_i|y(t_i),z(t_i)}$ to $g_{x_i|t_i}$.

Intuitively, for each measurement setting $y_i$, we do not work directly with
the raw outcome labels $x_i$. Instead, we replace them by a set of
$|{\cal X}_{y_i}|-1$ independent contrast variables. 
The corresponding
operators $G_{y_i,z_i}$ play the role of a basis for these contrasts and will
allow us to express the observed distribution through linear moments.

Although the joint distribution $P_{X_1,Y_1,X_2,Y_2}$ is defined in \eqref{eq:GBR},
the conditional distribution $P_{X_1,X_2|Y_1,Y_2}$ is given as
\begin{align}
P_{X_1,X_2|Y_1,Y_2}(x_1,x_2|y_1,y_2):=
|\overline{\cal Y}|^2 P_{X_1,Y_1,X_2,Y_2}(x_1,y_1,x_2,y_2).
\end{align}
We call a Hermitian matrix $R_{1,2}$ with trace $1$ on
${\cal H}_{(I,1)} \otimes {\cal H}_{(I,2)}$ a pseudo density matrix (PDM).

\begin{definition}
We say that a PDM $R_{1,2}$ is compatible with the
distribution $P_{X_1,Y_1,X_2,Y_2}$
when the relation
\begin{align}
\Tr (G_{t_1}\otimes G_{t_2})R_{1,2}=
\langle G_{t_1}\otimes G_{t_2}\rangle_{P_{X_1,Y_1,X_2,Y_2}}\Label{VBSA}
\end{align}
holds for
$t_1 \in {\cal T}_1$,
$t_2 \in {\cal T}_2$, where
\begin{align}
\langle G_{t_1}\otimes G_{t_2}\rangle_{P_{X_1,Y_1,X_2,Y_2}}
:= \sum_{x_1,x_2}
g_{x_1|t_1} g_{x_2|t_2}
P_{X_1,X_2|Y_1,Y_2}(x_1,x_2|y(t_1),y(t_2)).\Label{BAS}
\end{align}
We say that a density matrix $\rho_1$ on ${\cal H}_{(I,1)}$
is compatible with the distribution $P_{X_1,Y_1,X_2,Y_2}$
when the relation
\begin{align}
\Tr G_{t_1} \rho_1=
\langle G_{t_1}\otimes I \rangle_{P_{X_1,Y_1,X_2,Y_2}}\Label{VBSA1}
\end{align}
holds for $t_1 \in {\cal T}_1$.
The compatibility of $\rho_2$ on ${\cal H}_{(I,2)}$
is defined in the same way.
\end{definition}

\begin{lemma}\Label{LV7}
For $P_{X_1,Y_1,X_2,Y_2} \in {\cal M}[{\cal S}_G]$,
there exist
a density matrix $\rho_1$ on ${\cal H}_{(I,1)}$
and
a density matrix $\rho_2$ on ${\cal H}_{(I,2)}$
that are  compatible with the distribution $P_{X_1,Y_1,X_2,Y_2}$.
\end{lemma}

\begin{proof}
Let
$\rho_1$ be Alice's input state when Bob does not act,
and
$\rho_2$ be Bob's input state when Alice does not act.
Then, $\rho_1$ satisfies the condition \eqref{VBSA1}.
Hence, the density matrix $\rho_1$ on ${\cal H}_{(I,1)}$
is compatible with the distribution $P_{X_1,Y_1,X_2,Y_2}$.
Similarly,
the density matrix $\rho_2$ on ${\cal H}_{(I,2)}$
is compatible with the distribution $P_{X_1,Y_1,X_2,Y_2}$.
\end{proof}

Any parallel strategy is completely specified by a bipartite initial state
$\rho$ on ${\cal H}_{(I,1)}\otimes {\cal H}_{(I,2)}$,
since Charlie introduces no sequential dependence between Alice and Bob.
We denote the corresponding parallel strategy by $S_P(\rho)$.
For such a strategy, the state $\rho$ itself is compatible with the observed distribution
${\cal M}[S_P(\rho)]$.
The following relation shows that the converse also holds:
a distribution belongs to ${\cal M}[{\cal S}_P]$ if and only if it satisfies
the two non-signalling conditions characteristic of parallel strategies
and admits a compatible bipartite density matrix;
\begin{align}
{\cal M}[{\cal S}_{P}]
=
\left\{
P_{X_1,Y_1,X_2,Y_2} \in {\cal P}_{\cal U}
\left|
\begin{array}{l}
P_{X_1,Y_1,X_2,Y_2}\ \hbox{satisfies } X_1-Y_1-Y_2
\ \hbox{and}\ Y_1-Y_2-X_2,\\
\hbox{there exists a bipartite density matrix }
\rho \hbox{ on } {\cal H}_{(I,1)}\otimes{\cal H}_{(I,2)}\\
\hbox{that is compatible with }
P_{X_1,Y_1,X_2,Y_2}.
\end{array}
\right.
\right\}.
\Label{ZL1}
\end{align}
Hence, the existence of a compatible density matrix is essential.

An individual strategy is precisely a parallel strategy whose initial bipartite state is a product state.
Accordingly, within the characterization of parallel strategies in \eqref{ZL1},
the additional Markovian condition $X_1-Y_1-Y_2-X_2$
singles out the individual class, since for fixed measurement choices
it expresses the independence of the two observed outcomes $X_1$ and $X_2$.
Moreover, for any distribution in ${\cal M}[{\cal S}_G]$,
Lemma~\ref{LV7} provides compatible local density matrices
$\rho_1$ on ${\cal H}_{(I,1)}$ and $\rho_2$ on ${\cal H}_{(I,2)}$;
under the condition $X_1-Y_1-Y_2-X_2$,
the product state $\rho_1\otimes\rho_2$ is itself compatible with the distribution.
Therefore,
\begin{align}
{\cal M}[{\cal S}_{I}]
&=
\left\{P_{X_1,Y_1,X_2,Y_2} \in{\cal P}_{\cal U}
\left|
\begin{array}{l}
P_{X_1,Y_1,X_2,Y_2} \hbox{ satisfies }
 X_1-Y_1-Y_2-X_2,\\
\hbox{there exists a compatible density matrix }\rho
\hbox{ with }P_{X_1,Y_1,X_2,Y_2}
\end{array}
\right.
\right\}
 \notag\\
&=
\{P_{X_1,Y_1,X_2,Y_2} \in {\cal M}[{\cal S}_G]
\mid
P_{X_1,Y_1,X_2,Y_2} \hbox{ satisfies }
 X_1-Y_1-Y_2-X_2
\}.
\Label{CWSI}
\end{align}

Equation~\eqref{CWSI} should be compared with the parallel characterization \eqref{ZL1}.
There, the class ${\cal M}[{\cal S}_P]$ is described by the two directional non-signalling conditions
$X_1-Y_1-Y_2$ and $Y_1-Y_2-X_2$, together with the existence of a compatible bipartite density matrix.
The present condition \eqref{CWSI} strengthens this by requiring that the compatible bipartite state be effectively a product state.
At the distribution level, this product structure is reflected exactly by the additional Markovian condition
$X_1-Y_1-Y_2-X_2$, which expresses the conditional independence of the two local outcomes once the measurement choices are fixed.
In this sense, the individual class is obtained from the parallel class by imposing the extra independence structure that removes all nontrivial input correlations between Alice and Bob.

\section{Method under projective state reduction via pseudo-density matrix}\Label{S5-2}
This section specializes the distribution-level framework of
Section~\ref{S5} to projective state reduction.  We first establish compatibility
relations for a memoryless sequential strategy and then discuss memoryful
extensions, uniqueness of the reconstructed objects, and the parallel class.
These results supply the reconstruction ingredients used in
Section~\ref{SSAL}; they do not yet provide the final direct membership test.

\subsection{Adaptive strategy without memory}\Label{SSPDM}
To carry out a more detailed analysis, we focus on
the case with projective state reductions.
For an element $y_i \in {\cal Y}_i$,
we choose a projection-valued measure
$\{E_{x_i|y_i}^i\}_{x_i \in {\cal X}_{y_i}}$,
where
$E_{x_i|y_i}^i$ is a projection for $x_i\in {\cal X}_{y_i}$.
We take the corresponding measurement process to be
$\{\Gamma_{x_i|y_i}^i\}_{x_i \in {\cal X}_{y_i}}$
in the formalism of Section \ref{sec:problem-setting} as
$\Gamma_{x_i|y_i}^i(\rho):= E_{x_i|y_i}^i\rho E_{x_i|y_i}^i$.
Hence, we have
$\Tr_{(O,i)} C[\Gamma_{x_i|y_i}^i] =\Gamma_{x_i|y_i}^{i,\dagger}(I)
=E_{x_i|y_i}^i$.

We study the class ${\cal S}_{N|1\to 2}$.
When Charlie's strategy $S$ belongs to the classes
${\cal S}_{N|1\to 2}$,
the strategy is written as the pair of
Alice's initial state $\rho_1$ before the measurement by Alice
and the Choi matrix $C_{1\to 2}$
of the TP-CP map that Charlie applies to Alice's output state.
In the following, we denote this strategy
by $S_{1\to 2}(\rho_1,C_{1\to 2})$.
The aim of this section is the development of a method to identify
whether the distribution
${\cal M}[S_{1\to 2}(\rho_1,C_{1\to 2})] $ belongs to
${\cal M}[{\cal S}_{N,1\to 2}]\cap {\cal M}[{\cal S}_{N,2\to 1}]$.

\begin{definition}
We say that
a pair of a Hermitian matrix $C_{1\to 2}$
on ${\cal H}_{(I,1)}\otimes {\cal H}_{(I,2)}$
and a density matrix $\rho_1$ on ${\cal H}_{(I,1)}$ is compatible
with $P_{X_1,Y_1,X_2,Y_2}$ via
a pseudo density matrix $R_{1,2}$
when $R_{1,2}$ is compatible with $P_{X_1,Y_1,X_2,Y_2}$
and
\begin{align}
\rho_1=\Tr_2 R_{1,2}, \quad
\Tr_2  {C}_{1\to 2}=I, \quad
R_{1,2}= C_{1\to 2}^{T_1} \circ (\rho_{1} \otimes I)\Label{XS} ,
\end{align}
where $T_i$ expresses the partial transpose of the $i$-th system.
\end{definition}
When $\rho_1$ is invertible, the second condition in \eqref{XS} is implied by the other two conditions, because they yield
$\rho_1 = (\Tr_2 C_{1\to2}) \circ \rho_1$.

Similarly, we say that
a pair of a Hermitian matrix $C_{2\to 1}$
on ${\cal H}_{(I,1)}\otimes {\cal H}_{(I,2)}$
and a density matrix $\rho_2$ on ${\cal H}_{(I,2)}$ is compatible
with $P_{X_1,Y_1,X_2,Y_2}$ via a pseudo density matrix $R_{1,2}$
when $R_{1,2}$ is compatible with $P_{X_1,Y_1,X_2,Y_2}$
and
\begin{align}
\rho_2=\Tr_1 R_{1,2}, \quad
\Tr_1  {C}_{2\to 1}=I, \quad
R_{1,2}= C_{2\to 1}^{T_2} \circ (I \otimes \rho_{2} )\Label{XS2} .
\end{align}
Hereafter, we omit the phrase ``via a pseudo density matrix $R_{1,2}$''
in the above definitions.

\begin{lemma}
A pair of a Hermitian matrix $C_{1\to 2}$
on ${\cal H}_{(I,1)}\otimes {\cal H}_{(I,2)}$
and a density matrix $\rho_1$ on ${\cal H}_{(I,1)}$ is compatible
with $P_{X_1,Y_1,X_2,Y_2}$ if and only if the relations
\begin{align}
\Tr_2  {C}_{1\to 2}&=I \Label{MGH} \\
\langle G_{t_1}\otimes G_{t_2} \rangle_{P_{X_1,Y_1,X_2,Y_2}}
&=\Tr
\big({C}_{1\to 2}^{T_1}
\circ\big( \rho_1  \otimes I \big)\big)
 \big( G_{t_1} \otimes  G_{t_2}\big) \Label{FN3Y}
\end{align}
hold for $t_1,t_2 \in {\cal T}$.
\end{lemma}

\begin{proof}
Assume that a pair of a Hermitian matrix $C_{1\to 2}$
on ${\cal H}_{(I,1)}\otimes {\cal H}_{(I,2)}$
and a density matrix $\rho_1$ on ${\cal H}_{(I,1)}$ is compatible
with $P_{X_1,Y_1,X_2,Y_2}$ via $R_{1,2}$.
The combination of \eqref{VBSA} and the third equation of \eqref{XS} yields \eqref{FN3Y}.
Assume the relations \eqref{MGH} and \eqref{FN3Y}.
We set $R_{1,2}$ to be $\big({C}_{1\to 2}^{T_1}
\circ\big( \rho_1  \otimes I \big)\big)$.
The relation \eqref{FN3Y} guarantees that
$R_{1,2}$ is compatible with $P_{X_1,Y_1,X_2,Y_2}$, which implies the third equation of \eqref{XS}.
The combination of
the \eqref{MGH} and the third equation of \eqref{XS}
implies the first equation of \eqref{XS}.
\end{proof}

\begin{theorem}\Label{TH6}

Alice's initial state $\rho_1$ is compatible with
${\cal M}[S_{1\to 2}(\rho_1,C_{1\to 2})]$.
\end{theorem}

\begin{proof}
The initial state $\rho_1$ on Alice's system satisfies
\begin{align}
P_{X_1|Y_1,Y_2}(x_1|y_1,0)
=
\Tr \rho_1 \Gamma_{x_1|y_1}^{1,\dagger}(I)
=
\Tr \rho_1 E_{x_1|y_1}^1
\end{align}
for $x_1,y_1$.
Since the distribution
${\cal M}[S_{1\to 2}(\rho_1,C_{1\to 2})]$ satisfies the Markovian chain $X_1-Y_1-Y_2$,
$\rho_1$ is compatible with ${\cal M}[S_{1\to 2}(\rho_1,C_{1\to 2})]$.
\end{proof}
To generalize
\cite[Theorem 1]{Liu3}, \cite[Theorem 4.1]{PhD}, we
introduce the following condition for
$\{G_t\}_{t \in {\cal T}}$ and $\{E_y\}_{y\in {\cal Y}}$:
\begin{description}
\item[(A2)]
For any $t_i \in {\cal T}_i\setminus \{0\}$,
$G_{t_i}^2$ is a scalar multiple of the identity operator.
Also, $E_{y(t_i)}$ is the spectral decomposition of $G_{t_i}$,
i.e.,
$|{\cal X}_{y(t_i)}|=2$ and any two distinct elements $x_i,x_i'$ in ${\cal X}_{y(t_i)}$ satisfy the condition
$g_{x_i|y(t_i),z(t_i)}\neq g_{x_i'|y(t_i),z(t_i)} \neq 0$.
\end{description}
Pauli matrices on an $n$-qubit system satisfy the above
condition.

\begin{lemma}\Label{LJG1}
Assume the condition (A2).

We have
\begin{align}
\langle G_{t_1}\otimes G_{t_2} \rangle_{{\cal M}[S_{1\to 2}(\rho_1,C_{1\to 2})]}
=\Tr
\big({C}_{1\to 2}^{T_1}
\circ\big( \rho_1  \otimes I \big)\big)
 \big( G_{t_1} \otimes  G_{t_2}\big) \Label{FN3}
\end{align}
for $t_1,t_2 \in {\cal T}$.
\end{lemma}

\begin{proof}

In this proof, we denote the joint distribution
${\cal M}[S_{1\to 2}(\rho_1,C_{1\to 2})]$
by $P_{X_1,Y_1,X_2,Y_2}(x_1,y_1,x_2,y_2)$.
Due to the condition (A2), for $t (\neq 0)\in {\cal T}$,
the set ${\cal X}_{y(t)}$ consists of two elements.
Without loss of notation, suppose that ${\cal X}_{y(t)}
=\{0,1\}$
and $g_{x|t}=(-1)^x \kappa_{t}$, where $G_{t}^2=\kappa_{t}^2 I$.
Hence, we have $G_{t}= \kappa_t E_{0|y(t)}-\kappa_t E_{1|y(t)}$.
Since
\begin{align}
&\sum_{x_1} g_{x_1|t} E_{x_1|y(t)} \rho_{1} E_{x_1|y(t)} \notag \\
=&
\kappa_{t} E_{0|y(t)} \rho_{1} E_{0|y(t)}
-\kappa_{t} E_{1|y(t)} \rho_{1} E_{1|y(t)} \notag \\
=&
\frac{1}{2}
\Big(\kappa_t E_{0|y(t)} \rho_{1} E_{0|y(t)}
-\kappa_t E_{0|y(t)} \rho_{1} E_{1|y(t)}
+\kappa_t E_{1|y(t)} \rho_{1} E_{0|y(t)}
-\kappa_t E_{1|y(t)} \rho_{1} E_{1|y(t)} \Big)\notag \\
&+\frac{1}{2}
\Big(\kappa_t E_{0|y(t)} \rho_{1} E_{0|y(t)}
+\kappa_t E_{0|y(t)} \rho_{1} E_{1|y(t)}
-\kappa_t E_{1|y(t)} \rho_{1} E_{0|y(t)}
-\kappa_t E_{1|y(t)} \rho_{1} E_{1|y(t)} \Big)\notag \\
=&
\frac{1}{2}
\Big(E_{0|y(t)} \rho_{1} G_{t}
+E_{1|y(t)} \rho_{1} G_{t}
+G_{t} \rho_{1} E_{0|y(t)}
+G_{t} \rho_{1} E_{1|y(t)} \Big)
=
\rho_1 \circ G_{t},
\end{align}
we have
\begin{align}
&\langle G_{t_1}\otimes G_{t_2} \rangle_{{\cal M}[S_{1\to 2}(\rho_1,C_{1\to 2})]}
\notag \\
=&
\sum_{x_1,x_2\in {\cal X}}
g_{x_1|t_1} g_{x_2|t_2}
P_{X_1,X_2|Y_1,Y_2}(x_1,x_2|y(t_1),y(t_2)) \notag \\
=&
\sum_{x_1,x_2}g_{x_1|t_1}g_{x_2|t_2}
\Tr_2 \Big(\Tr_1 \Big({C}_{1\to 2}^{T_1}
(E_{x_1|y(t_1)}\rho_{1}E_{x_1|y(t_1)} \otimes I)
 E_{x_2|y(t_2)}\Big) \notag \\
=&\Tr
{C}_{1\to 2}^{T_1}
\Big( \Big(\sum_{x_1} g_{x_1|t_1} E_{x_1|y(t_1)} \rho_{1}
E_{x_1|y(t_1)}\Big)
 \otimes  G_{y(t_2)}\Big) \notag \\
=&\Tr
{C}_{1\to 2}^{T_1}
\big( (\rho_1 \circ G_{t_1})
 \otimes  G_{t_2}\big) 
=\Tr
{C}_{1\to 2}^{T_1}
\big( \rho_1  \otimes I \big)
\circ \big( G_{t_1} \otimes  G_{t_2}\big) 
=\Tr
{C}_{1\to 2}^{T_1}
\circ\big( \rho_1  \otimes I \big)
 \big( G_{t_1} \otimes  G_{t_2}\big) .\Label{FN3B}
\end{align}
\end{proof}

Lemma \ref{LJG1} implies that
the PDM ${C}_{1\to 2}^{T_1}
\circ\big( \rho_1  \otimes I \big)$ is compatible
with ${\cal M}[S_{1\to 2}(\rho_1,C_{1\to 2})]$.
Hence, we obtain the following theorem.

\begin{theorem}\Label{TH7}
Assume the condition (A2).

Then, the pair of a Hermitian matrix $C_{1\to 2}$
on ${\cal H}_{(I,1)}\otimes {\cal H}_{(I,2)}$
and a density matrix $\rho_1$ on ${\cal H}_{(I,1)}$ is compatible
with the distribution ${\cal M}[S_{1\to 2}(\rho_1,C_{1\to 2})]$.
\end{theorem}

Then,
Theorem \ref{TH7} implies the following corollary.

\begin{corollary}\Label{COR6}
Assume the condition (A2).
We have
\begin{align}
&{\cal M}[{\cal S}_{N,1\to 2}]\notag \\
=&
\left
\{P_{X_1,Y_1,X_2,Y_2} \in {\cal P}_{\cal U}
\left|
\begin{array}{l}
\hbox{There exists a compatible
pair of $ \rho_1\ge 0$ and $C_{1\to 2}\ge 0$
with } P_{X_1,Y_1,X_2,Y_2} \\
\hbox{such that }
P_{X_1,Y_1,X_2,Y_2}={\cal M}[S_{1\to 2}(\rho_1, C_{1\to 2})].
\end{array}
\right.\right\}
\Label{FFT1}.
\end{align}
\end{corollary}

\begin{proof}
The relation $\supset$ is trivial. The relation $\subset$ follows from Theorem \ref{TH7}.
\end{proof}

This is still not a direct distribution-level test, because it is formulated through the existence of a generating compatible pair. Its role is only to prepare the sharper criterion given later.

\begin{lemma}\Label{LL6}
Assume the condition (A2).
Assume that a pair of a Hermitian matrix $C_{1\to 2}$
on ${\cal H}_{(I,1)}\otimes {\cal H}_{(I,2)}$
and a density matrix $\rho_1$ on ${\cal H}_{(I,1)}$ is compatible
with $P_{X_1,Y_1,X_2,Y_2}$.
Then,
the PDM ${C}_{1\to 2}^{T_1} \circ\big( \rho_1  \otimes I \big)$
is compatible with ${\cal M}[{\cal S}_{1\to 2}(\rho_1,C_{1\to 2})]$.
That is, the distribution
$P_{X_1,Y_1,X_2,Y_2}$ satisfies

\begin{align}
\langle G_{t_1}\otimes G_{t_2} \rangle_{P_{X_1,Y_1,X_2,Y_2}}
=\Tr \big({C}_{1\to 2}^{T_1} \circ\big( \rho_1  \otimes I \big)\big) \big( G_{t_1} \otimes  G_{t_2}\big)
=\langle G_{t_1}\otimes G_{t_2} \rangle_{{\cal M}[S_{1\to 2}(\rho_1,C_{1\to 2})]}
 \Label{FN334}
\end{align}
for $t_1,t_2 \in {\cal T}$.

\end{lemma}

\begin{proof}
It is sufficient to show \eqref{FN334}.
The first equation of \eqref{FN334} follows from \eqref{FN3Y}.
The second equation of \eqref{FN334} follows from \eqref{FN3} of Lemma \ref{LJG1}.
\end{proof}

\begin{lemma}\Label{LL7}
Assume the condition (A2).
Assume that a pair of a Hermitian matrix $C_{1\to 2}$
on ${\cal H}_{(I,1)}\otimes {\cal H}_{(I,2)}$
and a density matrix $\rho_1$ on ${\cal H}_{(I,1)}$ is compatible
with $P_{X_1,Y_1,X_2,Y_2}\in {\cal M}[{\cal S}_G]$.
When the distributions
$P_{X_1,Y_1,X_2,Y_2}$ and ${\cal M}[S_{1\to 2}(\rho_1, C_{1\to 2})]$
satisfy the Markovian conditions
$X_1-Y_1-Y_2$ and $Y_1-Y_2-X_2$,
we have
${\cal M}[S_{1\to 2}(\rho_1, C_{1\to 2})]=P_{X_1,Y_1,X_2,Y_2}$.
\end{lemma}

\begin{proof}
In this proof, we denote the distribution
${\cal M}[S_{1\to 2}(\rho_1, {C}_{1\to 2})
$
by $\tilde{P}_{X_1,Y_1,X_2,Y_2}$.
It is sufficient to show the relation
\begin{align}
\tilde{P}_{X_1,X_2|Y_1=y_1,Y_2=y_2}= P_{X_1,X_2|Y_1=y_1,Y_2=y_2}
\end{align}
for $y_1,y_2 \in {\cal Y}$.
To show the above relation,
it is sufficient to show that
these two conditional distributions have the same expectation of
the three random variables
$g_{X_1|y_1,z_1} $,
$g_{X_2|y_2,z_2}$, $g_{X_1|y_1,z_1} g_{X_2|y_2,z_2}$
for $z_1\in {\cal Z}_{y_1}$, $z_2\in {\cal Z}_{y_2}$
because
these random variables span the linear function space on
${\cal X}_{y_1}\times {\cal X}_{y_2}$ with the constant function.

Since their expectations of the variable $g_{X_1|y_1,z_1} g_{X_2|y_2,z_2}$
is given as
$\langle G_{y_1,z_1}\otimes G_{y_1,z_1} \rangle_{\tilde{P}_{X_1,Y_1,X_2,Y_2}}$
and
$\langle G_{y_1,z_1}\otimes G_{y_2,z_2} \rangle_{{P}_{X_1,Y_1,X_2,Y_2}}$,
\eqref{FN334} guarantees that these two expectations take the same value.
The Markovian chain condition $X_1-Y_1-Y_2$ guarantees that
their expectations of the variable $g_{X_1|y_1,z_1} $
is given as
$\langle G_{y_1,z_1}\otimes I \rangle_{\tilde{P}_{X_1,Y_1,X_2,Y_2}}$
and
$\langle G_{y_1,z_1}\otimes I \rangle_{{P}_{X_1,Y_1,X_2,Y_2}}$.
Then,
\eqref{FN334} guarantees that these two expectations take the same value.
Also, in the same way, we can show that
their expectations of
the variable $g_{X_2|y_2,z_2} $ take the same value.
\end{proof}

\begin{lemma}\Label{LL8}
Assume the condition (A2).
Given $S \in {\cal S}_{N|1\to 2}$, we assume
that a pair of a Hermitian matrix $C_{2\to 1}$
on ${\cal H}_{(I,1)}\otimes {\cal H}_{(I,2)}$
and a density matrix $\rho_2$ on ${\cal H}_{(I,2)}$ is compatible
with ${\cal M}[S]$.
The following conditions are equivalent.
\begin{description}
\item[(B1)]
The distributions
${\cal M}[S]$ and
${\cal M}[S_{2\to 1}(\rho_2, C_{2\to 1})]$
satisfy the Markovian conditions
$X_1-Y_1-Y_2$ and $Y_1-Y_2-X_2$.

\item[(B2)]
${\cal M}[S_{2\to 1}(\rho_2, C_{2\to 1})]={\cal M}[S]$.
\end{description}
\end{lemma}

\begin{proof}
Lemma \ref{LL7} guarantees that the condition (B1) implies the condition (B2).
The opposite direction can be shown as follows.
When the relation ${\cal M}[S_{2\to 1}(\rho_2, C_{2\to 1})]
={\cal M}[S]$ holds,
the distribution
${\cal M}[S_{2\to 1}(\rho_2, C_{2\to 1})]
={\cal M}[S]$ satisfies the Markovian conditions
$X_1-Y_1-Y_2$ and $Y_1-Y_2-X_2$
because
${\cal M}[S]$
satisfies the Markovian condition $X_1-Y_1-Y_2$
and
${\cal M}[S_{2\to 1}(\rho_2, C_{2\to 1})]$
satisfies the Markovian condition $Y_1-Y_2-X_2$.
\end{proof}

The above conditions do not guarantee
the condition ${C}_{2\to 1}\ge 0$.
Adding this condition, we have the following corollary of Lemma \ref{LL8}.

\begin{corollary}\Label{COR9B}
Assume the condition (A2).
We have
\begin{align}
&{\cal M}[{\cal S}_{N,1\to 2}] \cap {\cal M}[{\cal S}_{N,2\to 1}] \\
=&
\left\{{\cal M}[S]\in {\cal M}[{\cal S}_{N,1\to 2}]\left|
 \begin{array}{l}
\hbox{There exists a compatible pair of
a Hermitian matrix $C_{2\to 1}\ge 0$
on ${\cal H}_{(I,1)}\otimes {\cal H}_{(I,2)}$} \\
\hbox{and a density matrix $\rho_2$ on ${\cal H}_{(I,2)}$}
\hbox{ with } {\cal M}[S]
\hbox{ such that }
{\cal M}[S]={\cal M}[S_{2\to 1}(\rho_2, C_{2\to 1})]
\end{array}
\right.\right
\} \notag\\
=&
\left\{{\cal M}[S]\in {\cal M}[{\cal S}_{N,1\to 2}]\left|
 \begin{array}{l}
\hbox{There exists a compatible pair of
a Hermitian matrix $C_{2\to 1}\ge 0$
on ${\cal H}_{(I,1)}\otimes {\cal H}_{(I,2)}$} \\
\hbox{and a density matrix $\rho_2$ on ${\cal H}_{(I,2)}$}
\hbox{ with } {\cal M}[S] \\
\hbox{such that }
{\cal M}[S_{2\to 1}(\rho_2, C_{2\to 1})]
\hbox{ and }{\cal M}[S] \hbox{ satisfy }X_1-Y_1-Y_2,Y_1-Y_2-X_2.
\end{array}
\right.\right\} .
\end{align}
\end{corollary}
\begin{proof}
The first equation follows from
Corollary \ref{COR6}.
The second equation follows from  Lemma \ref{LL8}.
\end{proof}

Corollary \ref{COR6} is useful for identifying
whether a strategy $S \in {\cal S}_{N,1\to 2}$
satisfies the condition
${\cal M}[S] \in {\cal M}[{\cal S}_{N,1\to 2}]\cap {\cal M}[{\cal S}_{N,2\to 1}]$.
This characterization of ${\cal M}[{\cal S}_{N,1\to 2}]\cap {\cal M}[{\cal S}_{N,2\to 1}]$ provides an explicit description
in Section \ref{S7}.
When the observed distribution belongs to the set ${\cal M}[{\cal S}_{N,1\to 2}] \cap {\cal M}[{\cal S}_{N,2\to 1}] $,
we cannot distinguish which order, $1\to 2$ or $2 \to 1$,
is correct.
The characterization of this set clarifies the limitation of
the method of projective state reduction.

\subsection{Adaptive strategy with memory}\Label{SSPDM2}
We next explain how the preceding analysis extends beyond the memoryless case.
Once the projective-state-reduction setting has been fixed, strategies with
memory can still be related to the memoryless framework in a structurally
simple way.
For classical memory, the observed distribution is obtained by probabilistically
mixing memoryless sequential strategies.
For quantum memory, the same idea survives after enlarging the input space so as
to include Charlie's memory system. 

In the class ${\cal S}_{C,1\to 2}$, Charlie has a classical memory system $M$.
Conditioned on the hidden memory value $M=m$, Charlie chooses an initial state
$\rho_{1,m}\ge 0$ and a Choi matrix $C_{1\to 2,m}\ge 0$.
Thus, for each fixed value of $m$, the corresponding branch is an ordinary
memoryless sequential strategy.
Since the value of $M$ is not observed by Alice and Bob, the actual observed
distribution is obtained only after averaging over $m$.
Hence ${\cal M}[{\cal S}_{C,1\to 2}]$ is given by the convex hull of
${\cal M}[{\cal S}_{N,1\to 2}]$ characterized in \eqref{FFT1}. 
Every classical-memory sequential strategy is also a quantum-memory sequential strategy, because the classical memory values can be encoded in mutually orthogonal states of a quantum memory system and the corresponding conditional operations can then be applied to these orthogonal states. Therefore,
\begin{align}
{\cal S}_{C,1\to2}&\subset{\cal S}_{Q,1\to2},\notag\\
{\cal M}[{\cal S}_{C,1\to2}]&\subset{\cal M}[{\cal S}_{Q,1\to2}].
\Label{classical-quantum-memory-inclusion}
\end{align}

It is important, however, not to confuse this branchwise positivity with the
positivity of the reconstructed Choi matrix obtained from the coarse-grained
observed distribution.
Even though every branch satisfies $C_{1\to 2,m}\ge 0$, the reconstructed matrix
$\hat C_{1\to 2}[P]$ associated with the averaged distribution $P$ need not be
positive, because the hidden classical memory variable $M$ is not accessible at
the distribution level.
In particular, as shown later in Example~\ref{Ex1} in Section \ref{S7-C}, 
positivity of
$\hat C_{1\to 2}[P]$ can already fail in the classical-memory subclass. 

Next, we study the class ${\cal S}_{Q,1\to 2}$.
In this case, without loss of generality, we may assume that Charlie's quantum
memory system ${\cal H}_{M,1}$ has the same dimension as
${\cal H}_{(I,1)}$.
Then Alice's projective measurement can be regarded as
$\{E_{x_1\mid y_1}^1\otimes I\}_{x_1\in{\cal X}_{y_1}}$
on the enlarged system
${\cal H}_{(I,1)}\otimes{\cal H}_{(M,1)}$.
In this sense, a quantum-memory strategy can be viewed as an ordinary sequential
strategy on a larger input space.
Hence, the analysis of Section~\ref{SSPDM} applies to
${\cal S}_{Q,1\to 2}$ after this enlargement of the input space. 

\subsection{Uniqueness of compatible elements}
In the above analysis,
the compatible pair cannot be uniquely determined.
For the uniqueness, we assume the following condition.
\begin{description}
\item[(A1)]
The set $\{G_{y_i,z_i}\}_{y_i\in {\cal Y}_i, z_i\in {\cal Z}_{y_i}}\cup
\{I\}$
forms a basis of ${\cal B}({\cal H}_{(I,i)})$.
\end{description}
Under the above condition, ${\cal T}_i$ is given as $\{0,1,\ldots, d_i^2-1\}$.

\begin{example}\Label{ex1}
In a typical case,
we set
${\cal Y}_1={\cal Y}_2=\{1, \ldots, d+1\}$,
and
${\cal X}_y=\{1, \ldots, d\}$ and
${\cal Z}_y=\{1, \ldots, d-1\}$ for $y \in {\cal Y}$.
Then, we choose a rank-one projection-valued measure
$\{E_{x|y}\}_{x \in {\cal X}_y}$.
When $d$ is a prime power,
there are $d+1$ mutually unbiased bases.
When $d+1$ rank-one projection-valued measures
$\{E_{x|y}\}_{x \in {\cal X}_y}$
constructed from $d+1$ mutually unbiased bases,
the condition (A1) holds.

\end{example}

To describe the unique compatible elements,
we employ the dual basis
$\{H_t\}_{t \in {\cal T}}$ of $\{ G_t\}_{t \in {\cal T}}$ as
\begin{align}
\Tr G_{t'}H_t=\delta_{t',t}.\Label{dualB}
\end{align}
Thus, $H_0=\frac{1}{d}I$.
Then,
the unique compatible PDM with $P_{X_1,Y_1,X_2,Y_2}$
is given as
\begin{align}
{\cal R}[P_{X_1,Y_1,X_2,Y_2}]:= \sum_{t_1,t_2\in {\cal T}}
\langle G_{t_1}\otimes G_{t_2}\rangle_{P_{X_1,Y_1,X_2,Y_2}}
H_{t_1}\otimes H_{t_2}
\end{align}
because the unique compatiblility with $P_{X_1,Y_1,X_2,Y_2}$ is
guaranteed by the relation
\begin{align}
\Tr (G_{t_1}\otimes G_{t_2}){\cal R}[P_{X_1,Y_1,X_2,Y_2}]=
\langle G_{t_1}\otimes G_{t_2}\rangle_{P_{X_1,Y_1,X_2,Y_2}}\Label{VBSA2}
\end{align}
for $t_1,t_2 \in {\cal T}$.
Also, the matrices
 $\rho_1[P_{X_1,Y_1,X_2,Y_2}]$ and
$\rho_2[P_{X_1,Y_1,X_2,Y_2}]$ defined below are
compatible with $P_{X_1,Y_1,X_2,Y_2}$;
\begin{align}
\rho_1[P_{X_1,Y_1,X_2,Y_2}]:=\Tr_2 {\cal R}[P_{X_1,Y_1,X_2,Y_2}],\quad
\rho_2[P_{X_1,Y_1,X_2,Y_2}]:=\Tr_1 {\cal R}[P_{X_1,Y_1,X_2,Y_2}].
\Label{HJK1}
\end{align}
For $P_{X_1,Y_1,X_2,Y_2} \in {\cal M}[{\cal S}_G]$,
$\rho_1[P_{X_1,Y_1,X_2,Y_2}]$ is the unique compatible matrix with
$P_{X_1,Y_1,X_2,Y_2}$.
Hence, Lemma \ref{LV7} implies the following lemma.
\begin{lemma}\Label{LBM1A}
For $P_{X_1,Y_1,X_2,Y_2} \in {\cal M}[{\cal S}_G]$,
$\rho_1[P_{X_1,Y_1,X_2,Y_2}]$ is the unique compatible density matrix with
$P_{X_1,Y_1,X_2,Y_2}$.
\end{lemma}

Further, when $\rho_1[P_{X_1,Y_1,X_2,Y_2}]>0$,
the Hermitian matrix $\tilde{C}_{1\to 2}[P_{X_1,Y_1,X_2,Y_2}]$ on ${\cal H}_{(I,1)}\otimes {\cal H}_{(I,2)}$
satisfying the following condition is uniquely determined;
The pair of the Hermitian matrix $C_{1\to 2}$
and a density matrix $\rho_1[P_{X_1,Y_1,X_2,Y_2}]$
is compatible with $P_{X_1,Y_1,X_2,Y_2}$, i.e.,
\begin{align}
{\cal R}[P_{X_1,Y_1,X_2,Y_2}]
= \tilde{C}_{1\to 2}[P_{X_1,Y_1,X_2,Y_2}]^{T_1}\circ
(\rho_{1}[P_{X_1,Y_1,X_2,Y_2}] \otimes I)\Label{XSA} .
\end{align}
Similarly,
when  $\rho_2[P_{X_1,Y_1,X_2,Y_2}]>0$,
we define $\tilde{C}_{2\to 1}[P_{X_1,Y_1,X_2,Y_2}]$ as
\begin{align}
{\cal R}[P_{X_1,Y_1,X_2,Y_2}]
= \tilde{C}_{2\to 1}[P_{X_1,Y_1,X_2,Y_2}]^{T_2}\circ
(I \otimes \rho_{2}[P_{X_1,Y_1,X_2,Y_2}] ).\Label{XSA2}
\end{align}

When conditions (A1) and (A2) hold,
Theorem \ref{TH6} guarantees
\begin{align}
\rho_1[{\cal M}[S_{1\to 2}(\rho_1,C_{1\to 2})]]=\rho_1.
\end{align}
In addition, when $\rho_1>0$,
Theorem \ref{TH7} guarantees the relation
\begin{align}
\tilde{C}_{1\to 2}[{\cal M}[S_{1\to 2}(\rho_1,C_{1\to 2})]]
=C_{1\to 2}.
\end{align}
That is, the pair $(\rho_1[{\cal M}[S_{1\to 2}(\rho_1,C_{1\to 2})]]$,
$\tilde{C}_{1\to 2}[{\cal M}[S_{1\to 2}(\rho_1,C_{1\to 2})]])$
recovers the joint distribution ${\cal M}[S_{1\to 2}(\rho_1,C_{1\to 2})]$.
This statement is a generalization of
\cite[Theorem 1]{Liu3}, \cite[Theorem 4.1]{PhD}.

\subsection{Parallel strategy}
Under (A1), the compatible PDM ${\cal R}[P]$ is unique.  Positivity alone
does not imply that the observed marginal probabilities are independent of
the other player's measurement choice.  Thus \eqref{ZL1} becomes
\begin{align}
{\cal M}[{\cal S}_P]=\bigl\{P\in{\cal P}_{\cal U}:\;&
P\text{ satisfies }X_1-Y_1-Y_2\text{ and }Y_1-Y_2-X_2,\\
&{\cal R}[P]\ge0\bigr\}.\Label{FG1VS}
\end{align}
The non-signaling conditions and the unique compatible density matrix
together reproduce the observed distribution, as required in \eqref{ZL1}.
For $P\in{\cal M}[{\cal S}_{Q,1\to2}]$, the first condition holds by
Theorem~\ref{TH4I}, so the parallel membership criterion is precisely
$Y_1-Y_2-X_2$ together with ${\cal R}[P]\ge0$.

\section{Necessary and Sufficient Conditions for Memoryless Sequential Dynamics}\Label{SSAL}
We now convert the compatibility formalism into a direct criterion for
membership in ${\cal M}[{\cal S}_{N,1\to2}]$.  The first subsection
reconstructs a candidate Choi matrix from the observed conditional
probabilities.  The second identifies the additional distribution-level
consistency relation needed to prove that the reconstructed state and channel
reproduce the full observed distribution.

\subsection{Reconstruction of the Choi matrix under projective state reduction}
The analysis of adaptive strategies without memory in Section~\ref{SSPDM}
characterized the set ${\cal M}[{\cal S}_{N,1\to 2}]$ in Corollary~\ref{COR6}.
However, in order to decide whether a given joint distribution
$P_{X_1,Y_1,X_2,Y_2}$ belongs to ${\cal M}[{\cal S}_{N,1\to 2}]$,
this characterization is still limited in two respects.
First, it assumes the strong condition (A2)
for $\{G_{t_i}\}_{t_i \in {\cal T}_i}$ and $\{E_{y_i}\}_{y_i\in {\cal Y}_i}$.
Second, Corollary~\ref{COR6} does not yet provide a direct membership criterion
in terms of the observed distribution itself.
To obtain a more concrete characterization, we now replace (A2) by the following
rank-one assumption, under which the effective Choi matrix can be reconstructed
directly from $P_{X_1,Y_1,X_2,Y_2}$.
\begin{description}
\item[(A3)]
Each operator $E_{x_i|y_i}^i$ has rank one
for any $y_i \in {\cal Y}_i$ and $x_i \in {\cal X}_{y_i}$.
That is,
$E_{x_i|y_i}^i$ is written as $|x_i,y_i\rangle\langle x_i,y_i|$
by using a vector $|x_i,y_i\rangle$.~\cite{HayashiQIT2017}
\end{description}
We now derive a direct characterization of
the set ${\cal M}[{\cal S}_{N,1\to 2}]$.
For this purpose, define the conditional distribution
$P_{X_2|X_1,Y_1,Y_2}$ as follows
\begin{align}
P_{X_2|X_1,Y_1,Y_2}(x_2|x_1,y_1,y_2)
:=\frac{P_{X_1,X_2|Y_1,Y_2}(x_1,x_2|y_1,y_2)
}{
\sum_{x_2'}P_{X_1,X_2|Y_1,Y_2}(x_1,x_2'|y_1,y_2)}.
\end{align}

For example, Example \ref{ex1} satisfies the condition (A3).
We compare the cases satisfying conditions (A1) and (A2) with those satisfying conditions (A1) and (A3) under the $n$-qubit system.
In the $n$-qubit system,
the size $|{\cal Y}|$ in Example \ref{ex1} with the $n$-qubit system is $2^n+1$.
Next, we discuss the case satisfying the conditions (A1) and (A2) under
the $n$-qubit system.
In the $n$-qubit system,
the Pauli matrices
are given as $\{\sigma_{s_1}\otimes \cdots \otimes \sigma_{s_n}\}_{
s_j \in \{0,1,2,3\}}$.
When the identity matrix is excluded,
the number of Pauli matrices is $4^n-1$.
We choose the set ${\cal Y}$ to be the set of
Pauli matrices by excluding the identity matrix.
Then, the projective measurement is decided to be the spectral decomposition
of the corresponding Pauli matrix.
This case satisfies the conditions (A1) and (A2).
In this case, the size $|{\cal Y}|$ is $4^n-1$.
That is, the latter case has a much larger size $|{\cal Y}|$, i.e.,
the latter case requires substantially more experimental settings.
In fact, the condition (A2) has the requirement $|{\cal X}_{y(t)}|=2$,
which considerably increases the size $|{\cal Y}|$.

Under the condition (A3), we define
$\hat{C}_{1\to 2}[P_{X_1,Y_1,X_2,Y_2}]$ as
\begin{align}
&\hat{C}_{1\to 2}[P_{X_1,Y_1,X_2,Y_2}]\notag\\
:=&
\sum_{t_1\neq 0, t_2}
\sum_{x_1,x_2} g_{x_1|t_1} g_{x_2|t_2}
P_{X_2|X_1,Y_1,Y_2}(x_2|x_1,y(t_1),y(t_2))
H_{t_1}^T \otimes H_{t_2}\notag\\
&+
\sum_{t_2}
\Big(\sum_{x_2} g_{x_2|t_2}
P_{X_2|Y_1,Y_2}(x_2|0,y(t_2)) \notag\\
&-\sum_{t_1,t_1'\neq 0}
\sum_{x_1} g_{x_1|t_1}
P_{X_1|Y_1}(x_1|y(t_1))
 h_{t_1,t_1'}
\sum_{x_1',x_2} g_{x_1'|t_1'} g_{x_2|t_2}
P_{X_2|X_1,Y_1,Y_2}(x_2|x_1',y(t_1'),y(t_2)) \Big)
I \otimes H_{t_2},\Label{VBTR}
\end{align}
where
$ h_{t,t'} $ is defined as
\begin{align}
 h_{t,t'}:=\Tr H_{t} H_{t'}.
\end{align}
For distributions generated by
memoryless sequential strategies, the matrix $\hat{C}_{1\to 2}[P_{X_1,Y_1,X_2,Y_2}]$
recovers the underlying Choi matrix exactly.

\begin{theorem}\Label{TH8}
Assume the conditions (A1) and (A3).
Then,
$\hat{C}_{1\to 2}[{\cal M}[S_{1\to 2}(\rho_1,C_{1\to 2})]]$ equals
the Choi matrix $C_{1\to 2}$.
\end{theorem}

The combination of Theorems \ref{TH6} and \ref{TH8} guarantees that
the pair $(\rho_1[{\cal M}[S_{1\to 2}(\rho_1,C_{1\to 2})]]$,
$\hat{C}_{1\to 2}[{\cal M}[S_{1\to 2}(\rho_1,C_{1\to 2})]])$
recovers
the joint distribution ${\cal M}[S_{1\to 2}(\rho_1,C_{1\to 2})]$.

\begin{proof}
In this proof, we denote the distribution
${\cal M}[S_{1\to 2}(\rho_1,C_{1\to 2})$
by $\hat{P}_{X_1,Y_1,X_2,Y_2}$.
Theorem \ref{TH6} guarantees that
$\rho_1[{\cal M}[S_{1\to 2}(\rho_1,C_{1\to 2})]]=\rho_1$.
In order to show $\hat{C}_{1\to 2}[{\cal M}[S_{1\to 2}(\rho_1,C_{1\to 2})]]=C_{1\to 2}$,
it is sufficient to show
\begin{align}
\Tr (G_{t_1} \otimes G_{t_2} ) \hat{C}_{1\to 2}^{T_1}
=
\Tr (G_{t_1} \otimes G_{t_2} ) C_{1\to 2}^{T_1}\Label{VX7}
\end{align}
for $t_1,t_2 \in {\cal T}$.

For this aim, we denote the channel corresponding to $C_{1\to 2}$
by $\Lambda$.
Therefore, for
$y_1,y_2 \in {\cal Y}$,
and $x_1 \in {\cal X}_{y_1}$, $x_2 \in {\cal X}_{y_2}$,
the conditional distribution $\hat{P}_{X_2|X_1,Y_1,Y_2}(x_2|x_1,y_1,y_2)$
is given as
\begin{align}
&\hat{P}_{X_2|X_1,Y_1,Y_2}(x_2|x_1,y_1,y_2)
\notag \\
=&
\langle x_2,y_2|
\Lambda_{1\to 2} (|x_1,y_1\rangle\langle x_1,y_1|)|x_2,y_2\rangle
\Label{VBI1}.
\end{align}
Therefore, for $t_1,t_2 \in {\cal T}\setminus \{0\}$,
we have
\begin{align}
&
\Tr (G_{t_1} \otimes G_{t_2} ) \hat{C}_{1\to 2}^{T_1} \notag\\
=&\sum_{x_1,x_2} g_{x_1|t_1} g_{x_2|t_2}
\hat{P}_{X_2|X_1,Y_1,Y_2}(x_2|x_1,y(t_1),y(t_2)) \notag\\
=&\sum_{x_1}g_{x_1|t_1}
\Tr G_{t_2}\Lambda_{1\to 2}( |x_1,y(t_1)\rangle
\langle x_1,y(t_1)|) \notag\\
=&
\sum_{x_1}g_{x_1|t_1}
\Tr
((|x_1,y(t_1)\rangle \langle x_1,y(t_1)|) \otimes G_{t_2} )
C_{1\to 2}^{T_1} \notag\\
=&
\Tr (G_{t_1} \otimes G_{t_2} ) C_{1\to 2}^{T_1}.\Label{VX1}
\end{align}
When $t_2$ is $0$,
we replace $y(t_2)$
and $g_{x_2|t_2}$
by an arbitrary element $y \in {\cal Y}$ and $1$, respectively,
which guarantees the relation \eqref{VX1}.
Thus, we obtain \eqref{VX7} when $t_1 \neq 0$.

Since
$\rho_1=\frac{1}{d}I
+\sum_{t_1\neq 0}\langle G_{t_1} \rangle H_{t_1}$,
for $t_2 \in {\cal T}\setminus \{0\}$,
we have
\begin{align}
&\sum_{x_2} g_{x_2|t_2}
\hat{P}_{X_2|Y_1,Y_2}(x_2|0,y(t_2)) 
=\Tr (\rho_1 \otimes G_{t_2} ) C_{1\to 2}^{T_1}\notag \\
=&
\Tr \Big(\Big(\frac{1}{d}I +
\sum_{t_1\neq 0}
\langle G_{t_1} \rangle H_{t_1}
\Big) \otimes G_{t_2} \Big) C_{1\to 2}^{T_1}\notag \\
=&
\Tr \Big(
\Big(\frac{1}{d}I +
\sum_{t_1,t_1'\neq 0}
\langle G_{t_1} \rangle h_{t_1,t_1'} G_{t_1'}
\Big) \otimes G_{t_2} \Big)
C_{1\to 2}^{T_1}\notag \\
=&
\Tr \Big(\frac{1}{d}I \otimes G_{t_2} \Big) C_{1\to 2}^{T_1}
+\sum_{t_1,t_1'\neq 0}
\langle G_{t_1} \rangle h_{t_1,t_1'}
\Tr (G_{t_1'} \otimes G_{t_2} ) C_{1\to 2}^{T_1}\notag \\
=&
\Tr \Big(\frac{1}{d}I \otimes G_{t_2} \Big) C_{1\to 2}^{T_1}
+\sum_{t_1,t_1'\neq 0}
\langle G_{t_1} \rangle h_{t_1,t_1'}
\sum_{x_1,x_2} g_{x_1|t_1'} g_{x_2|t_2}
\hat{P}_{X_2|X_1,Y_1,Y_2}(x_2|x_1,y(t_1'),y(t_2)) \notag \\
=&
\Tr \Big(\frac{1}{d}I \otimes G_{t_2} \Big) C_{1\to 2}^{T_1}
+\sum_{t_1,t_1'\neq 0}
\sum_{x_1} g_{x_1|t_1}
\hat{P}_{X_1|Y_1}(x_1|y(t_1))
h_{t_1,t_1'} \sum_{x_1',x_2} g_{x_1'|t_1'} g_{x_2|t_2}
\hat{P}_{X_2|X_1,Y_1,Y_2}(x_2|x_1',y(t_1'),y(t_2)) ,
\end{align}
which implies
\begin{align}
&
\Tr (I \otimes G_{t_2} ) C_{1\to 2}^{T_1} \notag \\
=&
d\sum_{x_2} g_{x_2|t_2}
\hat{P}_{X_2|Y_1,Y_2}(x_2|0,y(t_2))
-d\sum_{t_1,t_1'\neq 0}
\sum_{x_1} g_{x_1|t_1}
\hat{P}_{X_1|Y_1}(x_1|y(t_1))
 \delta_{t_1,t_1'}
\sum_{x_1',x_2} g_{x_1'|t_1'} g_{x_2|t_2}
\hat{P}_{X_2|X_1,Y_1,Y_2}(x_2|x_1',y(t_1'),y(t_2))\notag  \\
=& \Tr (I \otimes G_{t_2} ) \hat{C}_{1\to 2}^{T_1}.
\Label{VH7}
\end{align}
When $t_2$ is $0$,
we replace $y(t_2)$
and $g_{x_2|t_2}$
by an arbitrary element $y \in {\cal Y}$ and $1$, respectively,
which guarantees the relation \eqref{VH7}.
Hence,
we obtain \eqref{VX7} when
$t_1 = 0$.
\end{proof}

Theorem~\ref{TH8} shows that, under (A1) and (A3), the reconstructed
matrix $\hat{C}_{1\to 2}[P_{X_1,Y_1,X_2,Y_2}]$ coincides with the true
Choi matrix for memoryless sequential strategies.
Together with the directional Markovianity established earlier, this
immediately yields the following necessary conditions on
$P_{X_1,Y_1,X_2,Y_2}$.

\begin{corollary}\Label{COR9}
Assume the conditions (A1) and (A3).
For $P_{X_1,Y_1,X_2,Y_2}\in {\cal M}[{\cal S}_{N|1\to 2}] $,
the relations $X_1-Y_1-Y_2$ and
$\hat{C}_{1\to 2}[P_{X_1,Y_1,X_2,Y_2}]\ge 0$
hold.
\end{corollary}

Corollary~\ref{COR9} therefore gives basic necessary conditions for
membership in ${\cal M}[{\cal S}_{N,1\to 2}]$: the directional Markovian
condition $X_1-Y_1-Y_2$ and the positivity of the reconstructed Choi matrix
$\hat{C}_{1\to 2}[P]$.

\subsection{Complete characterization of memoryless sequential dynamics}

The previous subsection showed that, under projective state reduction,
memoryless sequential distributions satisfy directional Markovianity together
with positivity of the reconstructed Choi matrix.
However, these conditions are still not sufficient to characterize
${\cal M}[{\cal S}_{N,1\to 2}]$ completely.
What is still missing is a consistency requirement ensuring that the observed
correlations arise from a single underlying channel.
We therefore introduce an additional algebraic condition for the joint
distribution $P_{X_1,Y_1,X_2,Y_2}$.
\begin{description}
\item[(C0)]
For every $y_1\in{\cal Y}$ and every $x_1\in{\cal X}_{y_1}$, the conditional
probability $P_{X_1\mid Y_1}(x_1\mid y_1)$ is strictly positive.
\item[(C1)]
The condition (C0) holds.
The relation
\begin{align}
&\sum_{x_1}
P_{X_2|X_1,Y_1,Y_2}(x_2|x_1,y_1,y_2) \notag\\
=&
d P_{X_2|Y_1,Y_2}(x_2|0,y_2)
-d\sum_{t_1,t_1'\neq 0}
\sum_{x_1} g_{x_1|t_1}
P_{X_1|Y_1}(x_1|y(t_1))
h_{t_1,t_1'}
\sum_{x_1'} g_{x_1'|t_1'}
P_{X_2|X_1,Y_1,Y_2}(x_2|x_1',y(t_1'),y_2) \Label{NM5}
\end{align}
holds for any $y_1,y_2 \in {\cal Y},
x_2 \in {\cal X}_{y_2}$.
\end{description}

Here it is important to distinguish the role of the assumptions
(A1), (A2), and (A3) from that of the new condition introduced above.
The assumptions (A1)--(A3) specify the operator and measurement structure
under which the reconstruction method is carried out.
By contrast, the condition (C1) is not an additional assumption on the
measurement setting, but a consistency requirement imposed directly on the
observed distribution $P_{X_1,Y_1,X_2,Y_2}$.

Condition (C1) should be understood as a compatibility requirement linking single-site and two-site statistics at the distribution level.
Operationally, it enforces that the observed joint probabilities can be generated by a single physical channel acting consistently on the post-measurement states produced by projective state reduction, rather than by different effective maps conditioned on hidden memory.

Importantly, (C1) does not introduce any additional physical assumption beyond memorylessness. Instead, it formalizes the fact that, under restricted projective measurements, Markovian conditional-independence relations and positivity constraints alone do not guarantee that all observed correlations originate from one underlying channel. Condition (C1) precisely excludes those distributions that satisfy all standard non-signaling and positivity requirements but nevertheless fail to admit a channel-level representation compatible with sequential, memoryless dynamics.

As shown below, condition~(C1) is automatically satisfied by all memoryless sequential strategies, whereas it need not hold in the larger quantum-memory class. More fundamentally, Theorem~\ref{TH10} shows that, within the class of distributions in ${\cal M}[{\cal S}_G]$ satisfying the directional Markovian condition $X_1-Y_1-Y_2$, condition~(C1) is equivalent to exact algebraic reproduction of the observed distribution by the reconstructed state and candidate Choi matrix. Consequently, any other condition on the observed distribution that guarantees such reproduction must imply condition~(C1). In this precise sense, condition~(C1) is the minimal additional algebraic consistency condition.

We define condition (C2) by exchanging the
roles of ${\cal H}_{I,1},{\cal H}_{O,1}$ and ${\cal H}_{I,2},{\cal H}_{O,2}$.
We compare the non-signaling condition
$Y_1-Y_2-X_2$ of the opposite direction, which is written as
\begin{align}
\sum_{x_1}P_{X_1|Y_1,Y_2}(x_1|y_1,y_2)
P_{X_2|X_1,Y_1,Y_2}(x_2|x_1,y_1,y_2)
=
P_{X_2|Y_1,Y_2}(x_2|0,y_2) \Label{NVO}
\end{align}
for any $y_1,y_2 \in {\cal Y}, x_2 \in {\cal X}_{y_2}$.
The following proposition shows that
any element $P_{X_1,Y_1,X_2,Y_2}\in{\cal M}[ {\cal S}_{N,1\to 2}]$
satisfies (C1) while it does not satisfy \eqref{NVO}
in general.

\begin{proposition}\Label{P1Auto}
Assume the conditions (A1) and (A3).
For Charlie's strategy $S\in {\cal S}_{N,1\to 2}$,
the distribution ${\cal M}[S]$ satisfies
the condition (C1).
\end{proposition}

\begin{proof}
We employ the same notation as the proof of Theorem \ref{TH8}.
For $y_1\in {\cal Y}\setminus \{0\}$,

we have
\begin{align}
&
\Tr (I \otimes G_{y_2,z_2} ) C_{1\to 2}^{T_1}
=\Tr G_{y_2,z_2}
\Lambda_{1\to 2}( I)
=
\sum_{x_1} \Tr G_{y_2,z_2}\Lambda_{1\to 2}(  |x_1,y_1\rangle \langle x_1,y_1|)\notag\\
\stackrel{(a)}{=} &
\sum_{x_1} \sum_{x_2} g_{x_2|y_2,z_2}
P_{X_2|X_1,Y_1,Y_2}(x_2|x_1,y_1,y_2),
\Label{BN1}
\\
&
\Tr (I \otimes I ) C_{1\to 2}^{T_1}
=\Tr
\Lambda_{1\to 2}( I)
=
\sum_{x_1} \Tr \Lambda_{1\to 2}(  |x_1,y_1\rangle \langle x_1,y_1|)\notag\\
\stackrel{(b)}{=} &
\sum_{x_1} \sum_{x_2}
P_{X_2|X_1,Y_1,Y_2}(x_2|x_1,y_1,y_2)
\Label{BN1T}
\end{align}
for any $y_2 \in {\cal Y}$ and $z_2 \in {\cal Z}_{y_2}$,
where $(a)$ and $(b)$ follow from \eqref{VBI1}.
The combination of
\eqref{VH7} with $t_2=(y_2,z_2)$
and \eqref{BN1} implies
\begin{align}
&\sum_{x_1} \sum_{x_2} g_{x_2|y_2,z_2}
P_{X_2|X_1,Y_1,Y_2}(x_2|x_1,y_1,y_2)
\notag \\
=&d\sum_{x_2} g_{x_2|y_2,z_2}
P_{X_2|Y_1,Y_2}(x_2|0,y_2)
-d\sum_{t_1,t_1'\neq 0}
\sum_{x_1} g_{x_1|t_1}
P_{X_1|Y_1}(x_1|y(t_1))
h_{t_1,t_1'}
\sum_{x_1',x_2} g_{x_1'|t_1'} g_{x_2|y_2,z_2}
P_{X_2|X_1,Y_1,Y_2}(x_2|x_1',y(t_1'),y_2).\Label{MV1}
\end{align}
Also, the combination of
\eqref{VH7} with $t_2=0$
and \eqref{BN1T} implies
\begin{align}
&\sum_{x_1} \sum_{x_2} P_{X_2|X_1,Y_1,Y_2}(x_2|x_1,y_1,y_2)
\notag \\
=&d\sum_{x_2}
P_{X_2|Y_1,Y_2}(x_2|0,y_2)
-d\sum_{t_1,t_1'\neq 0}
\sum_{x_1} g_{x_1|t_1}
P_{X_1|Y_1}(x_1|y(t_1))
h_{t_1,t_1'}
\sum_{x_1',x_2} g_{x_1'|t_1'}
P_{X_2|X_1,Y_1,Y_2}(x_2|x_1',y(t_1'),y_2).\Label{MV2}
\end{align}

We choose real numbers
$f_{x,y,z}$ such that
$\sum_{z \in {\cal Z}_{y}'}f_{z,y,x}g_{x'|y,z}=\delta_{x,x'}$.
Then, for $x_2$, we have
\begin{align}
&\sum_{x_1}
P_{X_2|X_1,Y_1,Y_2}(x_2|x_1,y_1,y_2) \notag\\
\stackrel{(a)}{=} & \sum_{z_2\in {\cal Z}_{y_2}'}f_{x_2,y_2,z_2}
\sum_{x_1} \sum_{x_2'} g_{x_2'|y_2,z_2}
P_{X_2|X_1,Y_1,Y_2}(x_2'|x_1,y_1,y_2)
\notag\\
\stackrel{(b)}{=} & \sum_{z_2\in {\cal Z}_{y_2}'}
f_{x_2,y_2,z_2}
\Big(
d\sum_{x_2'} g_{x_2'|y_2,z_2}
P_{X_2|Y_1,Y_2}(x_2'|0,y_2) \notag\\
&-d\sum_{t_1,t_1'\neq 0}
\sum_{x_1} g_{x_1|t_1}
P_{X_1|Y_1}(x_1|y(t_1))
h_{t_1,t_1'}
\sum_{x_1',x_2'} g_{x_1'|t_1'} g_{x_2'|y_2,z_2}
P_{X_2|X_1,Y_1,Y_2}(x_2'|x_1',y(t_1'),y_2)\Big)\notag \\
\stackrel{(c)}{=} &
d P_{X_2|Y_1,Y_2}(x_2|0,y_2)
-d\sum_{t_1,t_1'\neq 0}
\sum_{x_1} g_{x_1|t_1}
P_{X_1|Y_1}(x_1|y(t_1))
h_{t_1,t_1'}
\sum_{x_1'} g_{x_1'|t_1'}
P_{X_2|X_1,Y_1,Y_2}(x_2|x_1',y(t_1'),y_2) ,\Label{SK1}
\end{align}
where
$(a)$ and $(c)$ follow from the relation
$\sum_{z \in {\cal Z}_{y}'}f_{z,y,x}g_{x'|y,z}=\delta_{x,x'}$,
and $(b)$ follows from
the combination of \eqref{MV1} and \eqref{MV2}.
The relation \eqref{SK1} implies the condition (C1).
\end{proof}

We can now state the main result.
The previous subsection showed that any distribution generated by a
memoryless sequential strategy necessarily satisfies the directional
Markovian condition $X_1-Y_1-Y_2$ together with the positivity of the
reconstructed Choi matrix $\hat{C}_{1\to 2}[P]$.
Proposition~\ref{P1Auto} further showed that such distributions also satisfy
the additional algebraic condition (C1).
The following theorem shows that, under (A1) and (A3), these three conditions
are not only necessary but also sufficient, and hence give a complete
characterization of ${\cal M}[{\cal S}_{N,1\to 2}]$.

\begin{theorem}\Label{TH10}
Assume the conditions (A1), (A3), and (C0).
For a distribution $P \in {\cal M}[{\cal S}_G]$, the following conditions are equivalent:
\begin{description}
\item[(i)] $P \in {\cal M}[{\cal S}_{N,1\to 2}]$.
\item[(ii)] $P$ satisfies the Markovian condition $X_1-Y_1-Y_2$, the condition (C1),
and the positivity condition
\[
\hat{C}_{1\to 2}[P]\ge 0.
\]
\end{description}
Equivalently,
\begin{align}
{\cal M}[{\cal S}_{N,1\to 2}]
=
\{P \in {\cal M}[{\cal S}_G]\mid
P \text{ satisfies } X_1-Y_1-Y_2,\ (C1),\
\hat{C}_{1\to 2}[P]\ge 0
\}.
\Label{FG1}
\end{align}

Moreover, suppose that $P\in{\cal M}[{\cal S}_G]$ satisfies the Markovian condition $X_1-Y_1-Y_2$. Then condition~\emph{(C1)} holds if and only if the reconstructed pair $\bigl(\rho_1[P],\hat{C}_{1\to2}[P]\bigr)$ reproduces the full observed distribution in the sense that
\begin{align}
P_{X_1,Y_1,X_2,Y_2}
=
{\cal M}\!\left[
S_{1\to 2}\!\left(
\rho_1[P_{X_1,Y_1,X_2,Y_2}],
\hat{C}_{1\to 2}[P_{X_1,Y_1,X_2,Y_2}]
\right)
\right].
\Label{EE54}
\end{align}
Here, before positivity is imposed, the right-hand side of \eqref{EE54} denotes the array obtained by substituting the reconstructed state and candidate Choi matrix into the sequential Born formula; it is not asserted at this stage that $S_{1\to2}$ is a physical strategy. If, in addition, $\hat{C}_{1\to2}[P]\ge0$, then $S_{1\to2}\bigl(\rho_1[P],\hat{C}_{1\to2}[P]\bigr)$ is a physical memoryless sequential strategy that generates $P$.

Consequently, within the class of distributions in ${\cal M}[{\cal S}_G]$ satisfying $X_1-Y_1-Y_2$, condition~\emph{(C1)} is the weakest additional consistency condition, with respect to logical implication, that guarantees the reproduction property \eqref{EE54}.
\end{theorem}

\begin{proof}
Since Lemma \ref{LBM1A} guarantees 
\begin{align}
&
\{P \in {\cal M}[{\cal S}_G]|
P \hbox{ satisfies }
 X_1-Y_1-Y_2,
(C1),
\rho_1[P]
\ge 0,
\hat{C}_{1\to 2}[P]
\ge 0
\}\notag\\
=&
\{P \in {\cal M}[{\cal S}_G]|
P \hbox{ satisfies }
 X_1-Y_1-Y_2,
(C1),
\hat{C}_{1\to 2}[P]
\ge 0
\},
\end{align}
it is sufficient to show the relation
\begin{align}
{\cal M}[{\cal S}_{N,1\to 2}]=
\{P \in {\cal M}[{\cal S}_G]|
P \hbox{ satisfies }
 X_1-Y_1-Y_2,
(C1),
\rho_1[P]
\ge 0,
\hat{C}_{1\to 2}[P]
\ge 0
\}\Label{FG1B}.
\end{align}
Since the relation $\subset$ in
\eqref{FG1B} follows from Theorems
\ref{TH6}, \ref{TH8}, and Proposition~\ref{P1Auto},
it is sufficient to show the relation $\supset$ in
\eqref{FG1B}.
For this aim, it is sufficient to show the relation
\begin{align}
P_{X_1,Y_1,X_2,Y_2}=
{\cal M}[S_{1\to 2}(\rho_1[P_{X_1,Y_1,X_2,Y_2}],
C_{1\to 2}[P_{X_1,Y_1,X_2,Y_2}])]
\Label{DJ2}
\end{align}
for an element $P_{X_1,Y_1,X_2,Y_2}
\in
\{P \in {\cal M}[{\cal S}_G]|
P \hbox{ satisfies }
 X_1-Y_1-Y_2,
(C1),
\rho_1[P]
\ge 0,
\hat{C}_{1\to 2}[P]
\ge 0
\}$.

Since
$P_{X_1,Y_1,X_2,Y_2} \in {\cal M}[{\cal S}_G]$
satisfies the Markovian condition $ X_1-Y_1-Y_2$,
we have
\begin{align}
P_{X_1,X_2|Y_1,Y_2}(x_1,x_2|y_1,y_2)
=
P_{X_1|Y_1}(x_1|y_1)
P_{X_2|X_1,Y_1,Y_2}(x_1,x_2|y_1,y_2).
\end{align}
Hence, it is sufficient for \eqref{DJ2} to show that
\begin{align}
P_{X_1|Y_1}(x_1|y_1) &=\Tr |x_1,y_1\rangle \langle x_1,y_1|\rho_1[P_{X_1,Y_1,X_2,Y_2}]
\Label{SJ1}
\\
P_{X_2|X_1,Y_1,Y_2}(x_2|x_1,y_1,y_2)
&=\Tr
(|x_1,y_1\rangle \langle x_1,y_1|\otimes
|x_2,y_2\rangle \langle x_2,y_2|)
\hat{C}_{1\to 2}[P_{X_1,Y_1,X_2,Y_2}]^{T_1}
\Label{SJ2}
\end{align}
for $y_1,y_2\in{\cal Y}$, $x_2\in{\cal X}_{y_2}$, and
$x_1\in{\cal X}_{y_1}$.
Since $E_{x|y}$ is written as a linear sum of
$\{G_{y,z}\}_{z \in {\cal Z}_{y}'}$,
there are real numbers
$f_{x,y,z}$ such that
$\sum_{z \in {\cal Z}_{y}'}f_{z,y,x}g_{x'|y,z}=\delta_{x,x'}$.
Hence,
$|x_1,y_1\rangle \langle x_1,y_1|
=\sum_{z_1 \in {\cal Z}_{y_1}}f_{z_1,y_1,x_1} G_{y_1,z_1}$.
Thus, we have
\begin{align}
P_{X_1|Y_1}(x_1|y_1) \stackrel{(a)}{=} &
P_{X_1|Y_1,Y_2}(x_1|y_1,0)
=\sum_{z_1 \in {\cal Z}_{y_1}}f_{z_1,y_1,x_1}
\langle G_{y_1,z_1}\otimes I\rangle_{P_{X_1,Y_1,X_2,Y_2}} \notag\\
=&\sum_{z_1 \in {\cal Z}_{y_1}}f_{z_1,y_1,x_1}
\Tr G_{y_1,z_1}\rho_1[P_{X_1,Y_1,X_2,Y_2}]
=\Tr |x_1,y_1\rangle \langle x_1,y_1|\rho_1[P_{X_1,Y_1,X_2,Y_2}],
\end{align}
which implies \eqref{SJ1}.
Here, Step $(a)$ follows from
the Markovian condition $ X_1-Y_1-Y_2$
of $P_{X_1,Y_1,X_2,Y_2} \in {\cal M}[{\cal S}_G]$.

The definition of $\hat{C}_{1\to 2}[P_{X_1,Y_1,X_2,Y_2}]$
implies
\begin{align}
&\Tr (G_{t_1}\otimes G_{t_2})
\hat{C}_{1\to 2}[P_{X_1,Y_1,X_2,Y_2}]^{T_1}\notag\\
=&
\sum_{x_1,x_2} g_{x_1|t_1} g_{x_2|t_2}
P_{X_2|X_1,Y_1,Y_2}(x_2|x_1,y(t_1),y(t_2)) , \Label{XK1}\\
&\Tr (I \otimes G_{t_2})
\hat{C}_{1\to 2}[P_{X_1,Y_1,X_2,Y_2}]^{T_1}\notag\\
=&
d
\Big(\sum_{x_2} g_{x_2|t_2}
P_{X_2|Y_1,Y_2}(x_2|0,y(t_2)) \notag\\
&-\sum_{t_1,t_1'\neq 0}
\sum_{x_1} g_{x_1|t_1}
P_{X_1|Y_1}(x_1|y(t_1))
 h_{t_1,t_1'}
\sum_{x_1',x_2} g_{x_1'|t_1'} g_{x_2|t_2}
P_{X_2|X_1,Y_1,Y_2}(x_2|x_1',y(t_1'),y(t_2)) \Big)\Label{XK2}
\end{align}
for $t_1 \in {\cal T}\setminus \{0\}$ and $t_2 \in {\cal T}$.
The combination of the condition (C1) and \eqref{XK2}
implies
\begin{align}
\Tr (I \otimes G_{t_2})
\hat{C}_{1\to 2}[P_{X_1,Y_1,X_2,Y_2}]^{T_1}
=\sum_{x_2} g_{x_2|t_2}
\sum_{x_1}
P_{X_2|X_1,Y_1,Y_2}(x_2|x_1,y_1,y_2).
\Label{XK3}
\end{align}
For $z_1\in {\cal Z}_{y_1}'$
and $z_2\in {\cal Z}_{y_2}'$,
the combination of \eqref{XK1} and \eqref{XK3}
with $t_1=(y_1,z_1)$ and $t_2=(y_2,z_2)$
implies
\begin{align}
&
\sum_{x_1,x_2}
g_{x_1|y_1,z_1} g_{x_2|y_2,z_2}
\Tr
(|x_1,y_1\rangle \langle x_1,y_1|\otimes
|x_2,y_2\rangle \langle x_2,y_2|)
\hat{C}_{1\to 2}[P_{X_1,Y_1,X_2,Y_2}]^{T_1}\notag\\
=&
\sum_{x_1,x_2}
g_{x_1|y_1,z_1} g_{x_2|y_2,z_2}
P_{X_2|X_1,Y_1,Y_2}(x_2|x_1,y_1,y_2) \Label{XK4}.
\end{align}
Thus,
\begin{align}
&
\Tr
(|x_1,y_1\rangle \langle x_1,y_1|\otimes
|x_2,y_2\rangle \langle x_2,y_2|)
\hat{C}_{1\to 2}[P_{X_1,Y_1,X_2,Y_2}]^{T_1}\notag\\
=&
\sum_{z_1,z_2}
f_{x_1,y_1,z_1}
f_{x_2,y_2,z_2}
\sum_{x_1',x_2'}
g_{x_1'|y_1,z_1} g_{x_2'|y_2,z_2}
\Tr
(|x_1',y_1\rangle \langle x_1',y_1|\otimes
|x_2',y_2\rangle \langle x_2',y_2|)
\hat{C}_{1\to 2}[P_{X_1,Y_1,X_2,Y_2}]^{T_1}\notag\\
=&
\sum_{z_1,z_2}
f_{x_1,y_1,z_1}
f_{x_2,y_2,z_2}
\sum_{x_1',x_2'}
g_{x_1'|y_1,z_1} g_{x_2'|y_1,z_1}
P_{X_2|X_1,Y_1,Y_2}(x_2'|x_1',y_1,y_2) \notag\\
=&
P_{X_2|X_1,Y_1,Y_2}(x_2|x_1,y_1,y_2) \Label{XK5}.
\end{align}
Hence, we obtain \eqref{SJ2}.

We finally prove the converse implication in the moreover statement. Suppose that $P\in{\cal M}[{\cal S}_G]$ satisfies the Markovian condition $X_1-Y_1-Y_2$ and that the reconstructed pair reproduces $P$ in the sense of \eqref{EE54}. Then
\begin{align}
P_{X_2\mid X_1,Y_1,Y_2}(x_2\mid x_1,y_1,y_2)
&=\Tr\!\Bigl[\bigl(|x_1,y_1\rangle\langle x_1,y_1|\otimes|x_2,y_2\rangle\langle x_2,y_2|\bigr)\hat{C}_{1\to2}[P]^{T_1}\Bigr].
\Label{SJ2C}
\end{align}
Summing \eqref{SJ2C} over $x_1$ and using
\begin{align}
\sum_{x_1\in{\cal X}_{y_1}}|x_1,y_1\rangle\langle x_1,y_1|=I,
\end{align}
we obtain
\begin{align}
\sum_{x_1}P_{X_2\mid X_1,Y_1,Y_2}(x_2\mid x_1,y_1,y_2)
&=\Tr\!\Bigl[\bigl(I\otimes|x_2,y_2\rangle\langle x_2,y_2|\bigr)\hat{C}_{1\to2}[P]^{T_1}\Bigr].
\Label{XK5C}
\end{align}
On the other hand, the definition \eqref{VBTR} of $\hat{C}_{1\to2}[P]$, equivalently \eqref{XK2}, together with the expansion of $|x_2,y_2\rangle\langle x_2,y_2|$ in the basis $\{G_{y_2,z_2}\}_{z_2\in{\cal Z}_{y_2}'}$, shows that the right-hand side of \eqref{XK5C} is equal to
\begin{align}
&d P_{X_2\mid Y_1,Y_2}(x_2\mid0,y_2)\notag\\
&\quad-d\sum_{t_1,t_1'\neq0}\sum_{x_1}g_{x_1\mid t_1}P_{X_1\mid Y_1}(x_1\mid y(t_1))h_{t_1,t_1'}\sum_{x_1'}g_{x_1'\mid t_1'}P_{X_2\mid X_1,Y_1,Y_2}(x_2\mid x_1',y(t_1'),y_2).
\Label{XK6C}
\end{align}
Combining \eqref{XK5C} and \eqref{XK6C} yields \eqref{NM5}. Since condition~\emph{(C0)} is assumed, condition~\emph{(C1)} follows. This proves the converse implication and completes the proof of the equivalence in the moreover statement.
\end{proof}

Theorem~\ref{TH10} separates algebraic consistency from physical realizability. The directional Markovian condition fixes the relevant one-way non-signaling structure. Within the class of distributions satisfying this condition, condition~(C1) is equivalent to exact reproduction of the observed distribution by the reconstructed state and candidate Choi matrix. It follows that any alternative condition sufficient for such reproduction must imply condition~(C1). Thus condition~(C1) is minimal with respect to logical implication. Positivity of the reconstructed Choi matrix is the additional requirement that promotes the candidate map to a physical channel. The examples in Subsection~\ref{S9B} separately show that positivity and condition~(C1) cannot replace each other.
By summarizing \eqref{ZL1}, 
\eqref{CWSI},
Theorems \ref{TH4I}, and
\ref{TH8},
the conditions for the classes in the hierarchies \eqref{HI1}
are characterized as follows.
\begin{align}
\begin{array}{ccccc}
X_1-Y_1-Y_2~\&~ \hat{C}_{1\to2}\ge0~\&~(C1) & \supset &X_1-Y_1-Y_2-X_2 &\subset &
Y_1-Y_2-X_2~\&~ \hat{C}_{2\to1}\ge0 ~\&~(C2) \\
\cap &  &\cap & & \cap  \\
X_1-Y_1-Y_2 & \supset &X_1-Y_1-Y_2 ~\& ~Y_1-Y_2-X_2
~\&~ R_{1,2}\ge 0
&\subset & Y_1-Y_2-X_2
\end{array}\Label{HI16}
\end{align}

However, the conditions given in \eqref{HI16}
show only sufficient conditions for the conditions in \eqref{HI1}.

\begin{remark}
The full-support assumption \emph{(C0)} is imposed only to ensure that the
conditional probabilities appearing in \emph{(C1)} are well defined.
When \emph{(C0)} fails, the observed distribution necessarily assigns zero
probability to some rank-one projective outcome.
Under \emph{(A3)}, this means that the relevant state is supported on the
orthogonal complement of the corresponding projector, so that the analysis may
be reformulated from the outset on the smaller support subspace.
After this compression of the input space and removal of the impossible
outcomes, the same reconstruction argument applies to the reduced model.
Thus Theorem~\ref{TH10} should be understood as the nondegenerate
(full-support) version of the general support-restricted statement.
\end{remark}

\section{Distinguishability between ${\cal M}[{\cal S}_{N,1\to 2}]$ 
and ${\cal M}[{\cal S}_{N,2\to 1}]$ under qubit-Pauli matrices}\Label{S7}
We next specialize to qubit systems with Pauli projective measurements.  In
this section the observed distribution is assumed to arise from a forward
memoryless sequential strategy, and the question is whether the same
distribution also admits a reverse memoryless sequential representation.  This
restricted reverse problem is the order-indistinguishability problem.  The
assumed forward representation makes it possible to express the reverse
conditions directly in terms of the input state and channel parameters.

\subsection{Formulation for two qubits}\Label{S7-1}
We now specialize the general framework to the minimal nontrivial setting in which the order-indistinguishability problem can be analyzed explicitly. 
The aim of this section is not to solve the full reverse-direction membership problem for arbitrary observed distributions. 
Instead, we restrict attention to distributions already known to arise from a forward memoryless sequential strategy. 
Starting from a forward memoryless sequential strategy 
$S=S_{1\to 2}(\rho_1,C_{1\to 2}) \in {\cal S}_{N,1\to 2}$, 
we ask whether the resulting observed distribution ${\cal M}[S]$ can also be explained by a memoryless sequential strategy of the opposite direction. 
Equivalently, we ask when ${\cal M}[S]$ belongs to the intersection 
${\cal M}[{\cal S}_{N,1\to 2}] \cap {\cal M}[{\cal S}_{N,2\to 1}]$. 
This is the sense in which we use the term order-indistinguishability in the present section.

Similar to \cite{Liu2}, as a special case of Section~\ref{S5-2}, we consider
the case in which
${\cal H}_{I,1}$, ${\cal H}_{O,1}$, ${\cal H}_{I,2}$, and ${\cal H}_{O,2}$
are qubit systems, and the measurements are given by projective measurements
of the Pauli matrices
\begin{align}
\sigma_0&:=
\left(
\begin{array}{cc}
1 & 0 \\
0 & 1
\end{array}
\right),~
\sigma_1:=
\left(
\begin{array}{cc}
0 & 1 \\
1 & 0
\end{array}
\right), \\
\sigma_2&:=
\left(
\begin{array}{cc}
0 & -\sqrt{-1} \\
\sqrt{-1} & 0
\end{array}
\right),~
\sigma_3:=
\left(
\begin{array}{cc}
1 & 0 \\
0 & -1
\end{array}
\right).
\end{align}
That is, the state reduction on ${\cal H}_{(I,i)}$ to ${\cal H}_{(O,i)}$ is
given by the projection hypothesis for Pauli measurements.

Both Alice and Bob independently choose one of these Pauli measurements at
random.
Hence, $Y_1$ and $Y_2$ take values in $\{0,1,2,3\}$ and are uniformly
distributed, while $X_1$ and $X_2$ record the corresponding outcomes.
When $\sigma_0$ is chosen, there is no nontrivial measurement outcome.
Thus, the pair $(X_i,Y_i)$ takes one of seven possible values for each
$i=1,2$.

This measurement scheme is genuinely restricted and is not tomographically
complete.
In the present qubit-Pauli setting, the family
$\{C[\Gamma_{1,z_1}]\otimes C[\Gamma_{2,z_2}]\}_{z_1,z_2}$
contains at most $49$ linearly independent elements.
On the other hand, Appendix~\ref{app:qubit-counting} shows that the relevant
quotient space has dimension $88$.
Hence these tensors cannot span the quotient space, and exact process-matrix
reconstruction is impossible in this setting.
Accordingly, the underlying process matrix cannot in general be reconstructed
from the observed distribution.
The role of the present two-qubit analysis is therefore not process-level
reconstruction, but the explicit study of order distinguishability directly at
the level of the observed statistics.

At the same time, this Pauli setting satisfies the assumptions (A1), (A2),
and (A3).
In particular, the reconstructed quantities
$\hat{C}_{1\to 2}$ and $\tilde{C}_{1\to 2}$ coincide, so that they can be
used as concrete distribution-level tools in the analysis below.

We now turn to the central question of this section.
Starting from a forward memoryless sequential strategy
$S=S_{1\to 2}(\rho_1,C_{1\to 2}) \in {\cal S}_{N,1\to 2}$,
we ask whether the resulting observed distribution ${\cal M}[S]$ can also be
explained by a memoryless sequential strategy of the opposite direction.
Equivalently, we ask when ${\cal M}[S]$ belongs to the intersection
${\cal M}[{\cal S}_{N,1\to 2}] \cap {\cal M}[{\cal S}_{N,2\to 1}]$.

Set $P={\cal M}[S]$.  Corollary~\ref{COR9B} asks whether there is a
compatible reverse-direction state--Choi pair that reproduces $P$, with a
positive Choi matrix.  The compatible second marginal is fixed by the
observed Pauli statistics and equals $\rho_2[P]$.  We call this marginal
\emph{full rank} when $\rho_2[P]>0$ and \emph{singular} when it is not
invertible.  For a qubit state, the singular case is exactly the pure-state
case.  In the full-rank case, the compatibility relation~\eqref{XSA2}
uniquely determines the Hermitian, trace-preserving reverse Choi candidate
$\tilde C_{2\to1}[P]$.  Thus the existence question in
Corollary~\ref{COR9B} reduces to checking this candidate: the observed
forward distribution must obey the reverse-direction non-signaling condition,
the candidate's statistics must obey the forward-direction condition, and
the candidate must be positive.  We label these three checks as follows.
\begin{description}
\item[(D1)] $P$ satisfies $Y_1-Y_2-X_2$.
\item[(D2)] The array obtained from the reverse sequential Born formula
with $\rho_2[P]$ and $\tilde C_{2\to1}[P]$ satisfies
$X_1-Y_1-Y_2$.
\item[(D3)] $\tilde C_{2\to1}[P]\ge 0$.
\end{description}
Before (D3) is imposed, the array in (D2) is only an algebraic candidate;
we do not assume that it is generated by a physical reverse strategy.
The candidate is compatible with $P$ by construction.  Its
reverse-direction non-signaling condition follows from its trace-preserving
normalization; (D1) and (D2) supply the remaining directional conditions.
Once (D3) holds, the candidate is a physical reverse channel;
Lemma~\ref{LL8} then identifies its statistics with $P$.  Conversely, any physical reverse
memoryless realization must equal the unique compatible candidate and hence
satisfy (D1)--(D3).  This proves the claimed equivalence in the full-rank
case.

If $\rho_2[P]$ is singular, the unique reconstruction~\eqref{XSA2}
is unavailable, so (D2) and (D3) are not applied to it.  Instead,
Lemma~\ref{LBU2} shows that (D1) makes $P$ a product distribution in
this case.  A reverse replacement channel then realizes $P$ directly.
Thus the singular case is decided by (D1), without assigning a value to
$\tilde C_{2\to1}[P]$.

We first analyze (D1)--(D3) for a general qubit channel and then
specialize to Pauli, phase-damping, and depolarizing channels.  The
parameter used to organize the general analysis is defined with the
forward-channel representation in the next subsection.

\subsection{General qubit channel}\Label{S7-3}
We first study a general qubit channel before specializing to particular
families.
The point of this subsection is to translate the abstract reverse-compatibility
conditions (D1)--(D3) into explicit constraints on the parameters of a forward
memoryless sequential strategy
$S=S_{1\to 2}(\rho_1,C_{1\to 2})$.
This family-independent analysis isolates the geometric structure behind the
order-indistinguishable region and explains why the later channel-specific
formulas take the form they do.

In this section, we employ the formulas
\begin{align}
I \circ \sigma_i=\sigma_i, \quad \sigma_j \circ \sigma_i=\delta_{i,j} I
\Label{CZER}
\end{align}
for $i,j=1,2,3$.
We parameterize a forward memoryless sequential strategy
$S=S_{1\to 2}(\rho_1, C_{1\to 2})\in {\cal S}_{N,1\to 2}$ as follows.
The initial state $\rho_1$ is written as
\begin{align}
\rho_1=\frac{1}{2}(I+ \sum_{i=1}^3 c^1_i \sigma_i)
\end{align}
with the condition $\sum_{i=1}^3 (c^1_i)^2\le 1$.
For the Choi matrix $C_{1\to 2}$ of a qubit TP-CP map
$\Lambda_{1\to 2}$, let $t\in\{0,1,2,3\}$ be the rank of the
$3\times3$ coefficient matrix of the traceless--traceless Pauli part of
$C_{1\to 2}^{T_1}$.  Choose a singular-value decomposition of this part,
indexing its positive singular values and corresponding operator pairs by
$j=1,\ldots,t$.  In particular, when $t=2$ the pairs are only
$(\alpha_1,\beta_1)$ and $(\alpha_2,\beta_2)$; no
$\alpha_3$ or $\beta_3$ is included.  With
$\rho_2:=\Lambda_{1\to 2}(I/2)$, this gives~\cite{Choi1975}
\begin{align}
C_{1\to 2}^{T_1}= I \otimes \rho_2+ \sum_{j=1}^t \frac{\lambda_j}{2}
\alpha_j \otimes \beta_j,\Label{HJ8}
\end{align}
where, for $1\le j\le t$, $1\ge\lambda_j>0$ and $\alpha_j,\beta_j$
are traceless Hermitian matrices of operator norm one.  The families
$\{\alpha_j\}_{j=1}^t$ and $\{\beta_j\}_{j=1}^t$ are each
orthogonal; no operators with indices $j>t$ are specified.
Here $\alpha_j=v_j\cdot\sigma$ and $\beta_j=w_j\cdot\sigma$,
where $\|v_j\|_2=\|w_j\|_2=1$.  Thus
$\frac12\Tr(\alpha_j\alpha_k)=\frac12\Tr(\beta_j\beta_k)=\delta_{jk}$
for $1\le j,k\le t$, and their Hilbert--Schmidt norms are $\sqrt2$.
The coefficient matrix just mentioned is
$\sum_{j=1}^t\lambda_jv_jw_j^{\mathsf T}$; its rank is $t$,
which need not equal the rank of the full Choi matrix.  When $t=0$,
the sum in \eqref{HJ8} is empty.  This $t$ is distinct from the
measurement indices $t_i\in{\cal T}_i$ used earlier.
Here, the qubit state $\rho_2$ appearing in \eqref{HJ8} can be written in
Bloch form as
\begin{align}
\rho_2=\frac{1}{2}(I+ \sum_{i=1}^3 c^2_i \sigma_i).
\end{align}
The cases $t=0,1,2,3$ will determine the form of the feasible
region below.

We begin with condition (D1), which asks whether the observed distribution of
the forward strategy is compatible with the non-signaling structure required
from the reverse direction.
The following lemma shows that this already imposes strong restrictions on the
input Bloch vector $\rho_1$, and that the form of these restrictions depends
only on the rank parameter $t\in\{0,1,2,3\}$.

\begin{lemma}\Label{LBU1}
$S_{1\to 2}(\rho_1, C_{1\to 2})$ satisfies the condition (D1)
if and only if the pair satisfies one of the following four cases depending on $t$:
\begin{description}
\item[(1)] $t=3$: $c^1_i=0$ for $i=1,2,3$.
\item[(2)] $t=2$: for every $i=1,2,3$, either $c_i^1=0$ or
$\Tr(\alpha_1\sigma_i)=\Tr(\alpha_2\sigma_i)=0$.  In particular, at least two
indices $i$ satisfy $c_i^1=0$.
\item[(3)] $t=1$: for every $i=1,2,3$, either $c_i^1=0$ or
$\Tr(\alpha_1\sigma_i)=0$.  In particular, at least one index $i$ satisfies
$c_i^1=0$.
\item[(4)] $t=0$: no further condition is required.
\end{description}
\end{lemma}
\noindent\emph{Proof deferred to Appendix~\ref{app:proof-LBU1}.}

Once (D1) is imposed, the reconstructed quantities simplify considerably.
In particular, the next lemma identifies the reconstructed marginal state
$\rho_2[{\cal M}[S]]$ and the reconstructed bipartite operator
${\cal R}[{\cal M}[S]]$ explicitly.
It also translates condition (D2) into the corresponding restrictions on the
output Bloch vector $\rho_2$.
Thus, Lemma~\ref{LBU2} should be read as the companion to
Lemma~\ref{LBU1}: the first lemma resolves the input-side restriction coming
from (D1), while the second resolves the output-side restriction coming from
(D2).

\begin{lemma}\Label{LBU2}
Assume that $S=S_{1\to 2}(\rho_1,C_{1\to 2})$ satisfies (D1).
Part (ii) applies when $\rho_2[{\cal M}[S]]>0$; part (i) also holds on
the singular boundary.
\begin{enumerate}[(i)]
\item We have
\begin{align}
\rho_{2}[{\cal M}[S]]&= \rho_2 \Label{DSJK2}\\
{\cal R}[{\cal M}[S]]&= \rho_1 \otimes \rho_2 + \sum_{j=1}^t \frac{\lambda_j}{4}
\alpha_j \otimes \beta_j . \Label{DSJK}
\end{align}
\item $S$ satisfies the condition (D2)
if and only if the pair $(\rho_1, C_{1\to 2})$ satisfies one of the following four cases:
\begin{description}
\item[(1)] $t=3$: $c^2_i=0$ for $i=1,2,3$.
\item[(2)] $t=2$: for every $i=1,2,3$, either $c_i^2=0$ or
$\Tr(\beta_1\sigma_i)=\Tr(\beta_2\sigma_i)=0$.  In particular, at least two
indices $i$ satisfy $c_i^2=0$.
\item[(3)] $t=1$: for every $i=1,2,3$, either $c_i^2=0$ or
$\Tr(\beta_1\sigma_i)=0$.  In particular, at least one index $i$ satisfies
$c_i^2=0$.
\item[(4)] $t=0$: no further condition is required.
\end{description}
\end{enumerate}
If $\rho_2[{\cal M}[S]]$ is singular, (D1) implies $t=0$ and
${\cal R}[{\cal M}[S]]=\rho_1\otimes\rho_2$; (D2) is not needed.
\end{lemma}
\noindent\emph{Proof deferred to Appendix~\ref{app:proof-LBU2}.}

In the full-rank case, (D1) and (D2) constrain the local states.  On the
singular boundary, (D1) instead gives the product case $t=0$ directly.
We now turn to positivity.
We first examine the positivity of the original forward Choi matrix
$C_{1\to 2}$ itself.
The following lemma gives the corresponding positivity region under the
constraints already imposed by (D1) and (D2).

\begin{lemma}\Label{LBU0}
Let $S=S_{1\to2}(\rho_1,C_{1\to2})$ have representation~\eqref{HJ8}.
For $t\ge1$, assume (D1), $\rho_2[\mathcal M[S]]>0$, and (D2).
For $t=0$, neither (D2) nor full rank is required.
The relation $C_{1\to 2}\ge 0$ is equivalent to the applicable condition below.
\begin{description}
\item[(1)] $t=3$.
There are two local-unitary normal forms.
\begin{description}
\item[(A)] The positivity condition is equivalent to
$1\ge \max(-\lambda_1 +\lambda_2+\lambda_3,
\lambda_1 -\lambda_2+\lambda_3,
\lambda_1 +\lambda_2-\lambda_3)$.
\item[(B)] The positivity condition is equivalent to
$1\ge \lambda_1+\lambda_2+\lambda_3$.
\end{description}
\item[(2)] $t=2$: $1\ge (\lambda_1+\lambda_2)^2+
\sum_{i=1}^3(c_i^2)^2$.
\item[(3)] $t=1$: $1\ge \lambda_1^2+
\sum_{i=1}^3(c_i^2)^2$.
\item[(4)] $t=0$: no further condition is imposed.
\end{description}
\end{lemma}
\noindent\emph{Proof deferred to Appendix~\ref{app:proof-LBU0-LBU3}.}

The final step is to translate condition (D3), namely the complete positivity
of the reconstructed reverse-direction channel, into explicit inequalities.
The next lemma first writes the reconstructed reverse-direction Choi matrix
$\tilde{C}_{2\to 1}[{\cal M}[S]]$ in closed form and then gives the positivity
criterion for this matrix.
This is the key lemma for the order-indistinguishability problem, because (D3)
is the only genuinely channel-level constraint among (D1)--(D3).

\begin{lemma}\Label{LBU3}
Let $S=S_{1\to2}(\rho_1,C_{1\to2})$ have representation~\eqref{HJ8}
and satisfy (D1).  For $t\ge1$, assume also
$\rho_2[\mathcal M[S]]>0$ and (D2).  For $t=0$ with a full-rank
second marginal, (D2) holds automatically.  If the second marginal is
singular, (D2) and (D3) are not defined via the unique reverse
reconstruction; the product distribution instead has the explicit
reverse replacement channel $\Lambda_{2\to1}(\omega)=\rho_1\Tr\omega$.
\begin{enumerate}[(i)]
\item When $\rho_2[\mathcal M[S]]>0$, the unique reverse reconstruction
satisfies
\begin{align}
\tilde{C}_{2\to 1}[{\cal M}[S]]^{T_2}
= \rho_1 \otimes I + \sum_{j=1}^t \frac{\lambda_j}{2} \alpha_j \otimes \beta_j
\Label{VNI1}
\end{align}
For a singular second marginal, necessarily $t=0$; the same right-hand
side $\rho_1\otimes I$ is the Choi matrix of an explicitly chosen
reverse replacement channel, not a uniquely reconstructed matrix.
\item When $\rho_2[\mathcal M[S]]>0$, the condition
$\tilde{C}_{2\to 1}[{\cal M}[S]]\ge 0$ is equivalent to the
applicable condition below.  On the singular boundary the explicit
reverse replacement Choi matrix $\rho_1\otimes I$ is positive.
\begin{description}
\item[(1)] $t=3$.
There are two local-unitary normal forms.
\begin{description}
\item[(A)] The positivity condition is equivalent to
$1\ge \max(-\lambda_1 +\lambda_2+\lambda_3,
\lambda_1 -\lambda_2+\lambda_3,
\lambda_1 +\lambda_2-\lambda_3)$.
\item[(B)] The positivity condition is equivalent to
$1\ge \lambda_1+\lambda_2+\lambda_3$.
\end{description}
\item[(2)] $t=2$: $1\ge (\lambda_1+\lambda_2)^2+
\sum_{i=1}^3(c_i^1)^2$.
\item[(3)] $t=1$: $1\ge \lambda_1^2+
\sum_{i=1}^3(c_i^1)^2$.
\item[(4)] $t=0$: no further condition is imposed.
\end{description}
\end{enumerate}
\end{lemma}
\noindent\emph{Proof deferred to Appendix~\ref{app:proof-LBU0-LBU3}.}

Taken together, Lemmas~\ref{LBU1}--\ref{LBU3} reduce the reverse-compatibility
problem for a general forward qubit strategy
$S_{1\to 2}(\rho_1,C_{1\to 2})$
to explicit constraints on a small number of geometric parameters.
In this sense, the general qubit analysis isolates the common structure behind
the order-indistinguishable region before any channel family is specified.
We now specialize to representative families---Pauli, phase-damping, and
depolarizing channels---for which these abstract constraints become concrete
and easily interpretable inequalities.

\subsection{Specific channels}

\subsubsection{Pauli channel}\Label{BASS}

We first consider the full Pauli family before passing to the more transparent
subfamilies treated later.
The point of this subsection is to make the reverse-compatibility conditions
(D1)--(D3) explicit in a canonical and analytically tractable channel family.
In particular, this family already contains the depolarizing channel as a
special case.

We consider the Pauli channel
\begin{align}
\Lambda(\rho):=
\sum_{j=0}^3 p_j \sigma_j \rho \sigma_j .
\end{align}
In particular, the depolarizing channel
\[
\Lambda_\lambda(\rho):=
(1-\lambda)\rho+\lambda \frac{I}{2}
\]
is obtained as the special case
$p_0=1-\frac{3}{4}\lambda$ and
$p_j=\frac{1}{4}\lambda$ for $j=1,2,3$.
Its Choi matrix can be written as
\begin{align}
C_{1\to 2}^{T_1}
=
I \otimes \frac{I}{2}
+\sum_{j=1}^3 \frac{\lambda_j}{2}\sigma_j\otimes\sigma_j,
\Label{HJ51}
\end{align}
where
$\lambda_j:=2p_0+2p_j-1$ for $j=1,3$ and
$\lambda_2:=-2p_0-2p_2+1$.

We fix the initial state in the form
\[
\rho_1=\frac{1}{2}(I+\kappa\sigma_1).
\]

\begin{lemma}\Label{LPauli}
For the Pauli channel with Choi matrix \eqref{HJ51} and initial state
$\rho_1=\frac12(I+\kappa\sigma_1)$, assume $|\kappa\lambda_1|<1$
when evaluating (D2) or (D3). The reverse-compatibility conditions
(D1)--(D3) take the following form:
\begin{enumerate}[(i)]
\item
(D1) holds if and only if
\[
\kappa\lambda_1=0,
\]
that is, $\lambda_1=0$ or $\kappa=0$.

\item
(D2) holds if and only if
\[
\lambda_1=0,
\qquad\text{or}\qquad
\kappa=0,
\qquad\text{or}\qquad
|\kappa|=1 \ \text{ and }\ |\lambda_1|<1.
\]

\item
For $|\kappa|<1$, condition (D3) is equivalent to
\begin{align}
1+ \frac{\kappa(\lambda_1^2 -1)}{\kappa^2\lambda_1^2-1}
+\frac{ \lambda_1 (\kappa^2 -1)}{\kappa^2\lambda_1^2-1}
&\ge 0,\Label{CBQ1-main}\\
1- \frac{\kappa(\lambda_1^2 -1)}{\kappa^2\lambda_1^2-1}
-\frac{ \lambda_1 (\kappa^2 -1)}{\kappa^2\lambda_1^2-1}
&\ge 0,\Label{CBQ2-main}\\
\left(\frac{(\kappa^2  \lambda_1 -1)( \lambda_1 -1)}{\kappa^2\lambda_1^2-1}\right)^2
=
\left(1-
\frac{ \lambda_1 (\kappa^2 -1)}{\kappa^2\lambda_1^2-1}
\right)^2
& \ge
\left(\frac{\kappa(\lambda_1^2 -1)}{\kappa^2\lambda_1^2-1}\right)^2
+(\lambda_2+\lambda_3)^2 ,
\Label{CBV1-main}\\
\left(\frac{(\kappa^2  \lambda_1 -1)( \lambda_1 +1)}{\kappa^2\lambda_1^2-1}\right)^2
=
\left(1+
\frac{ \lambda_1 (\kappa^2 -1)}{\kappa^2\lambda_1^2-1}
\right)^2
& \ge
\left(\frac{\kappa(\lambda_1^2 -1)}{\kappa^2\lambda_1^2-1}\right)^2
+(\lambda_2-\lambda_3)^2 .
\Label{CBV2-main}
\end{align}

\item
For $|\kappa|=1$ and $|\lambda_1|<1$, condition (D3) is equivalent to
\[
\lambda_2=\lambda_3=0.
\]
\end{enumerate}
\end{lemma}
\noindent\emph{Proof deferred to Appendix~\ref{app:pauli-family-calculations}.}

Lemma~\ref{LPauli} gives the generic Pauli-family form of the
reverse-compatibility conditions.
The first two conditions constrain the directional parameters
$(\kappa,\lambda_1)$, whereas the third is the complete-positivity condition
for the reconstructed reverse-direction channel.
Thus, the Pauli family already exhibits the basic geometry of the
order-indistinguishable region in a concrete form.
The next subsections extract especially transparent representative examples
from this family, namely the phase-damping and depolarizing channels.

\subsubsection{Phase damping channel with $\sigma_1$-basis}

We now consider a particularly simple representative of the Pauli family,
namely the phase-damping channel in the $\sigma_1$ basis.
This example is useful because, unlike the general Pauli case, the directional
conditions (D1) and (D2) collapse to the same simple restriction on the input
state, so that the remaining nontrivial issue is the reverse-direction
complete-positivity condition (D3).

We fix the initial state to be
$\rho_1=\frac{1}{2}(I+\kappa\sigma_1)$.
We consider the phase-damping channel with respect to the $\sigma_1$ basis,
whose Choi matrix $C_{1\to 2}$ satisfies the condition $\lambda_1=1$
under the form \eqref{HJ8}.
Then the directional conditions simplify as follows:
\begin{itemize}
\item (D1) holds if and only if $\kappa=0$.
\item On $|\kappa|<1$, (D2) holds if and only if $\kappa=0$.
\end{itemize}
At $|\kappa|=1$ the second marginal is pure: (D2)--(D3) are
undefined, while (D1) fails.  On the full-rank domain, (D3) is as follows.

\begin{lemma}\Label{LPhaseDampSigmaOne}
In the $\sigma_1$-basis phase-damping family, suppose $|\kappa|<1$.
Then (D3) holds if and only if $\lambda_2=-\lambda_3$ and
$|\lambda_3|\le1$.  At $|\kappa|=1$, (D3) is undefined.
\end{lemma}
\noindent\emph{Proof deferred to Appendix~\ref{app:phase-damping-sigma1}.}

Lemma~\ref{LPhaseDampSigmaOne} determines the reverse-direction positivity condition (D3) in the
present channel family.
On the full-rank domain, the directional conditions (D1) and (D2) are
both equivalent to $\kappa=0$; on the singular boundary (D1) fails.
Consequently, the full reverse-compatibility conditions
(D1)--(D3) can hold simultaneously only when $\kappa=0$.
Thus, in the $\sigma_1$-basis phase-damping family, the order-indistinguishable
region is governed entirely by the directional constraints, while the above
lemma shows that the CP-feasibility condition (D3) is consistent with that
same restriction.

\subsubsection{Phase damping channel with $\sigma_3$-basis}

We next consider the phase-damping channel in the $\sigma_3$ basis, under the
same fixed initial state
$\rho_1=\frac{1}{2}(I+\kappa\sigma_1)$.
Unlike the preceding $\sigma_1$-basis family, the directional
conditions here depend on the parameter $\lambda_1$.  The
reverse-direction positivity condition (D3) also takes a different form.

More precisely, we consider the phase-damping channel whose Choi matrix
$C_{1\to 2}$ satisfies the condition $\lambda_3=1$ under the form \eqref{HJ8}.
For simplicity, we also assume $\lambda_2=-\lambda_1$.
Here (D1) holds exactly when $\kappa\lambda_1=0$.
On $|\kappa\lambda_1|<1$, the (D2) criterion is the Pauli-family
criterion of Lemma~\ref{LPauli}; it is not generally equivalent to
$\kappa=0$.  At $|\kappa\lambda_1|=1$, (D2)--(D3) are
undefined and (D1) fails.  The full-rank (D3) condition follows.

\begin{lemma}\Label{LPhaseDampSigmaThree}
In the $\sigma_3$-basis phase-damping family, on
$|\kappa\lambda_1|<1$, (D3) holds if and only if
$\kappa=0$ or $|\lambda_1|=1$ and $|\kappa|<1$.
At $|\kappa\lambda_1|=1$, (D3) is undefined and (D1) fails.
\end{lemma}
\noindent\emph{Proof deferred to Appendix~\ref{app:phase-damping-sigma3}.}

Lemma~\ref{LPhaseDampSigmaThree} determines the reverse-direction positivity condition (D3) in the
present channel family.
Combining the actual (D1) criterion $\kappa\lambda_1=0$ with
the full-rank (D3) criterion, all three conditions hold simultaneously
only when $\kappa=0$.
Thus, in the $\sigma_3$-basis phase-damping family, the order-indistinguishable
region is characterized simply by the vanishing of the input Bloch parameter
$\kappa$.

\subsubsection{Depolarizing channel}

We next consider the depolarizing channel, which provides the cleanest
analytic description of the order-indistinguishable region in the main text.
Recall that, for a forward memoryless sequential strategy
$S=S_{1\to 2}(\rho_1,C_{1\to 2})$, the observed distribution ${\cal M}[S]$
is order-indistinguishable if and only if it satisfies all three
reverse-compatibility conditions (D1), (D2), and (D3).
In the depolarizing family, these conditions can be written explicitly,
and their combination yields a sharp boundary in the $(\lambda,\kappa)$-plane.

Assume that the initial state is
\[
\rho_1=\frac{1}{2}(I+\kappa \sigma_1),
\qquad
\kappa\in[-1,1].
\]
When the channel corresponds to the depolarizing Choi matrix \eqref{HJ51},
we have
\[
\lambda_j=1-\lambda
\quad \text{for } j=1,3,
\qquad
\lambda_2=-(1-\lambda).
\]
The directional conditions are then given by:
\begin{itemize}
\item (D1) holds if and only if $\lambda=1$ or $\kappa=0$.
\item On the full-rank domain $|\kappa(1-\lambda)|<1$,
(D2) holds if and only if $\lambda=1$, or $\kappa=0$, or
$|\kappa|=1$ and $\lambda>0$.  At $\lambda=0$ and
$|\kappa|=1$, (D2)--(D3) are undefined and (D1) fails.
\end{itemize}
The remaining condition (D3) takes the following explicit form.

\begin{lemma}\Label{LL29}
In the depolarizing family, (D3) is defined for
$|\kappa(1-\lambda)|<1$.  On this domain it holds exactly in
the following cases:
\begin{description}
\item[Case 1] $\lambda = 0$: $|\kappa|<1$; the points
$|\kappa|=1$ are outside the domain of (D3).
\item[Case 2] $\lambda = 1$: $-1\le \kappa \le 1$.
\item[Case 3] $0<\lambda < 1$:
\begin{align}
|\kappa|
\le
\sqrt{\frac{4-3\lambda}{(3-2\lambda)(2\lambda^2-5\lambda+4)}} .
\Label{BNE}
\end{align}
\end{description}
\end{lemma}
\noindent\emph{Proof deferred to Appendix~\ref{app:proof-LL29}.}

Therefore, in the depolarizing family, the three conditions (D1), (D2), and
(D3) are satisfied simultaneously if and only if
$\lambda=1$ or
$\kappa=0$.
Indeed, when $\lambda=1$, all three conditions hold for every
$\kappa\in[-1,1]$.
For $0\le \lambda<1$, condition (D1) already forces $\kappa=0$, and then both
(D2) and (D3) are automatically satisfied.
Thus, in this family the order-indistinguishable region is described exactly
by the union of the completely depolarizing line $\lambda=1$ and the unbiased
input line $\kappa=0$.

\section{Membership problem of ${\cal M}[{\cal S}_{N,2\to 1}]$ under qubit-Pauli measurements}
\Label{S7-M}
This section removes the forward-memoryless assumption imposed in
Section~\ref{S7}.  Given an arbitrary observed distribution
$P\in{\cal M}[{\cal S}_G]$, we ask whether it belongs to
${\cal M}[{\cal S}_{N,2\to1}]$.  By Theorem~\ref{TH10} with the two parties
exchanged, the answer is determined by the reverse Markovian condition,
positivity of the reconstructed reverse Choi matrix, and condition~(C2).  We
first give the Pauli specialization of the algebraic condition in a symmetric
conditional model and then state the resulting membership criterion.

\subsection{Two-qubit Pauli form of condition (C1) in a symmetric conditional model}

We first recall the explicit two-qubit Pauli form of condition (C1).
Although the present section is devoted to the membership problem for
${\cal M}[{\cal S}_{N,2\to 1}]$, it is convenient to begin with the original
condition (C1), because the corresponding reverse-direction condition (C2) is
obtained simply by exchanging the roles of the two parties.
The following proposition is the two-qubit Pauli specialization of
condition (C1) under the symmetric conditional model introduced below.

\begin{proposition}\Label{PROP:C1-twoqubit}
In the two-qubit case, the relation \eqref{NM5} in condition \emph{(C1)}
is simplified as
\begin{align}
\sum_{x_1\in \{0,1\}}
P_{X_2|X_1,Y_1,Y_2}(x_2|x_1,y_1,y_2)
&=
2 P_{X_2|Y_1,Y_2}(x_2|0,y_2)
\notag\\
&\quad
-\sum_{y_1'\in \{1,2,3\}}
\sum_{x_1\in \{0,1\}} (-1)^{x_1}
P_{X_1|Y_1}(x_1|y_1')
\sum_{x_1'\in \{0,1\}} (-1)^{x_1'}
P_{X_2|X_1,Y_1,Y_2}(x_2|x_1',y_1',y_2)
\notag\\
&=
2 P_{X_2|Y_1,Y_2}(x_2|0,y_2)
-\sum_{y_1'\in \{1,2,3\}}
\langle \sigma_{y_1'}\otimes I\rangle
\sum_{x_1'\in \{0,1\}} (-1)^{x_1'}
P_{X_2|X_1,Y_1,Y_2}(x_2|x_1',y_1',y_2)
\Label{C1-2q}
\end{align}
for any $y_1,y_2 \in \{1,2,3\}$ and $x_2 \in \{0,1\}$.

Assume further that, for parameters
$\lambda_1,\lambda_2,\lambda_3 \in [0,1]$,
the conditional distribution $P_{X_2|X_1,Y_1,Y_2}$ has the symmetry
\begin{align}
P_{X_2|X_1,Y_1,Y_2}(x_2|x_1,y_1,y_2)
=
\left\{
\begin{array}{ll}
\frac{1}{2} & \hbox{ when } y_1 \neq y_2, \\[1mm]
\frac{1}{2}(1-\lambda_{y_1})+\lambda_{y_1} \delta_{x_2,x_1}
& \hbox{ when } y_1 = y_2,
\end{array}
\right.
\Label{C1-symmetry}
\end{align}
and define
$\eta_1,\eta_2,\eta_3,\zeta_1,\zeta_2,\zeta_3 \in [-1,1]$ by
\begin{align}
P_{X_1|Y_1}(x_1|y_1)&= \frac{1}{2}(1-\eta_{y_1})+\eta_{y_1} x_1,
\Label{eta-def}\\
P_{X_2|Y_1,Y_2}(x_2|0,y_2)&= \frac{1}{2}(1-\zeta_{y_2})+\zeta_{y_2} x_2.
\Label{zeta-def}
\end{align}
Thus $\eta_y=P_{X_1|Y_1}(1|y)-P_{X_1|Y_1}(0|y)$ and
$\zeta_y=P_{X_2|Y_1,Y_2}(1|0,y)-P_{X_2|Y_1,Y_2}(0|0,y)$; equivalently,
the corresponding $(-1)^x$-weighted expectations are $-\eta_y$ and
$-\zeta_y$, respectively.
Then:
\begin{enumerate}[(i)]
\item condition \emph{(C1)} implies
\begin{align}
\zeta_{y_2}=0
\qquad\text{for all }y_2\in\{1,2,3\};
\Label{zeta-zero}
\end{align}
\item if, in addition, the non-signaling condition $Y_1-Y_2-X_2$ holds, then
\begin{align}
\eta_{y_1}\lambda_{y_1}=0
\qquad\text{for all }y_1\in\{1,2,3\}.
\Label{etalambda-zero}
\end{align}
\end{enumerate}
\end{proposition}

\begin{proof}
The two-qubit simplification \eqref{C1-2q} follows directly from \eqref{NM5}
by setting $d=2$ and using the Pauli-measurement outcomes
$X_i\in\{0,1\}$, $Y_i\in\{1,2,3\}$.
The second equality in \eqref{C1-2q} is just the identity
\[
\langle \sigma_{y_1'}\otimes I\rangle
=
\sum_{x_1\in\{0,1\}}(-1)^{x_1}P_{X_1|Y_1}(x_1|y_1').
\]

We now analyze \eqref{C1-2q} under the symmetric model
\eqref{C1-symmetry}--\eqref{zeta-def}.

\smallskip
\noindent
{\bf Step 1: the case $y_1\neq y_2$.}
When $y_1\neq y_2$, the condition \eqref{C1-2q} becomes
\begin{align}
1
&=
(1-\zeta_{y_2})+2\zeta_{y_2} x_2
+
\sum_{y_1'\in \{1,2,3\}}
\eta_{y_1'}
\sum_{x_1'\in \{0,1\}} (-1)^{x_1'}
P_{X_2|X_1,Y_1,Y_2}(x_2|x_1',y_1',y_2)
\notag\\
&=
(1-\zeta_{y_2})+2\zeta_{y_2} x_2
+\eta_{y_2}
\sum_{x_1'\in \{0,1\}} (-1)^{x_1'}
P_{X_2|X_1,Y_1,Y_2}(x_2|x_1',y_2,y_2)
\notag\\
&=
(1-\zeta_{y_2})+2\zeta_{y_2} x_2-\eta_{y_2}\lambda_{y_2}.
\Label{C1-case-neq}
\end{align}
The left-hand side is independent of $x_2$.
Hence the coefficient of $x_2$ on the right-hand side must vanish, which
implies
\[
\zeta_{y_2}=0.
\]
Since $y_2$ is arbitrary, this proves \eqref{zeta-zero}.

\smallskip
\noindent
{\bf Step 2: the case $y_1=y_2$.}
Under the same model, we have
\begin{align}
P_{X_2,X_1|Y_1,Y_2}(x_2,x_1|y_1,y_2)
=\frac{1}{2}P_{X_1|Y_1}(x_1|y_1)
\qquad\text{when }y_1\neq y_2,
\Label{PXY-neq}
\end{align}
whereas for $y_1=y_2$,
\begin{align}
P_{X_2,X_1|Y_1,Y_2}(x_2,x_1|y_1,y_2)
=
\Bigl(\frac{1}{2}(1-\eta_{y_1})+\eta_{y_1}x_1\Bigr)
\Bigl(\frac{1}{2}(1-\lambda_{y_1})+\lambda_{y_1}\delta_{x_2,x_1}\Bigr).
\Label{PXY-eq}
\end{align}
Therefore,
\begin{align}
P_{X_2|Y_1,Y_2}(x_2|y_1,y_2)
&=
\frac{1}{2}(1-\eta_{y_1})
+\eta_{y_1}
\Bigl(\frac{1}{2}(1-\lambda_{y_1})+\lambda_{y_1}\delta_{x_2,1}\Bigr)
\notag\\
&=
\frac{1}{2}
+\eta_{y_1}\lambda_{y_1}\Bigl(-\frac{1}{2}+\delta_{x_2,1}\Bigr).
\Label{PX2-cond}
\end{align}

Now impose the non-signaling condition $Y_1-Y_2-X_2$.
Since \eqref{zeta-zero} implies
\[
P_{X_2|Y_1,Y_2}(x_2|0,y_2)=\frac12,
\]
the non-signaling condition gives
\begin{align}
\frac{1}{2}(1-\zeta_{y_2})+\zeta_{y_2} x_2
=
\frac12
=
\frac12+\eta_{y_1}\lambda_{y_1}\Bigl(-\frac12+\delta_{x_2,1}\Bigr).
\end{align}
Hence
\[
\eta_{y_1}\lambda_{y_1}=0
\qquad\text{for all }y_1\in\{1,2,3\},
\]
which proves \eqref{etalambda-zero}.
\end{proof}

Proposition~\ref{PROP:C1-twoqubit} gives the explicit qubit-Pauli form of the
algebraic consistency condition in the symmetric conditional model
\eqref{C1-symmetry}.
By exchanging the roles of the two parties, one obtains the corresponding
reverse-direction condition required for the membership problem of
${\cal M}[{\cal S}_{N,2\to 1}]$.

\subsection{Membership criterion for ${\cal M}[{\cal S}_{N,2\to 1}]$}

We now combine the above algebraic characterization with the other structural
conditions.
In the qubit-Pauli setting, the assumptions (A1), (A2), and (A3) hold.
Hence, by Theorem~\ref{TH10} with the roles of the two directions exchanged,
membership in ${\cal M}[{\cal S}_{N,2\to 1}]$ is governed by three ingredients:
the reverse-direction Markovian condition, the positivity of the reconstructed
Choi matrix for the direction $2\to1$, and the reverse-direction analogue of
the algebraic consistency condition.

The proposition above provides a concrete qubit-Pauli expression for the
algebraic part.
Therefore, under the symmetric conditional model \eqref{C1-symmetry},
the membership problem for ${\cal M}[{\cal S}_{N,2\to 1}]$ is reduced to an
explicit check of these three conditions.

\begin{theorem}\Label{THM:SN21-membership-qubit}
Assume the qubit-Pauli setting of Section~\ref{S7-1} and the symmetric
conditional model \eqref{C1-symmetry}.
Let $P_{X_1,Y_1,X_2,Y_2}\in {\cal M}[{\cal S}_G]$ be an observed distribution.
If
\begin{enumerate}[(i)]
\item the reverse-direction Markovian condition
$Y_1-Y_2-X_2$
holds,
\item the reconstructed Choi matrix for the direction $2\to1$ satisfies
$\hat C_{2\to1}[P_{X_1,Y_1,X_2,Y_2}] \ge 0$,
\item and the reverse-direction analogue of
Proposition~\ref{PROP:C1-twoqubit} holds (equivalently, the qubit-Pauli form of
condition \emph{(C2)} obtained by exchanging the labels $1$ and $2$),
\end{enumerate}
then the distribution belongs to ${\cal M}[{\cal S}_{N,2\to 1}]$.
Conversely, every element of ${\cal M}[{\cal S}_{N,2\to 1}]$ satisfies these
three conditions.
\end{theorem}

\begin{proof}
This is exactly Theorem~\ref{TH10} with the two directions exchanged,
specialized to the present qubit-Pauli setting.
Condition (i) is the directional Markovian condition for the direction
$2\to1$.
Condition (ii) is the corresponding positivity requirement for the
reconstructed Choi matrix.
Condition (iii) is the algebraic consistency condition needed in addition to
Markovianity and positivity; in the qubit-Pauli setting, its explicit form is
the reverse-direction counterpart of Proposition~\ref{PROP:C1-twoqubit}.
Hence the statement follows.
\end{proof}

Thus Section~\ref{S7} and the present section use the same reverse
compatibility question on different domains.  Conditions (D1)--(D3) exploit a
known forward representation, whereas
Theorem~\ref{THM:SN21-membership-qubit} applies the general distribution-level
criterion when no such representation is assumed.

\section{Discrimination of hierarchy
${\cal M}[{\cal S}_{P}]\subset {\cal M}[{\cal S}_{Q,1\to 2}]$ and
${\cal M}[{\cal S}_{N,1\to 2}]\subset {\cal M}[{\cal S}_{Q,1\to 2}]$
under qubit-Pauli matrices}
\Label{S7-C}
We finally compare the distribution-level signatures of the parallel and
memoryless sequential subclasses inside the larger quantum-memory class.  The
first subsection compares positivity of the reconstructed pseudo-density
matrix with the forward and reverse Choi positivity conditions.  The second
uses explicit memoryful strategies to show that reconstructed-Choi positivity
and condition~(C1) detect distinct failures of memoryless sequentiality.

\subsection{Hierarchy
${\cal M}[{\cal S}_{P}]\subset {\cal M}[{\cal S}_{Q,1\to 2}]$}
To analyze the hierarchy ${\cal M}[{\cal S}_{P}]\subset {\cal M}[{\cal S}_{Q,1\to 2}]$,
we first give a positivity test for the observed distribution of an
arbitrary $S\in{\cal S}_{Q,1\to2}$.  Put $P={\cal M}[S]$ and
$R={\cal R}[P]$, and write
\begin{align}
4R=I\otimes I+\sum_i a_i\sigma_i\otimes I
+\sum_j b_jI\otimes\sigma_j
+\sum_{i,j}T_{ij}\sigma_i\otimes\sigma_j,
\qquad K:=T-ab^{\mathsf T}.\Label{PDM-Bloch}
\end{align}
Here $T$ is the full correlation matrix and $K$ is the connected
correlation matrix.  Since $I,\sigma_1,\sigma_2,\sigma_3$ form a basis
on each qubit, the observed Pauli moments uniquely determine $a,b,T$.
Vector norms below are Euclidean and $\|T\|_{\mathrm F}$ is Frobenius.
Without assuming (D1) or (D2), 
we have the following lemma.

\begin{lemma}[General two-qubit PDM positivity test]
\Label{thm:PDM-bloch-positivity}
Put $F:=\|a\|_2^2+\|b\|_2^2+\|T\|_{\mathrm F}^2$.  For any
$S\in{\cal S}_{Q,1\to2}$, $R\ge0$ holds if and only if
\begin{align}
F&\le3,\Label{PDM-pos-e2}\\
1-F+2a^{\mathsf T}Tb-2\det T&\ge0,\Label{PDM-pos-e3}\\
\det R&\ge0.\Label{PDM-pos-e4}
\end{align}
\end{lemma}
\noindent\emph{Proof in Appendix~\ref{app:proof-PDM-bloch}.}
\begin{corollary}[Pure second marginal under (D1)]
\Label{cor:PDM-pure-second}
If $S\in{\cal S}_{Q,1\to2}$, $P={\cal M}[S]$ satisfies (D1), and
$\|b\|_2=1$, then $K=0$ and
$R=\rho_1[P]\otimes\rho_2[P]\ge0$.  Consequently,
$P\in{\cal M}[{\cal S}_P]$ by \eqref{FG1VS}.
\end{corollary}
\noindent\emph{Proof in Appendix~\ref{app:proof-PDM-bloch}.}
For $P={\cal M}[S]$ with $S\in{\cal S}_{Q,1\to2}$, membership in ${\cal M}[{\cal S}_P]$ is equivalent to (D1) and ${\cal R}[P]\ge0$, by \eqref{FG1VS}. In general, positivity is equivalent to \eqref{PDM-pos-e2}--\eqref{PDM-pos-e4}. If the second marginal state $\rho_2[P]$ is pure, however, Corollary~\ref{cor:PDM-pure-second} shows that (D1) already implies ${\cal R}[P]\ge0$. Purity of $\rho_2[P]$ alone does not imply this positivity.

The formulas below use the forward Choi representation \eqref{HJ8}. For general quantum-memory strategies, use Lemma~\ref{thm:PDM-bloch-positivity} instead.

\begin{lemma}[PDM positivity in the forward memoryless standard form]
\Label{LL29-A}
Let $S=S_{1\to2}(\rho_1,C_{1\to2})\in{\cal S}_{N,1\to2}$ have
representation \eqref{HJ8}, and assume (D1).  
If $t\ge1$, assume also $\rho_2[{\cal M}[S]]>0$ and (D2).  
Put $r_1=\|a\|_2$ and $r_2=\|b\|_2$ in
\eqref{PDM-Bloch}.  Then the nonzero singular values of $K$ are
$\lambda_1,\ldots,\lambda_t$, and $R\ge0$ is equivalent to the
applicable condition below.
\begin{description}
\item[(1)] $t=3$: $a=b=0$.  For the local-unitary form
$4R=I\otimes I+\sum_{j=1}^3\lambda_j\sigma_j\otimes\sigma_j$,
the condition is $1\ge\lambda_1+\lambda_2+\lambda_3$.
For the form with the sign of the second correlation term reversed,
it is $1\ge\max\{-\lambda_1+\lambda_2+\lambda_3,
\lambda_1-\lambda_2+\lambda_3,
\lambda_1+\lambda_2-\lambda_3\}$.
\item[(2)] $t=2$: the condition is
\begin{align}
(1-r_1^2)(1-r_2^2)\ge(\lambda_1+\lambda_2)^2.\Label{BP1T}
\end{align}
\item[(3)] $t=1$: the condition is
\begin{align}
(1-r_1^2)(1-r_2^2)\ge\lambda_1^2.\Label{BP2E}
\end{align}
\item[(4)] $t=0$: $R=\rho_1\otimes\rho_2\ge0$, with no (D2) or
full-rank assumption.
\end{description}
\end{lemma}
\noindent\emph{Proof in Appendix~\ref{app:proof-LL29A}.}
Lemmas~\ref{LBU1} and \ref{LBU2} provide the two-sided alignment used
in Lemma~\ref{LL29-A} only for the forward memoryless representation.
The formulas \eqref{BP1T} and \eqref{BP2E} supplement rather than
replace Lemma~\ref{thm:PDM-bloch-positivity}.  The comparison with
Lemmas~\ref{LBU0} and \ref{LBU3} is restricted to that standard form.
\subsection{Hierarchy
${\cal M}[{\cal S}_{N,1\to 2}] \subset {\cal M}[{\cal S}_{Q,1\to 2}]$}\Label{S9B}
We now turn to the hierarchy
${\cal M}[{\cal S}_{N,1\to 2}] \subset {\cal M}[{\cal S}_{Q,1\to 2}]$.
Every strategy in ${\cal S}_{Q,1\to 2}$ already satisfies the non-signaling
condition $X_1-Y_1-Y_2$, so this directional condition alone cannot separate
the memoryless sequential class from the larger quantum-memory class.
Proposition~\ref{P1Auto} shows that 
the positivity condition $\hat C_{1 \to 2}[P]\ge 0$ and
the additional algebraic condition (C1) are satisfied by all memoryless sequential strategies.
Here, both conditions are useful to distinguish 
an element of 
${\cal M}[{\cal S}_{N,1\to 2}]$ from elements in ${\cal M}[{\cal S}_{Q,1\to 2}]$.

First, we show an example illustrating the usefulness of
the positivity condition $\hat C_{1 \to 2}[P]\ge 0$.

\begin{example}\Label{ExB}
Charlie prepares a maximally entangled state
$|\Phi^+\rangle
:=
\frac{1}{\sqrt2}(|00\rangle+|11\rangle)$
on Alice's input system and an internal memory qubit $M$.
He sends Alice's half to Alice and keeps the memory qubit $M$.
After Alice performs her projective Pauli measurement,
Charlie discards Alice's output system and sends the memory qubit $M$
to Bob's input system.
\end{example}

\begin{lemma}\Label{LEM-ExB}
The strategy given by Example~\ref{ExB}
satisfies the condition \emph{(C1)} and
does not satisfy 
the positivity condition $\hat C_{1 \to 2}[P]\ge 0$.
\end{lemma}

\noindent\emph{Proof deferred to Appendix~\ref{app:proof-ExB}.}

Next, we present an example for the usefulness of
the condition (C1).

\begin{example}\Label{Ex1prime}
Charlie prepares the classical memory $\cX=\{0,1\}$ with the uniform
distribution.
He inputs the state $|(-1)^X,3\rangle$ to Alice's input system.
Later, Charlie applies the unitary $\sigma_X$ to Alice's output system,
then applies the depolarizing channel
$D_{1/2}(\rho):=\frac12\,\rho+\frac12 \rho_{mix}$,
where $\rho_{mix}:=\frac{I}{2}$,
and sends the result to Bob's input system.
\end{example}

\begin{lemma}\Label{LEM-Ex1prime}
The strategy given by Example~\ref{Ex1prime}
satisfies the positivity condition $\hat C_{1\to 2}[P]\ge 0$,
and does not satisfy the condition \emph{(C1)}.
\end{lemma}
\noindent\emph{Proof deferred to Appendix~\ref{app:proof-Ex1prime}.}

Since a classical memory can be embedded into a quantum memory as in \eqref{classical-quantum-memory-inclusion}, Example~\ref{Ex1prime} also provides an element of ${\cal M}[{\cal S}_{Q,1\to2}]$ that violates condition~\emph{(C1)}, even though its reconstructed Choi matrix is positive.

These two examples play complementary roles within ${\cal M}[{\cal S}_{Q,1\to2}]$. Example~\ref{ExB} shows that condition~\emph{(C1)} does not imply positivity of the reconstructed Choi matrix. By contrast, Example~\ref{Ex1prime} shows that positivity of the reconstructed Choi matrix does not imply condition~\emph{(C1)}. Thus the two conditions detect distinct obstructions to a memoryless sequential representation, and neither can replace the other in the characterization of ${\cal M}[{\cal S}_{N,1\to2}]$.

Example~\ref{ExB} showed that, when Charlie's memory is quantum,
the positivity condition
$\hat C_{1\to 2}[P]\ge 0 $
may fail.
This failure of positivity is not restricted to the quantum-memory case.
It can already occur when Charlie's memory is purely classical.
The following example provides such a classical-memory strategy.
Unlike Example~\ref{ExB}, however, the present example does not isolate
positivity alone: in this case, the observed distribution violates not only
the positivity condition $\hat C_{1\to 2}[P]\ge 0$,
but also condition \emph{(C1)}.
Thus, the example should be understood only as showing that the positivity
condition can already fail in the classical-memory subclass.

\begin{example}\Label{Ex1}
Charlie prepares the classical memory $\cX=\{0,1\}$.
He chooses $X=0$ with probability $\lambda \in (0,1)$,
and chooses $X=1$ with probability $1-\lambda$.
He inputs the state $|(-1)^X,3\rangle$ to Alice's input system.
Later, Charlie applies the unitary $\sigma_{X}$
to Alice's output system and sends the result to Bob's input system.
Here, to satisfy the condition (C0), 
the condition $\lambda \neq 0,1$ is needed. 
\end{example}

\begin{lemma}\Label{LEM-Ex1}
The strategy given by Example~\ref{Ex1}
satisfies neither condition \emph{(C1)} nor the positivity condition
$\hat C_{1\to 2}[P]\ge 0$.
\end{lemma}

\noindent\emph{Proof deferred to Appendix~\ref{app:proof-Ex1}.}

\section{Conclusion}
We studied the statistical identification of causal structure in a restricted
projective-observation regime, where Alice and Bob have access only to local
measurement data and full process tomography is impossible.
Our main focus was the memoryless sequential scenario, in which Charlie mediates
a definite order without retaining memory, and the central question was whether
this structure can be recognized directly from the observed distribution.

Our first main result is a complete distribution-level characterization of
${\cal M}[{\cal S}_{N,1\to 2}]$ under the assumptions (A1) and (A3).
We showed that the directional Markovian condition $X_1-Y_1-Y_2$ and the
positivity of the reconstructed Choi matrix $\hat{C}_{1\to2}[P]$ are not
sufficient by themselves.
The missing ingredient is the additional algebraic consistency condition (C1).
Theorem~\ref{TH10} shows that these three conditions are together necessary and sufficient. It also establishes that, within the class of distributions satisfying the directional Markovian condition, condition~(C1) is equivalent to exact algebraic reproduction of the observed distribution by the reconstructed state and candidate Choi matrix. Hence any alternative condition sufficient for such reproduction must imply condition~(C1). In this precise sense, condition~(C1) is the minimal additional algebraic consistency condition. Positivity then promotes the reconstructed candidate map to a physical channel and completes the memoryless-sequential membership criterion.
In this sense, the restricted projective problem can still be solved exactly at
the level of observed statistics, even though exact process-matrix
reconstruction is unavailable.

A central conceptual point of the paper is that condition (C1) is genuinely new.
It is not a reformulation of positivity, nor does it follow from the usual
non-signaling or directional Markovian constraints.
Rather, it expresses the consistency requirement that the observed joint
statistics arise from a single underlying channel acting on the post-measurement
states.
This clarifies why restricted projective observations require more than the
standard combination of Markovianity and complete positivity.
The example constructed in Section~\ref{S7-C} makes this point explicit by
showing that a strategy may satisfy the relevant positivity requirements and
still fail to belong to the memoryless sequential class because (C1) is
violated.

Our second main contribution is the explicit two-qubit Pauli analysis.
Although this setting is not tomographically complete, it is analytically
tractable and already captures the essential obstruction to order
identification.
For a forward memoryless sequential strategy
$S_{1\to2}(\rho_1,C_{1\to2})$, we translated the reverse-compatibility problem
into the concrete conditions (D1)--(D3), and we analyzed these conditions both
for general qubit channels and for representative channel families such as
Pauli, phase-damping, and depolarizing channels.
This yields a closed-form description of the order-indistinguishable region in
the minimal nontrivial dimension.

Our third contribution is the clarification of observable hierarchies.
We showed that positivity of reconstructed objects is not, by itself, enough to
distinguish the relevant operational classes.
In particular, the comparison between reconstructed-Choi positivity and
pseudo-density-matrix positivity reveals that different positivity tests lead to
different observable boundaries.
To separate the memoryless sequential class from the broader quantum-memory
class, an additional distribution-level condition is indispensable, and this is
precisely the role played by (C1).

Finally, we formulated the reverse-direction membership problem for ${\cal M}[{\cal S}_{N,2\to 1}]$ in the same qubit-Pauli regime. 
This should be viewed as the full version of the reverse-direction question: in Section~\ref{S7} we considered the same question only for distributions already known to belong to ${\cal M}[{\cal S}_{N,1\to 2}]$, which led to the order-indistinguishability problem, whereas here we treated arbitrary observed distributions. 
By applying the reverse-direction analogue of Theorem~\ref{TH10}, we obtained a concrete membership criterion in terms of the reverse Markovian condition, the positivity of the reconstructed reverse-direction Choi matrix, and the corresponding algebraic consistency condition.

Taken together, these results turn the qualitative limitation of causal-order
identification under restricted projective observations into an explicit
set-membership problem.
They show that even in a genuinely non-tomographic regime, one can still obtain
sharp and operationally checkable criteria for memorylessness, causal
compatibility, and order-indistinguishability.
This provides a concrete statistical framework for causal inference in quantum
network scenarios where experimental control is limited to fixed projective
measurements.

\subsection*{Acknowledgements}
The author thanks Dr Baichu Yu for valuable discussions.
The author was supported in part by
the General R\&D Projects of 1+1+1 CUHK-CUHK(SZ)-GDST Joint Collaboration Fund (Grant No. GRDP2025-022), 
and the Guangdong Provincial Quantum Science Strategic Initiative (Grant No. 
GDZX2505003).
Generative AI tools were used in the preparation of this manuscript to assist with drafting and revising the exposition, improving the overall presentation, and generating source code for figures designed by the author. Their mathematical use was limited to assistance with several calculations in the final section. The research questions, mathematical framework, main results, and proof strategy were developed by the author. All AI-assisted text, figure code, and calculations were reviewed and, where necessary, revised by the author, who takes full responsibility for the contents of the manuscript.

\appendix
\section{Additional materials for the two-qubit sections}\Label{app:qubit-patch}

\subsection{Basic positivity criterion for the qubit-Pauli analysis}\Label{app:proof-LL31}
The following lemma provides the basic positivity criterion 
used in the following appendices.

\begin{lemma}\Label{LL31}
The relation
$I \otimes I
+s_1 \sigma_1 \otimes I
+s_2 I \otimes \sigma_1
+\sum_{j=1}^3 \tau_j \sigma_j \otimes \sigma_j \ge 0$ holds
if and only if
the relations
\begin{align}
1+\tau_1+s_1+s_2 &\ge 0 \Label{AP1}\\
1+\tau_1-s_1-s_2 &\ge 0 \Label{AP1b}\\
1-\tau_1-s_1+s_2 &\ge 0 \Label{AP2}\\
1-\tau_1+s_1-s_2 &\ge 0 \Label{AP2b}\\
(1-\tau_1)^2 &\ge (s_1-s_2)^2+(\tau_2+\tau_3)^2 \Label{BP1}\\
(1+\tau_1)^2 &\ge (s_1+s_2)^2+(\tau_2-\tau_3)^2 .\Label{BP2}
\end{align}
hold.
\end{lemma}

\begin{proof}
Applying
Hadamard matrix, we exchange the bases of $\sigma_1$ and $\sigma_3$.
After this conversion, the matrix representation of
$I \otimes I
+s_1 \sigma_1 \otimes I
+s_2 I \otimes \sigma_1
+\sum_{j=1}^3 \tau_j \sigma_j \otimes \sigma_j$
in the basis
$|0,0\rangle, |1,0\rangle,|0,1\rangle,|1,1\rangle$ as
\begin{align}
\left(
\begin{array}{cccc}
1+\tau_1+s_1+s_2 & 0 & 0 &\tau_3-\tau_2 \\
0& 1-\tau_1-s_1+s_2 & \tau_3+\tau_2 &0 \\
0& \tau_3+\tau_2 & 1-\tau_1+s_1-s_2 & 0 \\
\tau_3-\tau_2 & 0& 0 & 1+\tau_1-s_1-s_2
\end{array}
\right)
\end{align}
The first $2\times2$ block is positive semidefinite precisely when
both its diagonal entries satisfy \eqref{AP1} and \eqref{AP1b}, and
its determinant satisfies
\begin{align}
(1+\tau_1+s_1+s_2)
(1+\tau_1-s_1-s_2) \ge (\tau_2-\tau_3)^2 ,
\end{align}
which is equivalent to \eqref{BP2}.
The second block is positive semidefinite precisely when both its
diagonal entries satisfy \eqref{AP2} and \eqref{AP2b}, and
its determinant satisfies
\begin{align}
(1-\tau_1-s_1+s_2) (1-\tau_1+s_1-s_2)
\ge (\tau_2+\tau_3)^2 ,
\end{align}
which is equivalent to \eqref{BP1}.  Both diagonal conditions are
necessary even at a zero pivot: the determinant condition alone does
not exclude a negative remaining diagonal entry.
\end{proof}

\subsection{Proof of Lemma~\ref{LBU1}}\Label{app:proof-LBU1}
\begin{proof}
Alice's averaged output state is
$\frac{1}{2}(I+ c^1_{Y_1} \sigma_{Y_1})$ for when $Y_1=1,2,3$,
and is
$\frac{1}{2}(I+ \sum_{j=1}^3 c^1_j \sigma_j)$ for $Y_1=0$.
The condition (D1) is equivalent to the condition
$\Lambda_{1\to 2}(\frac{1}{2}(I+ c^1_{i} \sigma_{i}))=
\Lambda_{1\to 2}(\frac{1}{2}(I+ \sum_{y=1}^3 c^1_y \sigma_y))$ for $i=1,2,3$,
which means that
\begin{align}
\rho_2+\sum_{j=1}^t \frac{c^1_i \lambda_j}{4}\Tr (\alpha_j \sigma_{i}) \beta_j =
\rho_2+\sum_{j=1}^t \sum_{y=1}^3 \frac{c^1_y \lambda_j}{4}
\Tr (\alpha_j \sigma_{y}) \beta_j \Label{BF0}
\end{align}
for $y_1=1,2,3$.
Since $\lambda_j>0$, the above condition is equivalent to
the condition
\begin{align}
c^1_i \Tr (\alpha_j \sigma_{i})  =
 \sum_{y=1}^3 c^1_y  \Tr (\alpha_j \sigma_{y}) \Label{BF1}
\end{align}
for $i=1,2,3$ and $j\le t$.
This condition is equivalent to
\begin{align}
c^1_i \Tr (\alpha_j \sigma_{i})  =0 \Label{BF2}
\end{align}
for $i=1,2,3$ and $j\le t$.
Indeed, for fixed $j$, denote the right-hand side of \eqref{BF1} by $s_j$.
Summing \eqref{BF1} over $i=1,2,3$ gives $s_j=3s_j$, and hence $s_j=0$.
Conversely, if \eqref{BF2} holds, then every term on the right-hand side of \eqref{BF1} vanishes, as does its left-hand side; hence \eqref{BF1} holds.

Case (1).
$t=3$.
Since $\alpha_1,\alpha_2,\alpha_3$ are linearly independent,
the matrix $(\Tr (\alpha_j \sigma_{i}))_{j,i}$ is invertible.
The condition (D1), i.e., the condition \eqref{BF1} is equivalent to the condition
$c^1_i=0$ for $i=1,2,3$.

Case (2).
$t=2$.
Since $\alpha_1,\alpha_2$ are linearly independent,
by \eqref{BF2}, condition (D1) holds if and only if, for every
$i=1,2,3$, either $c_i^1=0$ or
$\Tr(\alpha_1\sigma_i)=\Tr(\alpha_2\sigma_i)=0$.
Moreover, the $2\times3$ matrix
$(\Tr (\alpha_j\sigma_i))_{j=1,2;\,i=1,2,3}$ has rank two and therefore has
at least two nonzero columns.  Hence, $c_i^1=0$ for every index
corresponding to a nonzero column.
That is, at least two indices $i$ satisfy $c^1_i=0$.

Case (3).

$t=1$.
By \eqref{BF2}, condition (D1) holds if and only if, for every
$i=1,2,3$, either $c_i^1=0$ or $\Tr(\alpha_1\sigma_i)=0$.
Since $\alpha_1\neq0$, at least one index $i$ satisfies
$\Tr (\alpha_1 \sigma_{i}) \neq 0$, and hence $c_i^1=0$ for every
such index.
That is, at least one index $i$ satisfies $c^1_i=0$.

Case (4).

$t=0$.
The condition (D1) always holds
because the relation \eqref{BF0} always holds.
Thus no additional condition is required.
\end{proof}

\subsection{Proof of Lemma~\ref{LBU2}}\Label{app:proof-LBU2}
\begin{proof}
\textbf{Step 1: the PDM under (D1).}
By Lemma~\ref{LBU1}, the Bloch vector of $\rho_1$ is orthogonal to
each $\alpha_j$ for $j\le t$.  Hence
$\alpha_j\circ\rho_1=\alpha_j/2$.  The forward compatibility
relation \eqref{XS} applied to \eqref{HJ8} yields \eqref{DSJK}.
Taking the partial trace gives \eqref{DSJK2}.

\textbf{Step 2: reverse reconstruction in the full-rank case.}
Use \eqref{PDM-Bloch} with $a=c^1$ and $b=c^2$.  By \eqref{DSJK},
$K=\sum_{j=1}^t\lambda_jv_jw_j^{\mathsf T}$, where
$\alpha_j=v_j\cdot\sigma$, $\beta_j=w_j\cdot\sigma$, and the
$v_j$ and $w_j$ form orthonormal families.  For $|b|<1$, the reverse
reconstruction \eqref{XSA2} is unique.  Set
$Q:=\tilde C_{2\to1}[{\cal M}[S]]^{T_2}$ and write
\begin{align}
2Q=I\otimes I+\sum_i u_i\sigma_i\otimes I
+\sum_{i,j}M_{ij}\sigma_i\otimes\sigma_j.\Label{LBU2-Q}
\end{align}
There is no $I\otimes\sigma_j$ term since $\Tr_1Q=I$.
Use $\rho_2=(I+b\cdot\sigma)/2$ and
$\sigma_j\circ\rho_2=(\sigma_j+b_jI)/2$.  In $4R$, the term
$M_{ij}\sigma_i\otimes\sigma_j$ in $2Q$ contributes
$M_{ij}\sigma_i\otimes\sigma_j+b_jM_{ij}\sigma_i\otimes I$;
the $u_i\sigma_i\otimes I$ term contributes
$u_i\sigma_i\otimes I+u_ib_j\sigma_i\otimes\sigma_j$.
Comparing the local and correlation coefficients gives
\begin{align}
a=u+Mb,\qquad T=ub^{\mathsf T}+M,\qquad
K=M(I-bb^{\mathsf T}).\Label{LBU2-recon}
\end{align}
Since $\|b\|_2<1$, the matrix $I-bb^{\mathsf T}$ is invertible;
thus $M=K(I-bb^{\mathsf T})^{-1}$ and $u=a-Mb$.

\textbf{Step 3: the exact (D2) condition.}
The reverse candidate sends a second-system input with Bloch vector $v$
to a first-system output with Bloch vector $u+Mv$.  Without a second-player
measurement the vector is $b$; after a nonselective $\sigma_i$
measurement it is $b_i e_i$.  Thus (D2) is equivalent to
$M(b_i e_i-b)=0$ for each $i$.  Summing these equations gives
\begin{align}
0=\sum_{i=1}^3 M(b_i e_i-b)
=M\left(\sum_i b_i e_i-3b\right)=-2Mb.\Label{LBU2-D2-sum}
\end{align}
Hence $Mb=0$, and each original equation reduces to
$b_iM_{:i}=0$.  Conversely, the three column conditions imply
$Mb=\sum_i b_iM_{:i}=0$ and recover all three original equations.
By \eqref{LBU2-recon}, $Mb=0$ gives $K=M$.  Conversely,
$b_iK_{:i}=0$ for all $i$ gives $Kb=0$; hence
$K(I-bb^{\mathsf T})=K$, and the unique inverse in
\eqref{LBU2-recon} gives $M=K$ and $Mb=0$.  Thus
\begin{align}
\text{(D2)}\quad\Longleftrightarrow\quad
b_iK_{:i}=0\quad(i=1,2,3).\Label{LBU2-column}
\end{align}
Since $K_{:i}=\sum_{j=1}^t\lambda_jv_j(w_j)_i$ and the
$v_j$ are orthonormal with $\lambda_j>0$, this is equivalent to
$b_i(w_j)_i=0$ for every $i$ and $j\le t$.  For $t=3$, the
$w_j$ span $\mathbb R^3$, so $b=0$.  For $t=2$ or $t=1$,
each $i$ satisfies either $b_i=0$ or $(w_j)_i=0$ for all
$j\le t$, precisely the alternatives in part (ii).
For $t=0$ there is no restriction.  Under (D2), $M=K$ and
$u=a$, so \eqref{LBU2-Q} gives \eqref{VNI1}.

\textbf{Step 4: the singular boundary.}
If $\|b\|_2=1$, the qubit state $\rho_2$ is pure.  For each
first-player Pauli choice $y_1$, (D1) gives
\begin{align}
\rho_2=\sum_{x_1}p(x_1\mid y_1)\rho_{2\mid x_1,y_1}.
\Label{LBU2-pure-mixture}
\end{align}
Extremality forces each positive-weight conditional state to equal
$\rho_2$.  Therefore
$p(x_1,x_2\mid y_1,y_2)=p(x_1\mid y_1)
\Tr(E^2_{x_2\mid y_2}\rho_2)$ for all Pauli choices.
The same holds when the first player does nothing.  Uniqueness
under (A1) gives $R=\rho_1\otimes\rho_2$ and $K=0$.  By \eqref{DSJK} and the
orthonormality of the singular-vector families, $t=0$.  This argument
does not use a reverse reconstruction on the singular boundary.
\end{proof}

\subsection{Proofs of Lemmas~\ref{LBU0} and \ref{LBU3}}
\Label{app:proof-LBU0-LBU3}

\begin{proof}[Proof of Lemma~\ref{LBU0}]
For $t\ge1$, the full-rank assumption and (D1)--(D2) allow the
structural restrictions of Lemma~\ref{LBU2}; the resulting positivity
conditions reduce to Lemma~\ref{LL31}.  The $t=0$ case is direct.

\subsubsection*{Case (1): $t=3$.}
Lemma~\ref{LBU2} guarantees that $\rho_1=I/2$.
Applying a unitary on the first system, we can transform
$\alpha_1,\alpha_2,\alpha_3$ into either
$\sigma_1,\sigma_2,\sigma_3$ or
$\sigma_1,-\sigma_2,\sigma_3$.
Similarly, applying a unitary on the second system, we can transform
$\beta_1,\beta_2,\beta_3$ into either
$\sigma_1,\sigma_2,\sigma_3$ or
$\sigma_1,-\sigma_2,\sigma_3$.
Hence there are two local-unitary normal forms.

\smallskip
\noindent
\emph{Case (A).}
Then $C_{1\to 2}$ is unitarily equivalent to
\[
\frac12\bigl(
I\otimes I
+\lambda_1\sigma_1\otimes\sigma_1
-\lambda_2\sigma_2\otimes\sigma_2
+\lambda_3\sigma_3\otimes\sigma_3
\bigr).
\]
Since all four diagonal conditions \eqref{AP1}, \eqref{AP1b},
\eqref{AP2}, and \eqref{AP2b} hold here (as $0\le\lambda_1\le1$), Lemma~\ref{LL31} implies that $C_{1\to 2}\ge0$ is equivalent to
\[
(1-\lambda_1)^2\ge(-\lambda_2+\lambda_3)^2,
\qquad
(1+\lambda_1)^2\ge(-\lambda_2-\lambda_3)^2.
\]
Because $\lambda_j>0$, this is equivalent to
\[
1\ge
\max(-\lambda_1+\lambda_2+\lambda_3,\,
\lambda_1-\lambda_2+\lambda_3,\,
\lambda_1+\lambda_2-\lambda_3).
\]

\smallskip
\noindent
\emph{Case (B).}
Then $C_{1\to 2}$ is unitarily equivalent to
\[
\frac12\bigl(
I\otimes I
+\lambda_1\sigma_1\otimes\sigma_1
+\lambda_2\sigma_2\otimes\sigma_2
+\lambda_3\sigma_3\otimes\sigma_3
\bigr).
\]
Again all four diagonal conditions hold because $0\le\lambda_1\le1$,
and Lemma~\ref{LL31} yields
\[
(1-\lambda_1)^2\ge(\lambda_2+\lambda_3)^2,
\qquad
(1+\lambda_1)^2\ge(\lambda_2-\lambda_3)^2.
\]
Since $\lambda_j>0$, this is equivalent to
\[
1\ge \lambda_1+\lambda_2+\lambda_3.
\]

\subsubsection*{Case (2): $t=2$.}
By Lemma~\ref{LBU2}, the Bloch vector of $\rho_2$ is orthogonal to
$\operatorname{span}\{\beta_1,\beta_2\}$.  We choose a local-unitary coordinate
system on the second qubit in which this Bloch vector is in the $\sigma_1$
direction.  For the calculation below, set
$r_2:=\sqrt{\sum_{i=1}^3(c_i^2)^2}$.
Lemma~\ref{LBU2} guarantees that
\[
C_{1\to 2}
=
\frac12 I\otimes (I+r_2\sigma_1)^T
+\sum_{j=1}^2\frac{\lambda_j}{2}\alpha_j\otimes\beta_j^T.
\]
Assume first that $r_2\neq0$.
Since
$r_2\Tr \beta_1\sigma_1=r_2\Tr \beta_2\sigma_1=0$,
we can, by a unitary on the second system fixing $\sigma_1$, transform
$\beta_1,\beta_2$ into $\pm\sigma_2,\sigma_3$.
By a unitary on the first system, we can simultaneously transform
$\alpha_1,\alpha_2$ into $\pm\sigma_2,\sigma_3$ with the same sign choice.
Thus $C_{1\to 2}$ is unitarily equivalent to
\[
\frac12\bigl(
I\otimes I
+r_2 I\otimes \sigma_1
-\lambda_1\sigma_2\otimes\sigma_2
+\lambda_2\sigma_3\otimes\sigma_3
\bigr).
\]
Applying Lemma~\ref{LL31}, we obtain
\[
1\ge (\lambda_1+\lambda_2)^2+r_2^2.
\]
If $r_2=0$, the $t=3$ discussion shows that
$1\ge\lambda_1+\lambda_2$, which is equivalent to the same inequality.
This proves Case (2).

\subsubsection*{Case (3): $t=1$.}
By Lemma~\ref{LBU2}, the Bloch vector of $\rho_2$ is orthogonal to
$\beta_1$.  We choose local-unitary coordinates in which $\beta_1$ is in the
$\sigma_3$ direction and the Bloch vector lies in the
$\operatorname{span}\{\sigma_1,\sigma_2\}$ plane; in these coordinates,
set $r_2:=\sqrt{\sum_{i=1}^3(c_i^2)^2}$; thus
$(c_1^2)^2+(c_2^2)^2=r_2^2$.
Lemma~\ref{LBU2} yields
\[
C_{1\to 2}
=
\frac12 I\otimes (I+c_1^2\sigma_1+c_2^2\sigma_2)
+\frac{\lambda_1}{2}\alpha_1\otimes\beta_1^T.
\]
If $c_1^2\neq0$ and $c_2^2\neq0$, then
$c_1^2\Tr\beta_1\sigma_1=c_2^2\Tr\beta_1\sigma_2=0$,
so $\beta_1=\pm\sigma_3$.
A unitary on the first system transforms $\alpha_1$ to $\sigma_3$.
Hence $C_{1\to 2}$ is unitarily equivalent to
\[
\frac12\bigl(
I\otimes I
+r_2 I\otimes \sigma_1
\pm \lambda_1 \sigma_3\otimes\sigma_3
\bigr),
\]
and Lemma~\ref{LL31} implies
\[
1\ge \lambda_1^2+r_2^2.
\]
If one of $c_1^2,c_2^2$ vanishes, the $t=2$ discussion reduces again to the
same inequality. This proves Case (3).

\subsubsection*{Case (4): $t=0$.}
In this case
\[
C_{1\to 2}=I\otimes \rho_2\ge0,
\]
so no further condition is required.
\end{proof}

\begin{proof}[Proof of Lemma~\ref{LBU3}]
When the second marginal is full rank, statement (i) is precisely
\eqref{VNI1}, as established in Lemma~\ref{LBU2}.  If it is singular,
Lemma~\ref{LBU2} gives $t=0$ and
$\mathcal R[\mathcal M[S]]=\rho_1\otimes\rho_2$.
The replacement channel $\Lambda_{2\to1}(\omega)=\rho_1\Tr\omega$
has Choi matrix $\rho_1\otimes I\ge0$ and reproduces the product
distribution.  This does not assign a value to the undefined
$\tilde C_{2\to1}$.  We now prove (ii) in the full-rank case.
The point is that the positivity analysis for
$\tilde{C}_{2\to1}[{\cal M}[S]]$ is exactly parallel to the positivity analysis
for $C_{1\to2}$ in Lemma~\ref{LBU0}.
Indeed, once \eqref{VNI1} is available, the only difference is that the Bloch
parameters entering the local state now come from $\rho_1$ rather than $\rho_2$.
Therefore the same local-unitary reductions and the same applications of
Lemma~\ref{LL31} go through verbatim, with $c_i^2$ replaced by $c_i^1$.

More explicitly:

\subsubsection*{Case (1): $t=3$.}
Under (D1) and (D2), Lemma~\ref{LBU1} and Lemma~\ref{LBU2} imply
$\rho_1=I/2$ and $\rho_2=I/2$.
Hence \eqref{VNI1} has exactly the same two local-unitary normal forms as in
Lemma~\ref{LBU0}, namely
\[
\frac12\bigl(
I\otimes I
+\lambda_1\sigma_1\otimes\sigma_1
-\lambda_2\sigma_2\otimes\sigma_2
+\lambda_3\sigma_3\otimes\sigma_3
\bigr)
\]
and
\[
\frac12\bigl(
I\otimes I
+\lambda_1\sigma_1\otimes\sigma_1
+\lambda_2\sigma_2\otimes\sigma_2
+\lambda_3\sigma_3\otimes\sigma_3
\bigr).
\]
Applying Lemma~\ref{LL31} exactly as in the proof of Lemma~\ref{LBU0}, we
obtain the two conditions
\[
1\ge
\max(-\lambda_1+\lambda_2+\lambda_3,\,
\lambda_1-\lambda_2+\lambda_3,\,
\lambda_1+\lambda_2-\lambda_3)
\]
and
\[
1\ge \lambda_1+\lambda_2+\lambda_3.
\]

\subsubsection*{Case (2): $t=2$.}
By Lemma~\ref{LBU1}, the Bloch vector of $\rho_1$ is orthogonal to
$\operatorname{span}\{\alpha_1,\alpha_2\}$.  We choose local-unitary
coordinates on the first qubit in which this Bloch vector is in the
$\sigma_1$ direction.  For the calculation below, set
$r_1:=\sqrt{\sum_{i=1}^3(c_i^1)^2}$.
Now \eqref{VNI1} has the same form as the $t=2$ case of Lemma~\ref{LBU0},
except that the local Bloch parameter comes from $\rho_1$.
Thus the same unitary reduction gives a matrix unitarily equivalent to
\[
\frac12\bigl(
I\otimes I
+r_1 \sigma_1\otimes I
-\lambda_1\sigma_2\otimes\sigma_2
+\lambda_2\sigma_3\otimes\sigma_3
\bigr),
\]
and Lemma~\ref{LL31} yields
\[
1\ge (\lambda_1+\lambda_2)^2+r_1^2.
\]

\subsubsection*{Case (3): $t=1$.}
By Lemma~\ref{LBU1}, the Bloch vector of $\rho_1$ is orthogonal to
$\alpha_1$.  We choose local-unitary coordinates in which $\alpha_1$ is in the
$\sigma_3$ direction and the Bloch vector lies in the
$\operatorname{span}\{\sigma_1,\sigma_2\}$ plane; in these coordinates,
set $r_1:=\sqrt{\sum_{i=1}^3(c_i^1)^2}$; thus
$(c_1^1)^2+(c_2^1)^2=r_1^2$.
Similarly, \eqref{VNI1} is reduced to a matrix unitarily equivalent to
\[
\frac12\bigl(
I\otimes I
+r_1 \sigma_1\otimes I
\pm \lambda_1\sigma_3\otimes\sigma_3
\bigr),
\]
and Lemma~\ref{LL31} yields
\[
1\ge \lambda_1^2+r_1^2.
\]

\subsubsection*{Case (4): $t=0$.}
If $\rho_2>0$, equation~\eqref{VNI1} gives
$\tilde C_{2\to1}[\mathcal M[S]]=\rho_1\otimes I\ge0$.
If $\rho_2$ is singular, the replacement channel constructed above
has the same positive Choi matrix but is not denoted by
$\tilde C_{2\to1}[\mathcal M[S]]$.
\end{proof}

\subsection{Proof of Lemma~\ref{LPauli}}\Label{app:pauli-family-calculations}
We consider the Pauli channel
\begin{align}
\Lambda(\rho):=
\sum_{j=0}^3 p_j \sigma_j \rho \sigma_j .
\end{align}
Its Choi matrix $C_{1\to 2}$ is given as
\begin{align}
 C_{1\to 2}^{T_1}
= I \otimes \frac{I}{2}
+ \sum_{j=1}^3
\frac{\lambda_j}{2} \sigma_j \otimes \sigma_j,
\Label{HJ51A}
\end{align}
where $\lambda_j:=2p_0+2p_j-1$ for $j=1,3$ and
$\lambda_2:=-2p_0-2p_2+1$.
Assume that the initial state is
\[
\rho_1 =\frac{1}{2}(I+ \kappa \sigma_1).
\]

We first compute ${\cal R}[{\cal M}[S]]$.
A direct calculation gives
\begin{align}
&{\cal R}[{\cal M}[S]]
= \rho_1 \otimes \frac{I}{2}
+\frac{\kappa}{2}\frac{\lambda_1}{2} I \otimes \sigma_1
+ \frac{1}{2}\sum_{j=1}^3
\frac{\lambda_j}{2}\sigma_j \otimes \sigma_j \notag\\
=&
\frac{I}{2} \otimes
\Bigl(\frac{I}{2}
+\kappa \frac{\lambda_1}{2}  \sigma_1\Bigr)
+\frac{\kappa}{2} \sigma_1 \otimes \frac{ I}{2}
+ \frac{1}{2}\sum_{j=1}^3
\frac{\lambda_j}{2}\sigma_j \otimes \sigma_j \notag\\
=&
\frac{I}{2} \otimes
\Bigl(\frac{I}{2}
+\kappa \frac{\lambda_1}{2}  \sigma_1\Bigr)
+\frac{\kappa}{2} \sigma_1 \otimes \frac{ I}{2}
+
\frac{\lambda_1}{4}\sigma_1 \otimes \sigma_1
+ \frac{1}{2}\sum_{j=2}^3
\frac{\lambda_j}{2}\sigma_j \otimes \sigma_j \notag\\
=&
\frac{I}{2} \otimes
\Bigl(\frac{I}{2}
+\kappa \frac{\lambda_1}{2} \sigma_1\Bigr)
+\frac{1}{2}\sigma_1 \otimes
\Bigl(\frac{\kappa}{2}  I
+\frac{\lambda_1}{2}
\sigma_1 \Bigr)
+ \frac{1}{2}\sum_{j=2}^3
\frac{\lambda_j}{2}\sigma_j \otimes \sigma_j .
\Label{HJ58A}
\end{align}
Hence,
\begin{align}
\rho_2[{\cal M}[S]]
=
\frac{I}{2}
+\kappa \frac{\lambda_1}{2}\sigma_1 .
\Label{HJ58B}
\end{align}

Since
\[
(a I+b\sigma_1)\circ (c I+d\sigma_1)
= (ac+bd) I+(ad+bc)\sigma_1,
\]
for $|\kappa|<1$ we have
\begin{align}
\tilde{C}_{2\to 1}[{\cal M}[S]]^{T_2}
=&
 \frac{I}{2} \otimes I
+\frac{1}{2}\sigma_1\otimes \Big(
\frac{\kappa(\lambda_1^2 -1)}{\kappa^2\lambda_1^2-1}I
+\frac{ \lambda_1 (\kappa^2 -1)}{\kappa^2\lambda_1^2-1}
\sigma_1\Big) \notag\\
&
+ \sum_{j=2}^3
\frac{\lambda_j}{2}\sigma_j \otimes \sigma_j \notag\\
=&
 \frac{I}{2}\otimes I
+\frac{\kappa(\lambda_1^2 -1)}{2(\kappa^2\lambda_1^2-1)}
\sigma_1 \otimes I
+
\frac{ \lambda_1 (\kappa^2 -1)}{2(\kappa^2\lambda_1^2-1)}
\sigma_1\otimes \sigma_1 \notag\\
&
+ \sum_{j=2}^3
\frac{\lambda_j}{2}
\sigma_j \otimes \sigma_j .
\Label{HJ53A}
\end{align}
When $|\kappa|=1$ and $|\lambda_1|<1$, we have
\begin{align}
\tilde{C}_{2\to 1}[{\cal M}[S]]^{T_2}
=&
 \frac{I}{2} \otimes I
+\frac{\kappa}{2}\sigma_1\otimes I
+ \sum_{j=2}^3
\frac{\lambda_j}{2}\sigma_j \otimes \sigma_j \notag\\
=&
 \frac{I}{2}\otimes I
+\frac{\kappa}{2}\sigma_1 \otimes I
+ \sum_{j=2}^3
\frac{\lambda_j}{2}
\sigma_j \otimes \sigma_j .
\Label{HJ5GA}
\end{align}

We can now read off the directional conditions.  If
$|\kappa\lambda_1|=1$, the reverse reconstruction and (D2)--(D3)
are undefined.  In this case (D1) fails, so no reverse memoryless
representation exists.
The condition (D1) holds if and only if
$\kappa \lambda_1=0$,
which is equivalent to the condition
$\lambda_1=0$ or $\kappa=0$.
This proves (i).

On $|\kappa\lambda_1|<1$, (D2) holds if and only if
$\kappa \lambda_1=0$ or
\[
\frac{ \lambda_1 (\kappa^2 -1)}{2(\kappa^2\lambda_1^2-1)}
=0,
\]
which is equivalent to the condition
$\lambda_1=0$ or $\kappa=0$, or $|\kappa|=1$ and
$|\lambda_1|<1$.
This proves (ii).

We next prove (iii), namely the characterization of (D3) in the interior case
$|\kappa|<1$.
Applying Lemma~\ref{LL31} to \eqref{HJ53A}, we find that
condition (D3) holds if and only if
\begin{align}
1+ \frac{\kappa(\lambda_1^2 -1)}{\kappa^2\lambda_1^2-1}
+\frac{ \lambda_1 (\kappa^2 -1)}{\kappa^2\lambda_1^2-1}
&\ge 0,\Label{CBQ1}\\
1- \frac{\kappa(\lambda_1^2 -1)}{\kappa^2\lambda_1^2-1}
-\frac{ \lambda_1 (\kappa^2 -1)}{\kappa^2\lambda_1^2-1}
&\ge 0,
\Label{CBQ2}\\
\left(\frac{(\kappa^2  \lambda_1 -1)( \lambda_1 -1)}{\kappa^2\lambda_1^2-1}\right)^2
=
\left(1-
\frac{ \lambda_1 (\kappa^2 -1)}{\kappa^2\lambda_1^2-1}
\right)^2
& \ge
\left(\frac{\kappa(\lambda_1^2 -1)}{\kappa^2\lambda_1^2-1}\right)^2
+(\lambda_2+\lambda_3)^2 ,
\Label{CBV1}\\
\left(\frac{(\kappa^2  \lambda_1 -1)( \lambda_1 +1)}{\kappa^2\lambda_1^2-1}\right)^2
=
\left(1+
\frac{ \lambda_1 (\kappa^2 -1)}{\kappa^2\lambda_1^2-1}
\right)^2
& \ge
\left(\frac{\kappa(\lambda_1^2 -1)}{\kappa^2\lambda_1^2-1}\right)^2
+(\lambda_2-\lambda_3)^2 .
\Label{CBV2}
\end{align}

For completeness, the two other diagonal conditions of
Lemma~\ref{LL31} also hold automatically here.  Write
$u=\kappa(1-\lambda_1^2)/(1-\kappa^2\lambda_1^2)$ and
$m=\lambda_1(1-\kappa^2)/(1-\kappa^2\lambda_1^2)$.
The four diagonal entries are $1+m+u=(1+\kappa)(1+\lambda_1)/
(1+\kappa\lambda_1)$,
$1+m-u=(1-\kappa)(1+\lambda_1)/(1-\kappa\lambda_1)$,
$1-m+u=(1+\kappa)(1-\lambda_1)/(1-\kappa\lambda_1)$,
and $1-m-u=(1-\kappa)(1-\lambda_1)/(1+\kappa\lambda_1)$.
All are nonnegative for $|\kappa|<1$ and $|\lambda_1|\le1$.
In particular, \eqref{CBQ1} and \eqref{CBQ2} impose no
additional restriction.  This argument also covers $\lambda_1=0$
and negative $\lambda_1$, without dividing by $\lambda_1$.
Hence, for $|\kappa|<1$, condition (D3) is equivalent exactly to
\eqref{CBQ1}--\eqref{CBV2}.
This proves (iii).

Finally, we treat $|\kappa|=1$ with $|\lambda_1|<1$ and prove (iv).
In this case $s_1=\kappa$, $s_2=\tau_1=0$ and $|\kappa|=1$,
so all four diagonal conditions in Lemma~\ref{LL31} hold for
\eqref{HJ5GA}.  Thus
Lemma~\ref{LL31} guarantees that
condition (D3) holds if and only if
\begin{align}
1\ge 1 +
\max
\bigl((\lambda_2+\lambda_3)^2,(\lambda_2-\lambda_3)^2\bigr),
\Label{CBV2Y}
\end{align}
which is equivalent to $\lambda_2=\lambda_3=0$.
This proves (iv).

\subsection{Phase damping channel with $\sigma_1$-basis}\Label{app:phase-damping-sigma1}
We fix the initial state to be $\rho_1 =\frac{1}{2}(I+ \kappa \sigma_1)$.
We consider the phase damping channel with respect to the $\sigma_1$-basis, whose Choi matrix
$C_{1\to 2}$ satisfies the condition $\lambda_1=1$
under the form \eqref{HJ8}.
Then,
the condition (D1) holds if and only if $\kappa=0$.
The condition (D2) holds if and only if $\kappa=0$.
It remains to analyze condition (D3).
We prove the full-rank assertion of Lemma~\ref{LPhaseDampSigmaOne}.
\begin{proof}[Proof of Lemma~\ref{LPhaseDampSigmaOne}]
At $|\kappa|=1$ the second marginal is pure and (D3) is
undefined.  We work throughout on $|\kappa|<1$.
In this case, the condition (D3) is equivalent to
\eqref{CBV1} and \eqref{CBV2}, which are simplified as
\begin{align}
0 & \ge(\lambda_2+\lambda_3)^2 \\
4=(\frac{2(\kappa^2  -1)}{\kappa^2-1 })^2
& \ge (\lambda_2-\lambda_3)^2 .\Label{CBV2-phase1}
\end{align}
Thus, $\lambda_2=-\lambda_3$ and $|\lambda_3|\le1$.
Equality is included because positive semidefiniteness allows the
endpoint.
\end{proof}

\subsection{Phase damping channel with $\sigma_3$-basis}\Label{app:phase-damping-sigma3}
Next, with the same fixed initial state $\rho_1 =\frac{1}{2}(I+ \kappa \sigma_1)$,
we consider the phase damping channel with respect to the $\sigma_3$-basis, whose Choi matrix
$C_{1\to 2}$ satisfies the condition $\lambda_3=1$
under the form \eqref{HJ8}.
Additionally, we assume
$\lambda_2=-\lambda_1$ for simplicity.
The relevant directional criteria are those of Lemma~\ref{LPauli};
they do not depend on $\lambda_3$.  We prove the full-rank (D3)
criterion stated in Lemma~\ref{LPhaseDampSigmaThree}.
\begin{proof}[Proof of Lemma~\ref{LPhaseDampSigmaThree}]
If $|\kappa|=1$ and $|\lambda_1|<1$, the reconstruction exists,
but Lemma~\ref{LPauli} requires $\lambda_2=\lambda_3=0$ for
(D3), contrary to $\lambda_3=1$.  If
$|\kappa\lambda_1|=1$, (D3) is undefined and (D1) fails.
It remains to consider $|\kappa|<1$.
In this case, the condition (D3) is equivalent to
\eqref{CBV1} and \eqref{CBV2}, which are simplified as
\begin{align}
(\frac{(\kappa^2  \lambda_1 -1)( \lambda_1 -1)}{\kappa^2\lambda_1^2-1 })^2
=(1-
\frac{ \lambda_1 (\kappa^2 -1)}{\kappa^2\lambda_1^2-1 }
)^2
& \ge (\frac{\kappa(\lambda_1^2 -1)}{\kappa^2\lambda_1^2-1  })^2
+(1-\lambda_1)^2 \Label{CCV1}\\
(\frac{(\kappa^2  \lambda_1 -1)( \lambda_1 +1)}{\kappa^2\lambda_1^2-1 })^2
=(1+
\frac{ \lambda_1 (\kappa^2 -1)}{\kappa^2\lambda_1^2-1 }
)^2
& \ge (\frac{\kappa(\lambda_1^2 -1)}{\kappa^2\lambda_1^2-1  })^2
+(1+\lambda_1)^2 .\Label{CCV2}
\end{align}

First, we discuss the condition \eqref{CCV1}
by considering two cases.

Case 1: $\lambda_1 = 1$.
If $\lambda_1 = 1$, then $1-\lambda_1 = 0$. The inequality \eqref{CCV1} becomes:
\[\left(\frac{0}{\kappa^2-1}\right)^2 \ge \left(\frac{0}{\kappa^2-1}\right)^2 + 0^2,\]
which is true for all $\kappa$ as long as
$\kappa^2 <1$.

Case 2: $\lambda_1 \ne 1$.
Multiplying by $(\kappa^2\lambda_1^2-1)^2/(1-\lambda_1)^2$ in \eqref{CCV1}, we have
\[(\kappa^2\lambda_1+1)^2 \ge \kappa^2(\lambda_1+1)^2 + (\kappa^2\lambda_1^2-1)^2\]
Expanding both sides gives
\[0 \ge \kappa^4\lambda_1^4 - \kappa^4\lambda_1^2 - \kappa^2\lambda_1^2 + \kappa^2.\]
Factoring gives
\[0 \ge \kappa^2(\kappa^2\lambda_1^2 - 1)(\lambda_1^2 - 1),\]
which implies $\kappa=0$ or $\lambda_1=-1$
due to the conditions $\lambda_1 \neq 1$ and $\kappa^2 <1$.

Next, we discuss the condition \eqref{CCV2}
by considering two cases.

Case 1: $\lambda_1 = -1$.
If $\lambda_1 = -1$, then $1+\lambda_1 = 0$. The inequality \eqref{CCV2} becomes:
\[\left(\frac{0}{\kappa^2-1}\right)^2 \ge \left(\frac{0}{\kappa^2-1}\right)^2 + 0^2,\]
which is true for all $\kappa$ as long as
$\kappa^2 <1$.

Case 2: $\lambda_1 \ne -1$.
Multiplying by $(\kappa^2\lambda_1^2-1)^2/(1+\lambda_1)^2$ in \eqref{CCV2}, we have
\[(\kappa^2\lambda_1-1)^2 \ge \kappa^2(\lambda_1-1)^2 + (\kappa^2\lambda_1^2-1)^2\]
Expanding both sides gives
\[0 \ge \kappa^4\lambda_1^4 - \kappa^4\lambda_1^2 - \kappa^2\lambda_1^2 + \kappa^2.\]
Factoring out common terms:
\[0 \ge \kappa^2(\lambda_1^2 - 1)(\kappa^2\lambda_1^2 - 1),\]
which implies $\kappa=0$ or $\lambda_1=1$
due to the conditions $\lambda_1 \neq -1$ and $\kappa^2 <1$.

In summary, under $\kappa^2 <1$,
the conditions \eqref{CCV1} and \eqref{CCV2} are equivalent to
the relation $\kappa=0$ or $\lambda_1=\pm 1$.
\end{proof}

\subsection{Proof of Lemma~\ref{LL29}}\Label{app:proof-LL29}
\begin{proof}
For $|\kappa|=1$ and $\lambda=0$, the second marginal is pure,
so (D3) is undefined and (D1) fails.  For $|\kappa|=1$ and
$0<\lambda\le1$, the second marginal is full rank; by
Lemma~\ref{LPauli}, (D3) requires $\lambda_2=\lambda_3=0$,
equivalent to $\lambda=1$.  We now consider $|\kappa|<1$,
which converts (D3) to
the conditions \eqref{CBV1} and \eqref{CBV2}.

{\bf Analysis for \eqref{CBV1}:\quad}
We show that the condition \eqref{CBV1} always holds.
This condition is simplified as
\begin{align}
(1-
\frac{ \lambda_1 (\kappa^2 -1)}{\kappa^2\lambda_1^2-1 }
)^2
& \ge (\frac{\kappa(\lambda_1^2 -1)}{\kappa^2\lambda_1^2-1  })^2
\end{align}
Since $\kappa^2\lambda_1^2-1 <0$, it is converted to
\begin{align}
|(\lambda_1 - 1)(\kappa^2\lambda_1 + 1)|
=|\kappa^2\lambda_1^2-1 - \lambda_1 (\kappa^2 -1)|
 & \ge |\kappa(\lambda_1^2 -1)|.
\end{align}
When $\kappa\ge 0$, it is converted to
\begin{align}
(1-\lambda_1)(\kappa^2\lambda_1 + 1)
 \ge -\kappa(\lambda_1^2 -1)= \kappa(1-\lambda_1)(1+\lambda_1).
\end{align}
That is,
\begin{align}
(1-\kappa )(1-\lambda_1 )(1-\kappa\lambda_1)\ge 0.\Label{GH1}
\end{align}
When $\kappa\le 0$, it is converted to
\begin{align}
(1-\lambda_1)(\kappa^2\lambda_1 + 1)
 \ge \kappa(\lambda_1^2 -1)=-\kappa(1-\lambda_1)(1+\lambda_1).
\end{align}
That is,
\begin{align}
(1+\kappa )(1-\lambda_1)(1+\kappa\lambda_1 )\ge 0. \Label{GH2}
\end{align}
However, the conditions \eqref{GH1} and \eqref{GH2} always hold.

{\bf Analysis for \eqref{CBV2}:\quad}
The condition \eqref{CBV2} is simplified as
\begin{align}
(\frac{(\kappa^2  (1-\lambda)-1)(2-\lambda)}{\kappa^2(1-\lambda)^2-1 })^2
& \ge (\frac{\kappa((1-\lambda)^2 -1)}{\kappa^2(1-\lambda)^2-1  })^2
+4(1-\lambda)^2 ,\Label{CBF2}
\end{align}
which is equivalent to
\begin{align}
\left(\frac{(\kappa^2 (1-\lambda)-1)(2-\lambda)}{\kappa^2(1-\lambda)^2-1}\right)^2 - \left(\frac{\kappa((1-\lambda)^2 -1)}{\kappa^2(1-\lambda)^2-1}\right)^2 \ge 4(1-\lambda)^2.\Label{CBF8}
\end{align}
Hence, we solve the inequality \eqref{CBF8} with respect to $\kappa$
where $0 \le \lambda \le 1$.
To solve the inequality with respect to $\kappa$, we first simplify the expression.
This inequality is defined for all $\kappa$ such that the denominator is not zero, i.e., $\kappa^2(1-\lambda)^2-1 \neq 0$.

We apply the formula
$A^2-B^2 = (A-B)(A+B)$
when
$A = \frac{(\kappa^2 (1-\lambda)-1)(2-\lambda)}{\kappa^2(1-\lambda)^2-1}$,
$B = \frac{\kappa((1-\lambda)^2 -1)}{\kappa^2(1-\lambda)^2-1} = \frac{\kappa(\lambda^2-2\lambda)}{\kappa^2(1-\lambda)^2-1} = \frac{-\kappa\lambda(2-\lambda)}{\kappa^2(1-\lambda)^2-1}$.
Since $A-B = \frac{(2-\lambda)}{\kappa^2(1-\lambda)^2-1} \left( \kappa^2(1-\lambda)-1 + \kappa\lambda \right)$,
$A+B = \frac{(2-\lambda)}{\kappa^2(1-\lambda)^2-1} \left( \kappa^2(1-\lambda)-1 - \kappa\lambda \right)$,
the LHS of \eqref{CBF8} equals
\[\frac{(2-\lambda)^2}{(\kappa^2(1-\lambda)^2-1)^2} \left( \kappa^2(1-\lambda)-1 + \kappa\lambda \right) \left( \kappa^2(1-\lambda)-1 - \kappa\lambda \right).\]
The product of the two parenthetical factors simplifies as follows:
$(\kappa^2(1-\lambda)-1)^2 - (\kappa\lambda)^2 = ((\kappa-1)(\kappa(1-\lambda)+1))((\kappa+1)(\kappa(1-\lambda)-1)) = (\kappa^2-1)(\kappa^2(1-\lambda)^2-1)$.
Consequently, the product becomes
\[\frac{(2-\lambda)^2 (\kappa^2-1)(\kappa^2(1-\lambda)^2-1)}{(\kappa^2(1-\lambda)^2-1)^2} = \frac{(2-\lambda)^2(\kappa^2-1)}{\kappa^2(1-\lambda)^2-1}.\]
Thus, the inequality \eqref{CBF8} simplifies to:
\[\frac{(2-\lambda)^2(\kappa^2-1)}{\kappa^2(1-\lambda)^2-1} \ge 4(1-\lambda)^2.\]
We can rearrange this to:
\[\frac{(2-\lambda)^2(\kappa^2-1) - 4(1-\lambda)^2(\kappa^2(1-\lambda)^2-1)}{\kappa^2(1-\lambda)^2-1} \ge 0\]

The numerator is
\begin{align*}
N =& (\kappa^2-1)(2-\lambda)^2 - 4(1-\lambda)^2(\kappa^2(1-\lambda)^2-1) \\
=& \kappa^2((2-\lambda)^2 - 4(1-\lambda)^4) - ((2-\lambda)^2 - 4(1-\lambda)^2) \\
=& \kappa^2\lambda(3-2\lambda)(2\lambda^2-5\lambda+4) - \lambda(4-3\lambda) \\
=& \lambda \left[ \kappa^2(3-2\lambda)(2\lambda^2-5\lambda+4) - (4-3\lambda) \right].
\end{align*}

So the inequality is:
\[\frac{\lambda \left[ \kappa^2(3-2\lambda)(2\lambda^2-5\lambda+4) - (4-3\lambda) \right]}{\kappa^2(1-\lambda)^2-1} \ge 0.\]
We now analyze this inequality for different values of $\lambda \in [0, 1]$.

Case 1: $\lambda=0$.
The inequality becomes $\frac{0}{\kappa^2-1} \ge 0$, which is $0 \ge 0$. This is true for all $\kappa$ for which the expression is defined, which is $\kappa^2-1 \neq 0$. Thus, for $\lambda=0$,
the condition \eqref{CBV2} holds for every $\kappa\in (-1,1)$.

Case 2: $\lambda=1$.
The numerator is $1[\kappa^2(1)(1) - 1] = \kappa^2-1$. The denominator is $\kappa^2(0)^2-1 = -1$. The inequality is $\frac{\kappa^2-1}{-1} \ge 0$, which implies $\kappa^2-1 \le 0$, or $\kappa^2 \le 1$.
the condition \eqref{CBV2} holds for every $\kappa\in (-1,1)$.

Case 3: $0 < \lambda < 1$.
Since $\lambda>0$, we can divide by it. The inequality becomes:
\[\frac{\kappa^2(3-2\lambda)(2\lambda^2-5\lambda+4) - (4-3\lambda)}{\kappa^2(1-\lambda)^2-1} \ge 0\]
For $0 < \lambda < 1$, the terms $(3-2\lambda)$, $(2\lambda^2-5\lambda+4)$, and $(4-3\lambda)$ are all positive.
Since $\kappa^2<1$, we have $\kappa^2(1-\lambda)^2-1<0$.
Thus, \eqref{CBF8} is equivalent to
\begin{align}
|\kappa| \le \sqrt{\frac{4-3\lambda}{(3-2\lambda)(2\lambda^2-5\lambda+4)}}. \Label{BNF}
\end{align}

\end{proof}

\subsection{Proof of Lemma~\ref{thm:PDM-bloch-positivity} and
Corollary~\ref{cor:PDM-pure-second}}\Label{app:proof-PDM-bloch}
\begin{proof}[Proof of Lemma~\ref{thm:PDM-bloch-positivity}]
We prove positivity without computing the eigenvalues of $R$ explicitly. 
Let $\mu_1,\ldots,\mu_4$ be the eigenvalues of $R$, and set $e_2:=\sum_{i<j}\mu_i\mu_j$, $e_3:=\sum_{i<j<k}\mu_i\mu_j\mu_k$, and $e_4:=\mu_1\mu_2\mu_3\mu_4$.
Then $e_1=\Tr R=1$ and $e_4=\det R$.
Pauli orthogonality and the identities
$\Tr(\sigma_i\sigma_j)=2\delta_{ij}$ and
$\Tr(\sigma_i\sigma_j\sigma_k)=2\mathrm{i}\varepsilon_{ijk}$ give
\begin{align}
\Tr R^2&=\frac{1+F}{4},\\
\Tr R^3&=\frac{1+3F+6a^{\mathsf T}Tb-6\det T}{16}.
\Label{PDM-power-sums}
\end{align}
We now substitute these trace formulas into the expressions for $e_2$ and $e_3$ in terms of $\Tr R$, $\Tr R^2$, and $\Tr R^3$. Using $\Tr R=1$ and $e_4=\det R$, we obtain
\begin{align}
e_2&=\frac{1-\Tr R^2}{2}=\frac{3-F}{8},\\
e_3&=\frac{1-3\Tr R^2+2\Tr R^3}{6}
=\frac{1-F+2a^{\mathsf T}Tb-2\det T}{16},\\
e_4&=\det R.\Label{PDM-elementary}
\end{align}
If $R\ge0$, then $e_2,e_3,e_4\ge0$.  Conversely, $R$ is
Hermitian with four real eigenvalues and $e_1=\Tr R=1$.
The characteristic identity reads
$\det(xI+R)=x^4+e_1x^3+e_2x^2+e_3x+e_4$.
If $e_2,e_3,e_4\ge0$, this expression is positive for every $x>0$.
A negative eigenvalue $-x$ would make this determinant zero.
Therefore $R\ge0$, proving the equivalence.
\end{proof}
\begin{proof}[Proof of Corollary~\ref{cor:PDM-pure-second}]
Since $\|b\|_2=1$, Bob's input state
$\rho_2[P]$ is pure.  Under (D1), every Alice's choice $y_1$ satisfies 
\begin{align}
\rho_2[P]=\sum_{x_1}p(x_1\mid y_1)\rho_{2\mid x_1,y_1}.
\Label{PDM-pure-mixture}
\end{align}
Because $\rho_2[P]$ is pure, it cannot be a mixture of different states. Thus, whenever $p(x_1\mid y_1)>0$, we have $\rho_{2\mid x_1,y_1}=\rho_2[P]$,
which implies
\begin{align}
p(x_1,x_2\mid y_1,y_2)=p(x_1\mid y_1)
\Tr(E^2_{x_2\mid y_2}\rho_2[P]).\Label{PDM-pure-factor}
\end{align}
The same holds when either player does nothing, because
$S\in{\cal S}_{Q,1\to2}$ also satisfies $X_1-Y_1-Y_2$.
The factorization of the observed probabilities shows that $\rho_1[P]\otimes\rho_2[P]$ satisfies the compatibility condition \eqref{VBSA}. Since (A1) makes the compatible PDM unique, we have $R=\rho_1[P]\otimes\rho_2[P]$. Hence $T=ab^{\mathsf T}$ and $K=0$.
Since $S\in{\cal S}_{Q,1\to2}$, Theorem~\ref{TH4I} gives $X_1-Y_1-Y_2$, while (D1) gives $Y_1-Y_2-X_2$. Together with $R=\rho_1[P]\otimes\rho_2[P]\ge0$, these conditions imply $P\in{\cal M}[{\cal S}_P]$ by \eqref{FG1VS}.
\end{proof}
\subsection{Proof of Lemma~\ref{LL29-A}}\Label{app:proof-LL29A}
\begin{proof}
\textbf{Step 1: standard form and its domain.}
By \eqref{DSJK},
$R=\rho_1\otimes\rho_2+\frac14\sum_{j=1}^t
\lambda_j\alpha_j\otimes\beta_j$, so
$K=\sum_{j=1}^t\lambda_jv_jw_j^{\mathsf T}$.
For $t\ge1$, the second marginal is full rank: otherwise
Lemma~\ref{LBU2} would imply $t=0$.  Lemmas~\ref{LBU1} and
\ref{LBU2} give $K^{\mathsf T}a=0$ and $Kb=0$.
Local unitaries preserve positivity, singular values, $r_1$, and $r_2$.

\textbf{Step 2: $t=3$.}
Here $a=b=0$.  The relative orientation of the two singular frames
gives, up to local unitaries and relabeling, either
$4R=I\otimes I+\lambda_1\sigma_1\otimes\sigma_1
+\lambda_2\sigma_2\otimes\sigma_2
+\lambda_3\sigma_3\otimes\sigma_3$ or the form with the second
correlation term negated.  In the first form the eigenvalues of $4R$ are
$1+\lambda_1-\lambda_2+\lambda_3$,
$1-\lambda_1+\lambda_2+\lambda_3$,
$1+\lambda_1+\lambda_2-\lambda_3$, and
$1-\lambda_1-\lambda_2-\lambda_3$.
Their nonnegativity is equivalent to $1\ge\lambda_1+\lambda_2+\lambda_3$.
In the second form the eigenvalues are
$1+\lambda_1+\lambda_2+\lambda_3$,
$1-\lambda_1-\lambda_2+\lambda_3$,
$1+\lambda_1-\lambda_2-\lambda_3$, and
$1-\lambda_1+\lambda_2-\lambda_3$.
Their nonnegativity is precisely the maximum condition in part (1).
These are eigenvalues of $4R$, not of a Choi matrix.

\textbf{Step 3: $t=2$.}
Both local Bloch vectors lie in the respective one-dimensional kernels.
After local unitaries,
\begin{align}
4R=(I+r_1\sigma_1)\otimes(I+r_2\sigma_1)
+\epsilon\lambda_1\sigma_2\otimes\sigma_2
+\lambda_2\sigma_3\otimes\sigma_3,
\qquad\epsilon\in\{-1,1\}.\Label{PDM-t2-form}
\end{align}
Apply Lemma~\ref{LL31} with $s_1=r_1$, $s_2=r_2$,
$\tau_1=r_1r_2$, $\tau_2=\epsilon\lambda_1$,
$\tau_3=\lambda_2$.  All four diagonal conditions
\eqref{AP1}, \eqref{AP1b}, \eqref{AP2}, and \eqref{AP2b}
are automatic: their left-hand sides are, respectively,
$(1+r_1)(1+r_2)$, $(1-r_1)(1-r_2)$,
$(1-r_1)(1+r_2)$, and $(1+r_1)(1-r_2)$,
all nonnegative for $0\le r_1,r_2\le1$.
The difference between each squared left-hand side in \eqref{BP1},
\eqref{BP2} and its corresponding squared local-vector term equals
$(1-r_1^2)(1-r_2^2)$.  The remaining two conditions require this
product to exceed both $(\epsilon\lambda_1+\lambda_2)^2$ and
$(\epsilon\lambda_1-\lambda_2)^2$.  Their maximum is
$(\lambda_1+\lambda_2)^2$, proving \eqref{BP1T}.

\textbf{Step 4: rank-one correlation, $t=1$.}
Choose local coordinates with the single correlation term
$\lambda_1\sigma_3\otimes\sigma_3$.  The two kernel-alignment
relations place the local Bloch vectors in the respective planes
perpendicular to their correlation directions.  Independent rotations
around these directions put them along $r_1e_1$ and $r_2e_1$.
Thus $4R=(I+r_1\sigma_1)\otimes(I+r_2\sigma_1)
+\lambda_1\sigma_3\otimes\sigma_3$.  Apply Lemma~\ref{LL31} with
$s_1=r_1$, $s_2=r_2$, $\tau_1=r_1r_2$,
$\tau_2=0$, and $\tau_3=\lambda_1$.
The four diagonal entries are the nonnegative products in Step 3.
The remaining two conditions are
\begin{align}
(1-r_1r_2)^2&\ge(r_1-r_2)^2+\lambda_1^2,\\
(1+r_1r_2)^2&\ge(r_1+r_2)^2+\lambda_1^2.
\Label{PDM-t1-two-blocks}
\end{align}
Subtracting the squared local-vector term from the left-hand side
of either line gives $(1-r_1^2)(1-r_2^2)$, so the two conditions
are equivalent to \eqref{BP2E}.

\textbf{Step 5: rank-zero correlation, $t=0$.}
Equation~\eqref{DSJK} gives $R=\rho_1\otimes\rho_2\ge0$.
This proves all cases.
\end{proof}

\subsection{Proof of Lemma~\ref{LEM-ExB}}\Label{app:proof-ExB}
\subsubsection{Failure of the positivity condition $\hat C_{1\to 2}[P]\ge 0$}
\begin{proof}
Conditioned on Alice's measurement choice $Y_1=y_1$ and outcome $X_1=x_1$,
the memory qubit is projected onto the complex-conjugate state
$|x_1,y_1\rangle^*$, where $^*$ expresses the complex conjugate.
Hence Bob receives the state
$|x_1,y_1\rangle^*\langle x_1,y_1|^*
=(|x_1,y_1\rangle\langle x_1,y_1|)^T$
for the same conditioning.
In the qubit-Pauli setting, this coincides with the action of the transpose map
$\rho\mapsto \rho^{\mathsf T}$ on the post-measurement state
$|x_1,y_1\rangle\langle x_1,y_1|$.
Therefore, under the reconstruction formula of Theorem~\ref{TH8},
the matrix $\hat C_{1\to2}[P]$ is exactly the Choi matrix of the transpose map.
With our Choi convention, this matrix is the swap operator
$\sum_{m,n=0}^1 |m\rangle\langle n| \otimes |n\rangle\langle m|$,
whose eigenvalues are $1,1,1,-1$.
Hence $\hat C_{1\to2}[P]$ is not positive.
\end{proof}

\subsubsection{Verification of Condition \emph{(C1)}}
\begin{proof}
We verify condition \emph{(C1)} by using the two-qubit Pauli form given in
Proposition~\ref{PROP:C1-twoqubit}.  
Conditioned on Alice's measurement choice $Y_1=y_1$ and outcome $X_1=x_1$,
the memory qubit is projected onto the complex-conjugate state
$\ket{x_1,y_1}^{*}$.
Hence, when Bob measures the Pauli observable $\sigma_{y_2}$, the conditional
distribution is
\begin{align}
P_{X_2\mid X_1,Y_1,Y_2}(x_2\mid x_1,y_1,y_2)
=
\bra{x_2,y_2}
\left(\ket{x_1,y_1}^{*}\!\bra{x_1,y_1}^{*}\right)
\ket{x_2,y_2}.
\Label{ExB-cond}
\end{align}
Alice's reduced state is maximally mixed. Therefore, her local statistics
are unbiased for every Pauli measurement:
\begin{align}
P_{X_1\mid Y_1}(x_1\mid y_1)=\frac12,
\qquad
\eta_{y_1}
:=
\sum_{x_1\in\{0,1\}}(-1)^{x_1}
P_{X_1\mid Y_1}(x_1\mid y_1)
=0
\Label{ExB-eta}
\end{align}
for all $y_1\in\{1,2,3\}$. In particular, all of Alice's Pauli-measurement
outcomes have strictly positive probability, and hence
condition~\emph{(C0)} holds. Therefore, the correction term in
Proposition~\ref{PROP:C1-twoqubit},
\begin{align}
\sum_{y_1'\in\{1,2,3\}}
\eta_{y_1'}
\sum_{x_1'\in\{0,1\}}(-1)^{x_1'}
P_{X_2\mid X_1,Y_1,Y_2}(x_2\mid x_1',y_1',y_2),
\end{align}
vanishes identically. Therefore, the right-hand side of \eqref{C1-2q}
reduces to
$2P_{X_2\mid Y_1,Y_2}(x_2\mid 0,y_2)$.

For each fixed $y_1$, the two states conditionally received by Bob are the
complex conjugates of the two eigenstates of $\sigma_{y_1}$. The sum of
their density matrices is therefore the identity operator. It follows from
\eqref{ExB-cond} that
\begin{align}
\sum_{x_1\in\{0,1\}}
P_{X_2\mid X_1,Y_1,Y_2}(x_2\mid x_1,y_1,y_2)
=1.
\Label{ExB-C1-LHS}
\end{align}
When Alice does nothing, Bob's reduced state is maximally mixed. Hence
\begin{align}
P_{X_2\mid Y_1,Y_2}(x_2\mid 0,y_2)=\frac12,
\end{align}
and consequently
\begin{align}
2P_{X_2\mid Y_1,Y_2}(x_2\mid 0,y_2)=1.
\Label{ExB-C1-RHS}
\end{align}
Thus, \eqref{ExB-C1-LHS} and \eqref{ExB-C1-RHS} show that the two sides of
\eqref{C1-2q} coincide for all $y_1,y_2\in\{1,2,3\}$ and
$x_2\in\{0,1\}$. Hence condition \emph{(C1)} holds.
\end{proof}

\subsection{Proof of Lemma~\ref{LEM-Ex1prime}}\Label{app:proof-Ex1prime}
\subsubsection{Verification of the positivity condition $\hat C_{1\to 2}[P]\ge 0$}
\begin{proof}
We first rewrite the reconstruction formula \eqref{VBTR} in the present
two-qubit Pauli setting.
Here we have
\[
G_0=I,\qquad G_j=\sigma_j,\qquad
H_0=\frac{I}{2},\qquad H_j=\frac{\sigma_j}{2}\qquad (j=1,2,3),
\]
and the coefficients are $g_{x,j}=(-1)^x$ for $x\in\{0,1\}$.
Since in Example~\ref{Ex1prime} the classical memory $X$ is chosen uniformly,
Alice's local distribution is unbiased for every Pauli choice, namely
\[
P_{X_1\mid Y_1}(0\mid y_1)=P_{X_1\mid Y_1}(1\mid y_1)=\frac12
\qquad (y_1=1,2,3).
\]
Hence the correction term in \eqref{VBTR} vanishes, and the reconstructed
matrix reduces to
\begin{align}
\hat C_{1\to2}[P]
=
\frac12\,I\otimes I
+\frac12\sum_{j=1}^3 b_j\,I\otimes \sigma_j
+\frac14\sum_{i,j=1}^3 c_{ij}\,\sigma_i^{T}\otimes \sigma_j,
\Label{Ex1prime-recon}
\end{align}
where
\begin{align}
b_j
&:=
\sum_{x_2\in\{0,1\}}(-1)^{x_2}
P_{X_2\mid Y_1,Y_2}(x_2\mid 0,j),
\Label{Ex1prime-bj}
\\
c_{ij}
&:=
\sum_{x_1,x_2\in\{0,1\}}(-1)^{x_1+x_2}
P_{X_2\mid X_1,Y_1,Y_2}(x_2\mid x_1,i,j).
\Label{Ex1prime-cij}
\end{align}

We now compute these coefficients.
When Alice does nothing, Bob receives the state
\[
D_{1/2}(|0,3\rangle\langle 0,3|)
=
\frac12|0,3\rangle\langle 0,3|+\frac14 I,
\]
so
\[
b_1=b_2=0,\qquad b_3=\frac12.
\]
Next, conditioned on Alice's measurement, the relevant conditional states at
Bob are as follows.

\begin{itemize}
\item If $Y_1=1$, then the two values of the memory variable $X$ lead to the
same post-unitary state, so conditioned on $X_1=x_1$ Bob receives
$D_{1/2}(|x_1,1\rangle\langle x_1,1|)$.
Hence
\[
P_{X_2\mid X_1,Y_1,Y_2}(x_2\mid x_1,1,1)
=
\frac34\,\delta_{x_2,x_1}+\frac14\,(1-\delta_{x_2,x_1}),
\]
while for $Y_2=2,3$ the distribution is unbiased.
Therefore
\[
c_{11}=1,\qquad c_{12}=c_{13}=0.
\]

\item If $Y_1=2$, then averaging over $X\in\{0,1\}$ yields the maximally mixed
state at Bob, so
\[
c_{21}=c_{22}=c_{23}=0.
\]

\item If $Y_1=3$, then conditioned on Alice's outcome one has $X=X_1$, and the
subsequent unitary $\sigma_X$ always maps Alice's output to $|0,3\rangle$.
Thus Bob receives the fixed state
$D_{1/2}(|0,3\rangle\langle 0,3|)$, independently of $X_1$, and hence
\[
c_{31}=c_{32}=c_{33}=0.
\]
\end{itemize}

Substituting these coefficients into \eqref{Ex1prime-recon}, and noting that
$\sigma_1^{T}=\sigma_1$, we obtain
\begin{align}
\hat C_{1\to2}[P]
=
\frac12\Bigl(
I\otimes I
+\frac12\,I\otimes \sigma_3
+\frac12\,\sigma_1\otimes \sigma_1
\Bigr).
\Label{Ex1prime-C}
\end{align}

In the computational basis
$\{|00\rangle,|01\rangle,|10\rangle,|11\rangle\}$,
the matrix \eqref{Ex1prime-C} is written as
\begin{align}
\hat C_{1\to2}[P]
=
\begin{pmatrix}
\frac34 & 0 & 0 & \frac14\\
0 & \frac14 & \frac14 & 0\\
0 & \frac14 & \frac34 & 0\\
\frac14 & 0 & 0 & \frac14
\end{pmatrix}.
\Label{Ex1prime-C-matrix}
\end{align}
Its eigenvalues are
\begin{align}
\frac12-\frac{\sqrt2}{4},\quad
\frac12-\frac{\sqrt2}{4},\quad
\frac12+\frac{\sqrt2}{4},\quad
\frac12+\frac{\sqrt2}{4},
\end{align}
all of which are strictly positive.
Hence, $\hat C_{1\to2}[P]\ge 0$.
\end{proof}

\subsubsection{Failure of Condition \emph{(C1)}}
\begin{proof}
By Proposition~\ref{PROP:C1-twoqubit}, in the two-qubit Pauli setting,
condition \emph{(C1)} is equivalent to
\begin{align}
\sum_{x_1\in\{0,1\}}
P_{X_2\mid X_1,Y_1,Y_2}(x_2\mid x_1,y_1,y_2)
&=
2P_{X_2\mid Y_1,Y_2}(x_2\mid0,y_2)
\notag\\
&\quad
-\sum_{y_1'\in\{1,2,3\}}\eta_{y_1'}
\sum_{x_1'\in\{0,1\}}(-1)^{x_1'}
P_{X_2\mid X_1,Y_1,Y_2}(x_2\mid x_1',y_1',y_2),
\Label{eq:C1-2q-Ex1prime}
\end{align}
for all $y_1,y_2\in\{1,2,3\}$ and $x_2\in\{0,1\}$, where
$\eta_{y_1'}
:=
\sum_{x_1\in\{0,1\}}(-1)^{x_1}P_{X_1\mid Y_1}(x_1\mid y_1')$.

We show that \eqref{eq:C1-2q-Ex1prime} fails for
$(y_1,y_2,x_2)=(1,3,0)$.
Since the classical memory $X$ is chosen uniformly,
Alice's local statistics are unbiased for every Pauli choice, and hence
\begin{align}
\eta_{1}=\eta_{2}=\eta_{3}=0.
\Label{Ex1prime-eta}
\end{align}
Therefore the correction term in \eqref{eq:C1-2q-Ex1prime} vanishes identically.

Next, when Alice measures $Y_1=1$ and Bob measures $Y_2=3$,
Bob receives a depolarized $\sigma_1$-eigenstate.
Hence Bob's $\sigma_3$-measurement is unbiased, so
\[
P_{X_2\mid X_1,Y_1,Y_2}(0\mid x_1,1,3)=\frac12
\qquad (x_1=0,1).
\]
Therefore the left-hand side of \eqref{eq:C1-2q-Ex1prime} equals
\begin{align}
\sum_{x_1\in\{0,1\}}
P_{X_2\mid X_1,Y_1,Y_2}(0\mid x_1,1,3)
=
1.
\Label{Ex1prime-LHS}
\end{align}

On the other hand, when Alice does nothing ($Y_1=0$),
Bob receives the depolarized state
\[
D_{1/2}(|0,3\rangle\langle 0,3|)
=
\frac12|0,3\rangle\langle 0,3|+\frac14 I.
\]
Hence
\begin{align}
P_{X_2\mid Y_1,Y_2}(0\mid 0,3)=\frac34.
\Label{Ex1prime-Bobmarginal}
\end{align}
Since the correction term vanishes by \eqref{Ex1prime-eta},
the right-hand side of \eqref{eq:C1-2q-Ex1prime} becomes
\begin{align}
2P_{X_2\mid Y_1,Y_2}(0\mid 0,3)=\frac32.
\Label{Ex1prime-RHS}
\end{align}
Comparing \eqref{Ex1prime-LHS} and \eqref{Ex1prime-RHS}, we obtain
$1\neq \frac32$.
Thus \eqref{eq:C1-2q-Ex1prime} is violated, and therefore
Example~\ref{Ex1prime} does not satisfy \emph{(C1)}.
\end{proof}

\subsection{Proof of Lemma~\ref{LEM-Ex1}}\Label{app:proof-Ex1}
\subsubsection{Failure of Condition \emph{(C1)}}
\begin{proof}
We employ the notation in Subsection \ref{app:proof-Ex1prime}.
That is, we use \eqref{eq:C1-2q-Ex1prime} as a condition 
equivalent to the condition \emph{(C1)}.
We show that the relation \eqref{eq:C1-2q-Ex1prime} fails for $
(y_1,y_2,x_2)=(1,3,0)$.
The important point is that, for this choice of $(y_1,y_2,x_2)$,
the relevant probabilities are independent of the mixing parameter $\lambda$.
Indeed, when $Y_1=1$, Alice's post-measurement state is a $\sigma_1$-eigenstate,
and after Charlie applies $\sigma_X$, Bob still receives a $\sigma_1$-eigenstate.
Hence Bob's $\sigma_3$-measurement is unbiased, regardless of the value of $X$
and therefore regardless of $\lambda$.

Thus,
\[
P_{X_2\mid X_1,Y_1,Y_2}(0\mid x_1,1,3)=\frac12
\qquad (x_1=0,1),
\]
and the left-hand side is
\[
\sum_{x_1\in\{0,1\}}
P_{X_2\mid X_1,Y_1,Y_2}(0\mid x_1,1,3)=1.
\]

Next, when Alice does nothing ($Y_1=0$), Charlie sends
$\sigma_X|(-1)^X,3\rangle$ to Bob.
But
\[
\sigma_X|(-1)^X,3\rangle = |0,3\rangle
\qquad (X=0,1),
\]
so Bob always receives the state $|0,3\rangle$, again independently of $\lambda$.
Hence
\[
P_{X_2\mid Y_1,Y_2}(0\mid 0,3)=1,
\]
and therefore the first term on the right-hand side equals
\[
2P_{X_2\mid Y_1,Y_2}(0\mid 0,3)=2.
\]

Finally, for every $y_1'\in\{1,2,3\}$, the correction term vanishes:
\[
\sum_{x_1'\in\{0,1\}}(-1)^{x_1'}
P_{X_2\mid X_1,Y_1,Y_2}(0\mid x_1',y_1',3)=0.
\]
Indeed, for $y_1'=1,2$, Bob's $\sigma_3$-measurement is unbiased, so the two
probabilities are both $1/2$.
For $y_1'=3$, Bob always obtains $x_2=0$, independently of $x_1'$, so the two terms are both $1$ and their alternating sum is
$1-1=0$.
Therefore the right-hand side is $2$, whereas the left-hand side is $1$.
This contradiction shows that condition \emph{(C1)} fails for the strategy of
Example~\ref{Ex1}.
\end{proof}

\subsubsection{Failure of the positivity condition $\hat C_{1\to 2}[P]\ge 0$}
\begin{proof}
Recall that, in Example~\ref{Ex1}, Charlie uses a hidden classical variable
$X\in\{0,1\}$, chooses $X=0$ with probability $\lambda\in(0,1)$ and
$X=1$ with probability $1-\lambda$, inputs the state
$|(-1)^X,3\rangle$ to Alice, and then applies the unitary $\sigma_X$ to
Alice's output before sending the system to Bob.

We first compute the local statistics on Alice's side.
For $Y_1=1,2$, Alice's outcome is unbiased:
\[
P_{X_1\mid Y_1}(0\mid y_1)=P_{X_1\mid Y_1}(1\mid y_1)=\frac12
\qquad (y_1=1,2).
\]
For $Y_1=3$, we have
\[
P_{X_1\mid Y_1}(0\mid 3)=\lambda,
\qquad
P_{X_1\mid Y_1}(1\mid 3)=1-\lambda.
\]
Hence the reconstructed initial state is
\[
\rho_1[P]
=
\frac12\bigl(I+(2\lambda-1)\sigma_3\bigr).
\]

Next, when Alice does nothing ($Y_1=0$), Bob always receives the state
$|0,3\rangle$.
Therefore,
\[
\rho_2[P]
=
|0,3\rangle\langle 0,3|
=
\frac12(I+\sigma_3).
\]

We now determine the conditional distribution
$P_{X_2\mid X_1,Y_1,Y_2}$.
From the definition of Example~\ref{Ex1}, Bob's conditional statistics are
as follows:

\begin{itemize}
\item If $Y_1=1$, then
\[
Y_2=1:\ x_2=x_1 \ \text{with probability }1,
\qquad
Y_2=2,3:\ \text{uniform}.
\]

\item If $Y_1=2$, then
\[
Y_2=2:\ x_2=x_1 \ \text{with probability }\lambda,
\quad
x_2=1-x_1 \ \text{with probability }1-\lambda,
\]
and for $Y_2=1,3$ the distribution is uniform.

\item If $Y_1=3$, then
\[
Y_2=3:\ x_2=0 \ \text{with probability }1,
\qquad
Y_2=1,2:\ \text{uniform}.
\]
\end{itemize}

Substituting these statistics into the qubit-Pauli reconstruction formula
for $\hat C_{1\to2}[P]$, with
\[
H_0=\frac{I}{2},
\qquad
H_j=\frac{\sigma_j}{2}
\quad (j=1,2,3),
\]
we obtain
\[
\hat C_{1\to2}[P]
=
I\otimes \frac{I+\sigma_3}{2}
+\frac12\,\sigma_1\otimes\sigma_1
-\frac{2\lambda-1}{2}\,\sigma_2\otimes\sigma_2.
\]

In the computational basis
$\{|00\rangle,|01\rangle,|10\rangle,|11\rangle\}$, this matrix is
\[
\hat C_{1\to2}[P]
=
\begin{pmatrix}
1 & 0 & 0 & \lambda\\
0 & 0 & 1-\lambda & 0\\
0 & 1-\lambda & 1 & 0\\
\lambda & 0 & 0 & 0
\end{pmatrix}.
\]
This matrix decomposes into two $2\times2$ blocks, so its eigenvalues are
\[
\frac{1\pm \sqrt{1+4\lambda^2}}{2},
\qquad
\frac{1\pm \sqrt{1+4(1-\lambda)^2}}{2}.
\]
Since $\lambda\in(0,1)$, we have
\[
\sqrt{1+4\lambda^2}>1,
\qquad
\sqrt{1+4(1-\lambda)^2}>1,
\]
and therefore
\[
\frac{1-\sqrt{1+4\lambda^2}}{2}<0,
\qquad
\frac{1-\sqrt{1+4(1-\lambda)^2}}{2}<0.
\]
Thus $\hat C_{1\to2}[P]$ has negative eigenvalues, and hence
\[
\hat C_{1\to2}[P]\not\ge 0.
\]

We conclude that the strategy of Example~\ref{Ex1} satisfies neither
condition \emph{(C1)} nor the positivity condition
$\hat C_{1\to2}[P]\ge 0$.
\end{proof}

\section{Dimension counting for process-matrix reconstruction in the qubit-Pauli setting}
\Label{app:qubit-counting}

In this appendix, we justify the counting argument used in
Section~\ref{S7-1}.
The point is that, even before imposing any special structure on Charlie's
strategy beyond physical validity, the generalized Born rule depends on the
process matrix only through its class in a certain quotient space.
Thus, exact process-matrix reconstruction requires the accessible tensors
$C[\Gamma_{1,z_1}]\otimes C[\Gamma_{2,z_2}]$ to span that quotient space.
We explain this reduction and evaluate the corresponding dimension in the
two-qubit case.

Let ${\cal B}$ denote the real vector space of Hermitian matrices on
\[
{\cal H}_{I,1}\otimes{\cal H}_{O,1}\otimes
{\cal H}_{I,2}\otimes{\cal H}_{O,2}.
\]
For Alice's and Bob's local operations $\Gamma_{1,z_1}$ and $\Gamma_{2,z_2}$,
the observed joint distribution is given by the generalized Born rule
\begin{equation}\Label{APP-GBR}
P_{Z_1,Z_2}(z_1,z_2)
=
\Tr\!\bigl(
W_S\,
C[\Gamma_{1,z_1}]\otimes C[\Gamma_{2,z_2}]
\bigr),
\end{equation}
where $W_S$ is Charlie's process matrix and
$C[\Gamma]$ is the (unnormalized) Choi matrix of $\Gamma$.

When a strategy $S$ belongs to the most general physical class
${\cal S}_{G}$,
its process matrix $W_S$ satisfies the standard linear
no-signalling-in-time constraints~\cite[(4),(5),(6)]{BMQ}
\begin{align}
(\Tr_{(I,1),(O,1),(O,2)}W_S)\otimes \rho_{mix,(O,2)}
&= \Tr_{(I,1),(O,1)}W_S,
\Label{APP-AB1}\\
(\Tr_{(O,1),(I,2),(O,2)}W_S)\otimes \rho_{mix,(O,1)}
&= \Tr_{(I,2),(O,2)}W_S,
\Label{APP-AB2}\\
W_S
=
(\Tr_{(O,1)}W_S)\otimes \rho_{mix,(O,1)}
+(\Tr_{(O,2)}W_S)\otimes \rho_{mix,(O,2)}
&\notag\\
\qquad
-(\Tr_{(O,1),(O,2)}W_S)\otimes \rho_{mix,(O,1),(O,2)}.
\Label{APP-AB3}
\end{align}
Here, $\rho_{mix}$ denotes the maximally mixed state on the indicated space.

These conditions can be rewritten in the dual form
\begin{align}
\Tr W_S \Bigl[
(\Tr_{(I,1),(O,1),(O,2)}A)\otimes \rho_{mix,(I,1),(O,1),(O,2)}
- \Tr_{(I,1),(O,1)}A \otimes \rho_{mix,(I,1),(O,1)}
\Bigr]
&=0,
\Label{APP-AB1A}\\
\Tr W_S \Bigl[
(\Tr_{(O,1),(I,2),(O,2)}A)\otimes \rho_{mix,(O,1),(I,2),(O,2)}
- \Tr_{(I,2),(O,2)}A \otimes \rho_{mix,(I,2),(O,2)}
\Bigr]
&=0,
\Label{APP-AB2A}\\
\Tr W_S \Bigl[
A
+(\Tr_{(O,1),(O,2)}A)\otimes \rho_{mix,(O,1)}\otimes \rho_{mix,(O,2)}
-(\Tr_{(O,1)}A)\otimes \rho_{mix,(O,1)}
-(\Tr_{(O,2)}A)\otimes \rho_{mix,(O,2)}
\Bigr]
&=0
\Label{APP-AB3A}
\end{align}
for every Hermitian matrix
$A\in{\cal B}$.

Motivated by these identities, define the following subspaces of ${\cal B}$:
\begin{align}
{\cal V}_1
:=
\Bigl\{
\sum_j
\rho_{mix,(I,1),(O,1)}\otimes A_{j,(I,2)}\otimes B_{j,(O,2)}
\;\Big|\;
\Tr B_{j,(O,2)}=0
\Bigr\},
\\
{\cal V}_2
:=
\Bigl\{
\sum_j
A_{j,(I,1)}\otimes B_{j,(O,1)}\otimes \rho_{mix,(I,2),(O,2)}
\;\Big|\;
\Tr B_{j,(O,1)}=0
\Bigr\},
\\
{\cal V}_3
:=
\Bigl\{
\sum_j
A_{j,(I,1)}\otimes B_{j,(O,1)}\otimes
A_{j,(I,2)}\otimes B_{j,(O,2)}
\;\Big|\;
\Tr B_{j,(O,1)}=\Tr B_{j,(O,2)}=0
\Bigr\}.
\end{align}
By construction, \eqref{APP-AB1A}--\eqref{APP-AB3A} imply that
\[
\Tr W_S A=0
\qquad
\forall A\in {\cal V}_1\oplus{\cal V}_2\oplus{\cal V}_3.
\]
Hence, in the generalized Born rule \eqref{APP-GBR},
only the equivalence class of
$C[\Gamma_{1,z_1}]\otimes C[\Gamma_{2,z_2}]$
modulo
${\cal V}_1\oplus{\cal V}_2\oplus{\cal V}_3$
can influence the observed statistics.
Therefore, exact reconstruction of the process matrix requires the accessible
family
\[
\bigl\{
C[\Gamma_{1,z_1}]\otimes C[\Gamma_{2,z_2}]
\bigr\}_{z_1,z_2}
\]
to span the quotient space
\[
{\cal B}/({\cal V}_1\oplus{\cal V}_2\oplus{\cal V}_3).
\]

A direct dimension count gives
\begin{align}
\dim {\cal V}_1
&=
d_{I,2}^2(d_{O,2}^2-1),\\
\dim {\cal V}_2
&=
d_{I,1}^2(d_{O,1}^2-1),\\
\dim {\cal V}_3
&=
d_{I,1}^2 d_{I,2}^2 (d_{O,1}^2-1)(d_{O,2}^2-1),
\end{align}
while
\[
\dim {\cal B}
=
d_{I,1}^2 d_{O,1}^2 d_{I,2}^2 d_{O,2}^2.
\]
Consequently,
\begin{align}
\dim {\cal B}/({\cal V}_1\oplus{\cal V}_2\oplus{\cal V}_3)
&=
d_{I,1}^2 d_{I,2}^2 (d_{O,1}^2+d_{O,2}^2-1)
-d_{I,2}^2(d_{O,2}^2-1)
-d_{I,1}^2(d_{O,1}^2-1).
\Label{APP-quotient-dim}
\end{align}

We now specialize to the two-qubit case
\[
d_{I,1}=d_{O,1}=d_{I,2}=d_{O,2}=2.
\]
Then \eqref{APP-quotient-dim} becomes
\begin{align}
\dim {\cal B}/({\cal V}_1\oplus{\cal V}_2\oplus{\cal V}_3)
=
4\cdot 4\cdot(4+4-1)-4\cdot(4-1)-4\cdot(4-1)
=88.
\Label{APP-quotient-dim-qubit}
\end{align}

This is the minimum number of linearly independent quotient-space directions
that any tomographically complete reconstruction scheme must access in the
two-qubit setting.

In Section~\ref{S7-1}, however, we consider the restricted qubit-Pauli
measurement scheme.
There, for each player $i=1,2$, the pair $(X_i,Y_i)$ takes only seven possible
values.
Hence the family
$\bigl\{
C[\Gamma_{1,z_1}]\otimes C[\Gamma_{2,z_2}]
\bigr\}_{z_1,z_2}
$ contains at most
$7\times 7 = 49$
elements.
Since $49<88$,
these tensors cannot span the quotient space
${\cal B}/({\cal V}_1\oplus{\cal V}_2\oplus{\cal V}_3)$.
Therefore, exact reconstruction of the underlying process matrix is impossible
in the qubit-Pauli setting considered in Section~\ref{S7-1}.

\section*{References}
\bibliography{ref-only}

\begin{thebibliography}{25}%
\makeatletter
\providecommand \@ifxundefined [1]{%
 \@ifx{#1\undefined}
}%
\providecommand \@ifnum [1]{%
 \ifnum #1\expandafter \@firstoftwo
 \else \expandafter \@secondoftwo
 \fi
}%
\providecommand \@ifx [1]{%
 \ifx #1\expandafter \@firstoftwo
 \else \expandafter \@secondoftwo
 \fi
}%
\providecommand \natexlab [1]{#1}%
\providecommand \enquote  [1]{``#1''}%
\providecommand \bibnamefont  [1]{#1}%
\providecommand \bibfnamefont [1]{#1}%
\providecommand \citenamefont [1]{#1}%
\providecommand \href@noop [0]{\@secondoftwo}%
\providecommand \href [0]{\begingroup \@sanitize@url \@href}%
\providecommand \@href[1]{\@@startlink{#1}\@@href}%
\providecommand \@@href[1]{\endgroup#1\@@endlink}%
\providecommand \@sanitize@url [0]{\catcode `\\12\catcode `\$12\catcode
  `\&12\catcode `\#12\catcode `\^12\catcode `\_12\catcode `\%12\relax}%
\providecommand \@@startlink[1]{}%
\providecommand \@@endlink[0]{}%
\providecommand \url  [0]{\begingroup\@sanitize@url \@url }%
\providecommand \@url [1]{\endgroup\@href {#1}{\urlprefix }}%
\providecommand \urlprefix  [0]{URL }%
\providecommand \Eprint [0]{\href }%
\providecommand \doibase [0]{https://doi.org/}%
\providecommand \selectlanguage [0]{\@gobble}%
\providecommand \bibinfo  [0]{\@secondoftwo}%
\providecommand \bibfield  [0]{\@secondoftwo}%
\providecommand \translation [1]{[#1]}%
\providecommand \BibitemOpen [0]{}%
\providecommand \bibitemStop [0]{}%
\providecommand \bibitemNoStop [0]{.\EOS\space}%
\providecommand \EOS [0]{\spacefactor3000\relax}%
\providecommand \BibitemShut  [1]{\csname bibitem#1\endcsname}%
\let\auto@bib@innerbib\@empty
\bibitem [{\citenamefont {Pearl}(2009)}]{Pearl2009}%
  \BibitemOpen
  \bibfield  {author} {\bibinfo {author} {\bibfnamefont {J.}~\bibnamefont
  {Pearl}},\ }\href {https://doi.org/10.1017/CBO9780511803161} {\emph {\bibinfo
  {title} {Causality: Models, Reasoning, and Inference}}},\ \bibinfo {edition}
  {second edition}\ ed.\ (\bibinfo  {publisher} {Cambridge University Press},\
  \bibinfo {year} {2009})\BibitemShut {NoStop}%
\bibitem [{\citenamefont {Peters}\ \emph {et~al.}(2017)\citenamefont {Peters},
  \citenamefont {Janzing},\ and\ \citenamefont {Sch{\"o}lkopf}}]{Peters2017}%
  \BibitemOpen
  \bibfield  {author} {\bibinfo {author} {\bibfnamefont {J.}~\bibnamefont
  {Peters}}, \bibinfo {author} {\bibfnamefont {D.}~\bibnamefont {Janzing}},\
  and\ \bibinfo {author} {\bibfnamefont {B.}~\bibnamefont {Sch{\"o}lkopf}},\
  }\href {https://mitpress.mit.edu/9780262037310/elements-of-causal-inference/}
  {\emph {\bibinfo {title} {Elements of Causal Inference: Foundations and
  Learning Algorithms}}}\ (\bibinfo  {publisher} {MIT Press},\ \bibinfo {year}
  {2017})\BibitemShut {NoStop}%
\bibitem [{\citenamefont {Chiribella}\ \emph {et~al.}(2009)\citenamefont
  {Chiribella}, \citenamefont {D'Ariano},\ and\ \citenamefont
  {Perinotti}}]{Chiribella2009}%
  \BibitemOpen
  \bibfield  {author} {\bibinfo {author} {\bibfnamefont {G.}~\bibnamefont
  {Chiribella}}, \bibinfo {author} {\bibfnamefont {G.~M.}\ \bibnamefont
  {D'Ariano}},\ and\ \bibinfo {author} {\bibfnamefont {P.}~\bibnamefont
  {Perinotti}},\ }\bibfield  {title} {\bibinfo {title} {Theoretical framework
  for quantum networks},\ }\href {https://doi.org/10.1103/PhysRevA.80.022339}
  {\bibfield  {journal} {\bibinfo  {journal} {Physical Review A}\ }\textbf
  {\bibinfo {volume} {80}},\ \bibinfo {pages} {022339} (\bibinfo {year}
  {2009})}\BibitemShut {NoStop}%
\bibitem [{\citenamefont {Oreshkov}\ \emph {et~al.}(2012)\citenamefont
  {Oreshkov}, \citenamefont {Costa},\ and\ \citenamefont {Brukner}}]{OCB}%
  \BibitemOpen
  \bibfield  {author} {\bibinfo {author} {\bibfnamefont {O.}~\bibnamefont
  {Oreshkov}}, \bibinfo {author} {\bibfnamefont {F.}~\bibnamefont {Costa}},\
  and\ \bibinfo {author} {\bibfnamefont {C.}~\bibnamefont {Brukner}},\
  }\bibfield  {title} {\bibinfo {title} {Quantum correlations with no causal
  order},\ }\href {https://doi.org/10.1038/ncomms2076} {\bibfield  {journal}
  {\bibinfo  {journal} {Nature Communications}\ }\textbf {\bibinfo {volume}
  {3}},\ \bibinfo {pages} {1092} (\bibinfo {year} {2012})},\ \Eprint
  {https://arxiv.org/abs/1105.4464} {arXiv:1105.4464 [quant-ph]} \BibitemShut
  {NoStop}%
\bibitem [{\citenamefont {Ara{\'u}jo}\ \emph {et~al.}(2015)\citenamefont
  {Ara{\'u}jo}, \citenamefont {Branciard}, \citenamefont {Costa}, \citenamefont
  {Feix}, \citenamefont {Giarmatzi},\ and\ \citenamefont
  {Brukner}}]{Araujo2015}%
  \BibitemOpen
  \bibfield  {author} {\bibinfo {author} {\bibfnamefont {M.}~\bibnamefont
  {Ara{\'u}jo}}, \bibinfo {author} {\bibfnamefont {C.}~\bibnamefont
  {Branciard}}, \bibinfo {author} {\bibfnamefont {F.}~\bibnamefont {Costa}},
  \bibinfo {author} {\bibfnamefont {A.}~\bibnamefont {Feix}}, \bibinfo {author}
  {\bibfnamefont {C.}~\bibnamefont {Giarmatzi}},\ and\ \bibinfo {author}
  {\bibfnamefont {C.}~\bibnamefont {Brukner}},\ }\bibfield  {title} {\bibinfo
  {title} {Witnessing causal nonseparability},\ }\href
  {https://doi.org/10.1088/1367-2630/17/10/102001} {\bibfield  {journal}
  {\bibinfo  {journal} {New Journal of Physics}\ }\textbf {\bibinfo {volume}
  {17}},\ \bibinfo {pages} {102001} (\bibinfo {year} {2015})}\BibitemShut
  {NoStop}%
\bibitem [{\citenamefont {Bavaresco}\ \emph {et~al.}(2021)\citenamefont
  {Bavaresco}, \citenamefont {Murao},\ and\ \citenamefont {Quintino}}]{BMQ}%
  \BibitemOpen
  \bibfield  {author} {\bibinfo {author} {\bibfnamefont {J.}~\bibnamefont
  {Bavaresco}}, \bibinfo {author} {\bibfnamefont {M.}~\bibnamefont {Murao}},\
  and\ \bibinfo {author} {\bibfnamefont {M.~T.}\ \bibnamefont {Quintino}},\
  }\bibfield  {title} {\bibinfo {title} {Strict hierarchy between parallel,
  sequential, and indefinite-causal-order strategies for channel
  discrimination},\ }\href {https://doi.org/10.1103/PhysRevLett.127.200504}
  {\bibfield  {journal} {\bibinfo  {journal} {Physical Review Letters}\
  }\textbf {\bibinfo {volume} {127}},\ \bibinfo {pages} {200504} (\bibinfo
  {year} {2021})},\ \Eprint {https://arxiv.org/abs/2011.08300}
  {arXiv:2011.08300 [quant-ph]} \BibitemShut {NoStop}%
\bibitem [{\citenamefont {Ebler}\ \emph {et~al.}(2018)\citenamefont {Ebler},
  \citenamefont {Salek},\ and\ \citenamefont {Chiribella}}]{ebler2018enhanced}%
  \BibitemOpen
  \bibfield  {author} {\bibinfo {author} {\bibfnamefont {D.}~\bibnamefont
  {Ebler}}, \bibinfo {author} {\bibfnamefont {S.}~\bibnamefont {Salek}},\ and\
  \bibinfo {author} {\bibfnamefont {G.}~\bibnamefont {Chiribella}},\ }\bibfield
   {title} {\bibinfo {title} {Enhanced communication with the assistance of
  indefinite causal order},\ }\href
  {https://doi.org/10.1103/PhysRevLett.120.120502} {\bibfield  {journal}
  {\bibinfo  {journal} {Physical Review Letters}\ }\textbf {\bibinfo {volume}
  {120}},\ \bibinfo {pages} {120502} (\bibinfo {year} {2018})},\ \Eprint
  {https://arxiv.org/abs/1711.10165} {arXiv:1711.10165 [quant-ph]} \BibitemShut
  {NoStop}%
\bibitem [{\citenamefont {Bavaresco}\ \emph {et~al.}(2022)\citenamefont
  {Bavaresco}, \citenamefont {Murao},\ and\ \citenamefont
  {Quintino}}]{bavaresco2022unitary}%
  \BibitemOpen
  \bibfield  {author} {\bibinfo {author} {\bibfnamefont {J.}~\bibnamefont
  {Bavaresco}}, \bibinfo {author} {\bibfnamefont {M.}~\bibnamefont {Murao}},\
  and\ \bibinfo {author} {\bibfnamefont {M.~T.}\ \bibnamefont {Quintino}},\
  }\bibfield  {title} {\bibinfo {title} {Unitary channel discrimination beyond
  group structures: Advantages of sequential and indefinite-causal-order
  strategies},\ }\href {https://doi.org/10.1063/5.0075919} {\bibfield
  {journal} {\bibinfo  {journal} {Journal of Mathematical Physics}\ }\textbf
  {\bibinfo {volume} {63}},\ \bibinfo {pages} {042203} (\bibinfo {year}
  {2022})},\ \Eprint {https://arxiv.org/abs/2105.13369} {arXiv:2105.13369
  [quant-ph]} \BibitemShut {NoStop}%
\bibitem [{\citenamefont {Zhao}\ \emph {et~al.}(2020)\citenamefont {Zhao},
  \citenamefont {Yang},\ and\ \citenamefont {Chiribella}}]{zhao2020quantum}%
  \BibitemOpen
  \bibfield  {author} {\bibinfo {author} {\bibfnamefont {X.}~\bibnamefont
  {Zhao}}, \bibinfo {author} {\bibfnamefont {Y.}~\bibnamefont {Yang}},\ and\
  \bibinfo {author} {\bibfnamefont {G.}~\bibnamefont {Chiribella}},\ }\bibfield
   {title} {\bibinfo {title} {Quantum metrology with indefinite causal order},\
  }\href {https://doi.org/10.1103/PhysRevLett.124.190503} {\bibfield  {journal}
  {\bibinfo  {journal} {Physical Review Letters}\ }\textbf {\bibinfo {volume}
  {124}},\ \bibinfo {pages} {190503} (\bibinfo {year} {2020})},\ \Eprint
  {https://arxiv.org/abs/1912.02449} {arXiv:1912.02449 [quant-ph]} \BibitemShut
  {NoStop}%
\bibitem [{\citenamefont {Felce}\ and\ \citenamefont
  {Vedral}(2020)}]{felce2020quantum}%
  \BibitemOpen
  \bibfield  {author} {\bibinfo {author} {\bibfnamefont {D.}~\bibnamefont
  {Felce}}\ and\ \bibinfo {author} {\bibfnamefont {V.}~\bibnamefont {Vedral}},\
  }\bibfield  {title} {\bibinfo {title} {Quantum refrigeration with indefinite
  causal order},\ }\href {https://doi.org/10.1103/PhysRevLett.125.070603}
  {\bibfield  {journal} {\bibinfo  {journal} {Physical Review Letters}\
  }\textbf {\bibinfo {volume} {125}},\ \bibinfo {pages} {070603} (\bibinfo
  {year} {2020})},\ \Eprint {https://arxiv.org/abs/2003.00794}
  {arXiv:2003.00794 [quant-ph]} \BibitemShut {NoStop}%
\bibitem [{\citenamefont {Hayashi}\ \emph {et~al.}(2026)\citenamefont
  {Hayashi}, \citenamefont {Cao}, \citenamefont {Yu},\ and\ \citenamefont
  {Zhao}}]{CharFreeIdentification}%
  \BibitemOpen
  \bibfield  {author} {\bibinfo {author} {\bibfnamefont {M.}~\bibnamefont
  {Hayashi}}, \bibinfo {author} {\bibfnamefont {L.}~\bibnamefont {Cao}},
  \bibinfo {author} {\bibfnamefont {B.}~\bibnamefont {Yu}},\ and\ \bibinfo
  {author} {\bibfnamefont {Y.-Y.}\ \bibnamefont {Zhao}},\ }\href
  {https://doi.org/10.48550/arXiv.2602.20997} {\bibinfo {title}
  {Characterization-free classification and identification of the environment
  between two quantum players}} (\bibinfo {year} {2026}),\ \bibinfo {note}
  {arXiv:2602.20997v1, posted 24 Feb 2026},\ \Eprint
  {https://arxiv.org/abs/2602.20997} {arXiv:2602.20997 [quant-ph]} \BibitemShut
  {NoStop}%
\bibitem [{\citenamefont {Liu}\ \emph {et~al.}(2025)\citenamefont {Liu},
  \citenamefont {Qiu}, \citenamefont {Dahlsten},\ and\ \citenamefont
  {Vedral}}]{Liu3}%
  \BibitemOpen
  \bibfield  {author} {\bibinfo {author} {\bibfnamefont {X.}~\bibnamefont
  {Liu}}, \bibinfo {author} {\bibfnamefont {Y.}~\bibnamefont {Qiu}}, \bibinfo
  {author} {\bibfnamefont {O.}~\bibnamefont {Dahlsten}},\ and\ \bibinfo
  {author} {\bibfnamefont {V.}~\bibnamefont {Vedral}},\ }\bibfield  {title}
  {\bibinfo {title} {Quantum causal inference with extremely light touch},\
  }\href {https://doi.org/10.1038/s41534-024-00956-0} {\bibfield  {journal}
  {\bibinfo  {journal} {npj Quantum Information}\ }\textbf {\bibinfo {volume}
  {11}},\ \bibinfo {pages} {1} (\bibinfo {year} {2025})}\BibitemShut {NoStop}%
\bibitem [{\citenamefont {Liu}\ \emph {et~al.}(2024{\natexlab{a}})\citenamefont
  {Liu}, \citenamefont {Chen},\ and\ \citenamefont {Dahlsten}}]{Liu2}%
  \BibitemOpen
  \bibfield  {author} {\bibinfo {author} {\bibfnamefont {X.}~\bibnamefont
  {Liu}}, \bibinfo {author} {\bibfnamefont {Q.}~\bibnamefont {Chen}},\ and\
  \bibinfo {author} {\bibfnamefont {O.}~\bibnamefont {Dahlsten}},\ }\bibfield
  {title} {\bibinfo {title} {Inferring the arrow of time in quantum
  spatiotemporal correlations},\ }\href
  {https://doi.org/10.1103/PhysRevA.109.032219} {\bibfield  {journal} {\bibinfo
   {journal} {Physical Review A}\ }\textbf {\bibinfo {volume} {109}},\ \bibinfo
  {pages} {032219} (\bibinfo {year} {2024}{\natexlab{a}})}\BibitemShut
  {NoStop}%
\bibitem [{\citenamefont {Liu}(2024)}]{PhD}%
  \BibitemOpen
  \bibfield  {author} {\bibinfo {author} {\bibfnamefont {X.}~\bibnamefont
  {Liu}},\ }\emph {\bibinfo {title} {Temporal Correlations in Quantum
  Information and Thermodynamics}},\ \href
  {https://orcid.org/0000-0002-2528-2679} {\bibinfo {type} {Ph.d.
  dissertation}},\ \bibinfo  {school} {Southern University of Science and
  Technology} (\bibinfo {year} {2024})\BibitemShut {NoStop}%
\bibitem [{\citenamefont {Liu}\ \emph {et~al.}(2026)\citenamefont {Liu},
  \citenamefont {Liu}, \citenamefont {Chen}, \citenamefont {Qiu}, \citenamefont
  {Vedral}, \citenamefont {Nie}, \citenamefont {Dahlsten},\ and\ \citenamefont
  {Lu}}]{LiuExperimental2026}%
  \BibitemOpen
  \bibfield  {author} {\bibinfo {author} {\bibfnamefont {H.}~\bibnamefont
  {Liu}}, \bibinfo {author} {\bibfnamefont {X.}~\bibnamefont {Liu}}, \bibinfo
  {author} {\bibfnamefont {Q.}~\bibnamefont {Chen}}, \bibinfo {author}
  {\bibfnamefont {Y.}~\bibnamefont {Qiu}}, \bibinfo {author} {\bibfnamefont
  {V.}~\bibnamefont {Vedral}}, \bibinfo {author} {\bibfnamefont
  {X.}~\bibnamefont {Nie}}, \bibinfo {author} {\bibfnamefont {O.}~\bibnamefont
  {Dahlsten}},\ and\ \bibinfo {author} {\bibfnamefont {D.}~\bibnamefont {Lu}},\
  }\bibfield  {title} {\bibinfo {title} {Experimental causal inference from
  measurements alone in the presence of spatial and temporal quantum
  correlations},\ }\bibfield  {journal} {\bibinfo  {journal} {npj Quantum
  Information}\ }\href {https://doi.org/10.1038/s41534-026-01337-5}
  {10.1038/s41534-026-01337-5} (\bibinfo {year} {2026})\BibitemShut {NoStop}%
\bibitem [{\citenamefont {Fullwood}\ and\ \citenamefont
  {Parzygnat}(2022)}]{FullwoodParzygnat2022StatesOverTime}%
  \BibitemOpen
  \bibfield  {author} {\bibinfo {author} {\bibfnamefont {J.}~\bibnamefont
  {Fullwood}}\ and\ \bibinfo {author} {\bibfnamefont {A.~J.}\ \bibnamefont
  {Parzygnat}},\ }\bibfield  {title} {\bibinfo {title} {On quantum states over
  time},\ }\href {https://doi.org/10.1098/rspa.2022.0104} {\bibfield  {journal}
  {\bibinfo  {journal} {Proceedings of the Royal Society A: Mathematical,
  Physical and Engineering Sciences}\ }\textbf {\bibinfo {volume} {478}},\
  \bibinfo {pages} {20220104} (\bibinfo {year} {2022})},\ \Eprint
  {https://arxiv.org/abs/2202.03607} {arXiv:2202.03607 [quant-ph]} \BibitemShut
  {NoStop}%
\bibitem [{\citenamefont {Parzygnat}\ and\ \citenamefont
  {Fullwood}(2023)}]{ParzygnatFullwood2023QuantumBayes}%
  \BibitemOpen
  \bibfield  {author} {\bibinfo {author} {\bibfnamefont {A.~J.}\ \bibnamefont
  {Parzygnat}}\ and\ \bibinfo {author} {\bibfnamefont {J.}~\bibnamefont
  {Fullwood}},\ }\bibfield  {title} {\bibinfo {title} {From time-reversal
  symmetry to quantum bayes' rules},\ }\href
  {https://doi.org/10.1103/PRXQuantum.4.020334} {\bibfield  {journal} {\bibinfo
   {journal} {PRX Quantum}\ }\textbf {\bibinfo {volume} {4}},\ \bibinfo {pages}
  {020334} (\bibinfo {year} {2023})},\ \Eprint
  {https://arxiv.org/abs/2212.08088} {arXiv:2212.08088 [quant-ph]} \BibitemShut
  {NoStop}%
\bibitem [{\citenamefont {Fullwood}\ and\ \citenamefont
  {Parzygnat}(2024)}]{FullwoodParzygnat2024OperatorRepresentation}%
  \BibitemOpen
  \bibfield  {author} {\bibinfo {author} {\bibfnamefont {J.}~\bibnamefont
  {Fullwood}}\ and\ \bibinfo {author} {\bibfnamefont {A.~J.}\ \bibnamefont
  {Parzygnat}},\ }\href {https://doi.org/10.48550/arXiv.2405.17555} {\bibinfo
  {title} {Operator representation of spatiotemporal quantum correlations}}
  (\bibinfo {year} {2024}),\ \bibinfo {note} {arXiv:2405.17555v2, revised 11
  Apr 2025},\ \Eprint {https://arxiv.org/abs/2405.17555} {arXiv:2405.17555
  [quant-ph]} \BibitemShut {NoStop}%
\bibitem [{\citenamefont {Song}\ \emph {et~al.}(2024)\citenamefont {Song},
  \citenamefont {Narasimhachar}, \citenamefont {Regula}, \citenamefont
  {Elliott},\ and\ \citenamefont
  {Gu}}]{SongNarasimhacharRegulaElliottGu2024CausalClassification}%
  \BibitemOpen
  \bibfield  {author} {\bibinfo {author} {\bibfnamefont {M.}~\bibnamefont
  {Song}}, \bibinfo {author} {\bibfnamefont {V.}~\bibnamefont {Narasimhachar}},
  \bibinfo {author} {\bibfnamefont {B.}~\bibnamefont {Regula}}, \bibinfo
  {author} {\bibfnamefont {T.~J.}\ \bibnamefont {Elliott}},\ and\ \bibinfo
  {author} {\bibfnamefont {M.}~\bibnamefont {Gu}},\ }\bibfield  {title}
  {\bibinfo {title} {Causal classification of spatiotemporal quantum
  correlations},\ }\href {https://doi.org/10.1103/PhysRevLett.133.110202}
  {\bibfield  {journal} {\bibinfo  {journal} {Physical Review Letters}\
  }\textbf {\bibinfo {volume} {133}},\ \bibinfo {pages} {110202} (\bibinfo
  {year} {2024})}\BibitemShut {NoStop}%
\bibitem [{\citenamefont {Song}\ and\ \citenamefont
  {Parzygnat}(2025)}]{SongParzygnat2025CausalExplanation}%
  \BibitemOpen
  \bibfield  {author} {\bibinfo {author} {\bibfnamefont {M.}~\bibnamefont
  {Song}}\ and\ \bibinfo {author} {\bibfnamefont {A.~J.}\ \bibnamefont
  {Parzygnat}},\ }\bibfield  {title} {\bibinfo {title} {Bipartite quantum
  states admitting a causal explanation},\ }\href
  {https://doi.org/10.1116/5.0303300} {\bibfield  {journal} {\bibinfo
  {journal} {AVS Quantum Science}\ }\textbf {\bibinfo {volume} {7}},\ \bibinfo
  {pages} {045002} (\bibinfo {year} {2025})},\ \Eprint
  {https://arxiv.org/abs/2507.14278} {arXiv:2507.14278 [quant-ph]} \BibitemShut
  {NoStop}%
\bibitem [{\citenamefont {Choi}(1975)}]{Choi1975}%
  \BibitemOpen
  \bibfield  {author} {\bibinfo {author} {\bibfnamefont {M.-D.}\ \bibnamefont
  {Choi}},\ }\bibfield  {title} {\bibinfo {title} {Completely positive linear
  maps on complex matrices},\ }\href
  {https://doi.org/10.1016/0024-3795(75)90075-0} {\bibfield  {journal}
  {\bibinfo  {journal} {Linear Algebra and its Applications}\ }\textbf
  {\bibinfo {volume} {10}},\ \bibinfo {pages} {285} (\bibinfo {year}
  {1975})}\BibitemShut {NoStop}%
\bibitem [{\citenamefont {Jamio{\l}kowski}(1972)}]{Jamiolkowski1972}%
  \BibitemOpen
  \bibfield  {author} {\bibinfo {author} {\bibfnamefont {A.}~\bibnamefont
  {Jamio{\l}kowski}},\ }\bibfield  {title} {\bibinfo {title} {Linear
  transformations which preserve trace and positive semidefiniteness of
  operators},\ }\href {https://doi.org/10.1016/0034-4877(72)90011-0} {\bibfield
   {journal} {\bibinfo  {journal} {Reports on Mathematical Physics}\ }\textbf
  {\bibinfo {volume} {3}},\ \bibinfo {pages} {275} (\bibinfo {year}
  {1972})}\BibitemShut {NoStop}%
\bibitem [{\citenamefont {Fitzsimons}\ \emph {et~al.}(2016)\citenamefont
  {Fitzsimons}, \citenamefont {Jones},\ and\ \citenamefont {Vedral}}]{FJV}%
  \BibitemOpen
  \bibfield  {author} {\bibinfo {author} {\bibfnamefont {J.~F.}\ \bibnamefont
  {Fitzsimons}}, \bibinfo {author} {\bibfnamefont {J.~A.}\ \bibnamefont
  {Jones}},\ and\ \bibinfo {author} {\bibfnamefont {V.}~\bibnamefont
  {Vedral}},\ }\bibfield  {title} {\bibinfo {title} {Quantum correlations which
  imply causation},\ }\href {https://doi.org/10.1038/srep18281} {\bibfield
  {journal} {\bibinfo  {journal} {Scientific Reports}\ }\textbf {\bibinfo
  {volume} {5}},\ \bibinfo {pages} {18281} (\bibinfo {year}
  {2016})}\BibitemShut {NoStop}%
\bibitem [{\citenamefont {Liu}\ \emph {et~al.}(2024{\natexlab{b}})\citenamefont
  {Liu}, \citenamefont {Jia}, \citenamefont {Qiu}, \citenamefont {Li},\ and\
  \citenamefont {Dahlsten}}]{Liu}%
  \BibitemOpen
  \bibfield  {author} {\bibinfo {author} {\bibfnamefont {X.}~\bibnamefont
  {Liu}}, \bibinfo {author} {\bibfnamefont {Z.}~\bibnamefont {Jia}}, \bibinfo
  {author} {\bibfnamefont {Y.}~\bibnamefont {Qiu}}, \bibinfo {author}
  {\bibfnamefont {F.}~\bibnamefont {Li}},\ and\ \bibinfo {author}
  {\bibfnamefont {O.}~\bibnamefont {Dahlsten}},\ }\bibfield  {title} {\bibinfo
  {title} {Unification of spatiotemporal quantum formalisms: mapping between
  process and pseudo-density matrices via multiple-time states},\ }\href
  {https://doi.org/10.1088/1367-2630/ad264c} {\bibfield  {journal} {\bibinfo
  {journal} {New Journal of Physics}\ }\textbf {\bibinfo {volume} {26}},\
  \bibinfo {pages} {033008} (\bibinfo {year} {2024}{\natexlab{b}})}\BibitemShut
  {NoStop}%
\bibitem [{\citenamefont {Hayashi}(2017)}]{HayashiQIT2017}%
  \BibitemOpen
  \bibfield  {author} {\bibinfo {author} {\bibfnamefont {M.}~\bibnamefont
  {Hayashi}},\ }\href {https://doi.org/10.1007/978-3-662-49725-8} {\emph
  {\bibinfo {title} {Quantum Information Theory: Mathematical Foundation}}},\
  \bibinfo {edition} {second edition}\ ed.,\ Graduate Texts in Physics\
  (\bibinfo  {publisher} {Springer},\ \bibinfo {year} {2017})\ \bibinfo {note}
  {eBook ISBN: 978-3-662-49725-8}\BibitemShut {NoStop}%
\end{thebibliography}%

\end{document}